\documentclass[11pt,a4paper]{article}
\usepackage{graphicx} 
\usepackage{hyperref}

\hypersetup{
    colorlinks = true,
   linktocpage = true,
    citecolor = {blue}
}

\usepackage{algpseudocode}
\usepackage{subcaption}
\usepackage{svg}

\usepackage{latexsym}
\usepackage{amsmath}
\usepackage{amssymb}
\usepackage{mathrsfs}
\usepackage{amsfonts}
\usepackage{eucal}
\usepackage{amsthm}
\usepackage{graphicx}
\usepackage{color}
\usepackage{youngtab}
\usepackage{slashed}
\usepackage{tikz-cd}

\newcommand {\CalN} {\mathcal N}

\newcommand {\bE}   {\mathbf{E}}

\newcommand {\ve}{\varepsilon}

\newcommand{\beq}{\begin{equation}}
\newcommand{\eeq}{\end{equation}}

\newtheorem{theorem}{Theorem}[section]
\newtheorem{lemma}{Lemma}[section]
\newtheorem{corollary}{Corollary}[section]
\newtheorem{rem}{Remark}

\newtheorem{definition}{Definition}

\title{Classical elliptic integrable systems from the moduli space of instantons}
\author{Andrei Grekov\\
Yang Institute for Theoretical Physics\\ 
Stony Brook University, Stony Brook NY 11794-3636, USA}
\date{}

\begin{document}

\maketitle

\begin{abstract}
This paper is intended to serve as a review of a series of papers with Nikita Nekrasov, where we achieved several important results concerning the relation between the moduli space of instantons and classical integrable systems.
We derive I. Krichever's Lax matrix for the elliptic Calogero-Moser system from the equivariant cohomology of the moduli space of instantons. This result also has K-theoretic and elliptic cohomology counterparts. Our methods rely upon the so-called $\theta$-transform of the $qq$-characters vev’s, defined as integrals of certain classes in these cohomology theories. The key step is the non-commutative Jacobi-like product formula for them.  We also obtained a natural answer for the eigenvector of the Lax matrix and the horizontal section for the associated isomonodromic connection in terms of the partition function of folded instantons. As an application of our formula, we demonstrate some progress towards the spectral duality of the many-body systems in question, as well as give a new look at the quantum-classical duality between their trigonometric version and the corresponding spin chains.

\end{abstract}

\tableofcontents 

\section{Introduction}
The interplay between instanton counting \cite{Nekrasov:2002qd} (and supersymmetric field theories in general) and integrable systems of various kinds is a rich and fascinating topic. In this paper, we will only focus on the elliptic Calogero-Moser system and its generalizations.
It is a system of $N$ particles on the elliptic curve $\Sigma_\tau = \mathbb{C}^\times / \mathbb{Z}$ with coordinates $x_i$ and momenta $p_i$, canonically conjugated under the Poisson bracket
\begin{equation}
    \{p_i, x_j\} = \delta_{ij} x_j
\end{equation}
and the Hamiltonian:
\begin{equation} \label{CalHam}
   H_2^{ellCM} = \frac{1}{2} \sum_{i =1}^N p_i^2 + m^2 \sum_{i \leq j} \wp \big(x_i / x_j\big)
\end{equation}
where $m$ is a coupling constant, and
$\wp$ is a Weierstrass elliptic function, which has a second-order pole on the elliptic curve. So,
\begin{equation}
    \wp(\mathfrak{q}z) = \wp(z)
\end{equation}
\begin{equation}
    \mathfrak{q} = e^{2 \pi i \tau}
\end{equation}
This model first appeared in the context of gauge theory in \cite{GN}. It made its appearance in supersymmetric gauge theory, first implicitly in \cite{DW}, then explicitly in \cite{Martinec:1995qn,DP}, where its spectral curve served as a Seiberg-Witten curve governing the effective action of the ${\CalN}=2^*$ $SU(N)$ Yang-Mills theory on $\mathbb{R}^4 = \mathbb{C}^2$. This relation was extended to the quantum case in the paper \cite{NekSh}. The idea was to restore one of the equivariant parameters of $\mathbb{C}^2$ for it to become a Plank's constant of the integrable model. But only in \cite{BPSCFT5} the wavefunction of the system was matched with the partition function of the supersymmetric theory in the presence of the surface defect directly. Mathematically it is a generating function of the equivariant Chern polynomial integrals over the moduli spaces of instantons on the orbifold of $\mathbb{C}^2 = \mathbb{R}^4$ by the $\mathbb{Z}_N$ action $(z_1, z_2) \mapsto (z_1, e^{\frac{2 \pi i}{N}} z_2)$ a.k.a Affine Laumon space.\\
Despite all that, the Lax matrix governing its classical dynamics through the Lax equation:
\begin{equation}
    \frac{d L}{dt} = [L,M]\,,
\end{equation}
(for some  operator $M$),
which was first found by Krichever in his paper \cite{K} (and later rederived in \cite{GN} )
\begin{equation}
    L_{ij}(z) = p_i \delta_{ij} +m(1- \delta_{ij}) \frac{\theta_{\mathfrak{q}}'(1) \theta_{\mathfrak{q}}(z x_i /x_j)}{\theta_{\mathfrak{q}}(z ) \theta_{\mathfrak{q}}( x_i /x_j)},\,\,\,\, i,j =1,...,N
\end{equation}
 still had no clear 4d geometric meaning. In this text, we are trying to become a step closer to filling in this gap. Our analysis is similar to the one in \cite{NL}, \cite{NLJ}, \cite{NLJ2}, \cite{NLJ3} , \cite{LJ4}, and partially \cite{NPS} .\\
 It is also interesting to notice, that the system (\ref{CalHam}) has a relativistic generalization, called the elliptic Ruijsenaars-Schneider model \cite{RS}, \cite{R}:
 \begin{equation}
     H^{ellRS} = \sum_{i =1}^N \prod_{j \neq i} \frac{\theta_{\mathfrak{q}}(e^{\beta m} x_i /x_j)}{\theta_{\mathfrak{q}}( x_i /x_j)} e^{\beta p_i}
 \end{equation}
 Now the dependence on momenta is trigonometric.
Geometrically, in the treatment above it corresponds to the replacement of equivariant cohomology of the moduli space of instantons with equivariant K-theory. We worked this case out too. \\
And finally, there is a version of the above systems, for which the dependence of Hamiltonians both on coordinates and momenta are elliptic with modular parameters $\mathfrak{q}$ and $p_{6d}$. This system, proposed in \cite{FGNR} is quite mysterious, although it has been studied extensively in \cite{BH}, \cite{MM}, with the most explicit presentation was given in \cite{KS}. Here is how it goes. Introduce the generating function first:
\begin{equation}
    \mathcal{O}(u) = \sum_{n \in \mathbb{Z}} (-u)^n \mathcal{O}_n =  \sum_{n_1,...,n_N \in\, \mathbb{Z}} p_{6d}^{\sum_i\frac{n_i^2 - n_i}{2}} (-u)^{\sum_i n_i} \prod_{i < j}^N \frac{\theta_\mathfrak{q} (e^{ \beta m (n_i - n_j)}\frac{x_i}{x_j})}{\theta_\mathfrak{q} (\frac{x_i}{x_j})} \prod_i^N e^{ \beta n_i p_i} 
\end{equation}
The conjecture is that the ratios:
\begin{equation}
    H_n^{DELL} = \mathcal{O}_0^{-1} \mathcal{O}_n, \,\, \,\, n =1,...,N
\end{equation}
form a commuting set of Hamiltonians.
This case corresponds to elliptic cohomology and we cover it as well, however, this system does not really have a Lax pair representation in the usual finite-dimensional sense. Hopefully, this point will become more clear from the main body of the text.\\
The answers for Lax matrices, which we obtain are built out of certain transforms of specific observables in the gauge theory, which are called the $qq$-characters. If we call the original plane, where our instantons live, $\mathbb{R}^4 = \mathbb{C}^2$ - the $12$ plane, mathematically the insertion of the $qq$-character corresponds to the integral over the complimentary $34$ plane inside $\mathbb{C}^4$ in the crossed instantons moduli space \cite{BPSCFT2}, \cite{BPSCFT3}.\\
It is curious that the eigenvector of the Lax matrix also has a natural interpretation in the full gauge-origami setup - it corresponds to the insertion of the partition function of instantons living in the $24$ plane (folded instantons). \\

\subsection{Organization of the paper}
The paper is organized as follows:
In section 3 we introduce all of our main notions: the moduli spaces of instantons, $qq$-Characters, correlation functions, and their explicit formulas given by equivariant localization.\\
In section 4 we gave a summary of all our main results. Subsections 4.1-4.4 are dedicated to the product formula for the $\theta$-transformed $qq$-character, which will be our main tool. In subsections 4.5 - 4.6, we state our results for the Lax matrices. Subsection 4.8 contains the answer for the horizontal section of the corresponding isomonodromic connection. In subsection 4.9 we take the limit of it to obtain the Lax eigenvector. In subsection 4.7 we show how the quantum-classical duality between trigonometric versions of considered systems and spin chains follows naturally from our product formula. In subsection 4.10 we sketch a way towards the construction of spectral duality for the integrable systems in question, however, this topic requires more work.\\
In sections 5-11 we give proofs to all our statements. Some of these proofs have a alternative shorter version \cite{NG,ncJac}, some of them have generalizations to be published in forthcoming papers \cite{GNhighertimes, GNlax}. \\
The first Appendix contains explicit formulas for $N=2$ particles. And the second one contains the version of the matrix Jacobi identity, which involves only the elliptic curve and no gauge theory, so it is self-contained. 

\subsection{Acknowledgements}
The author is immensely grateful to Nikita Nekrasov. Without Nikta, this paper would not have been possible, as most of the key ideas, on which it is based are due to him. We thank Andrei Okounkov for showing interest in this work and valuable discussions. The author is also grateful to Yan Soibelman for the invitation to his seminar and useful remarks. \\
The research was partly supported by NSF PHY Award 2310279. Any
opinions expressed are solely our own and do not represent the views of the National Science Foundation. 

\section{List of main notations}
$\lambda$ (or $\mu$ or $\nu$) - Young diagram, \\
$\boldsymbol{\Lambda} = (\lambda^{(1)},...,\lambda^{(N)})$ - $N$-tuples of Young diagrams,\\
$\square = (\alpha, a, b)$ - box in position $(a,b)$ in the diagram $\lambda^{(\alpha)}$,\\
$(u_1,..,u_N) = (e^{\beta a_1},..,e^{\beta a_N})$ - equivariant parameters for the $GL(N)$ torus action (\ref{equiv}) ,\\
$(q_1, q_2)$ - equivariant parameters for the rotation torus action in the 12 plane (\ref{equiv}),\\
$(q_3, q_4)$ - equivariant parameters for the rotation torus action in the 34 plane,\\
$q_i = e^{\beta \epsilon_i}$\\
$P_i = 1 - q_i$\\
$\epsilon_1 + \epsilon_2 +\epsilon_3 +\epsilon_4 = 0$, \\
$\epsilon_3 = m$ - mass of the hypermultiplet in the adjoint representation, which is related to the coupling constant in the integrable system,\\
$c_\square = c_{(ab)} = \epsilon_1(a-1)+ \epsilon_2 (b-1)$, \\
$\sigma_\square = \sigma_{(ij)} = \epsilon_3(i-1)+ \epsilon_4 (j-1)$,\\
$\beta$ - the size of the $5d$ circle -parameter controlling the limit to $4d$ case,\\
$p_{6d} = e^{2 \pi i\tau_{6d}}$ - modular parameter of $6d$ torus (elliptic cohomology torus) - parameter controlling the limit to $5d$ case (K-theory limit),\\
We will need the elliptic function and its derivative:
\begin{gather}
    \theta_{\mathfrak{q}}(z) =\sum_{n \in \mathbb{Z}} (-z)^n \mathfrak{q}^{\frac{n^2-n}{2}}  = \prod_{n=0}^\infty (1 - \mathfrak{q}^{n+1}) (1- \mathfrak{q}^n z) (1- \mathfrak{q}^{n+1}z^{-1}) \\
    E_{1}(z) = \frac{z \frac{d}{dz} \theta_{\mathfrak{q}}(z)}{\theta_{\mathfrak{q}}(z)} 
\end{gather} \\
They have the following quasi-periodicity properties:
\begin{gather}
    \theta_{\mathfrak{q}}( \mathfrak{q} z) = - z^{-1} \theta_{\mathfrak{q}}(z)  \\
    E_{1}(\mathfrak{q} z) = E_{1}(z) - 1
\end{gather}
We will use the following notations for some natural sheaves and their characters at the fixed points (\ref{Tbundle}, \ref{FPbundles}, \ref{foldednotations}):
\begin{gather}
    S_{ij} = N_{ij} - P_i P_j K_{ij} , \quad i,j =1,2 \,\,\text{or} \,\, 2,4 \\
    N_{12}|_{\boldsymbol{\Lambda}}: =e^{-\beta x} \,  \mathcal{W}|_{\boldsymbol{\Lambda}} =e^{-\beta x} \sum_{\alpha =1}^N e^{\beta a_\alpha} \\
    K_{12}|_{\boldsymbol{\Lambda}} :=e^{-\beta x} \,  \mathcal{V}|_{\boldsymbol{\Lambda}}  = e^{-\beta x} \sum_{(\alpha, a,b) \in \boldsymbol{\Lambda}} e^{\beta a_\alpha} q_1^{a-1} q_2^{b-1}\\ 
    N_{24}|_\mu = 1 \\
    K_{24}|_\mu = \sum_{(i,j) \in \mu} q_2^{i-1} q_4^{j-1}
\end{gather}
${\bE}$ - plethystic exponential map. 
The operation ${\bE}$ is defined as follows. For a virtual representation $\mathbb{V}$ of torus $T$  with a virtual character written as 
\begin{equation}
    \chi_T(\mathbb{V}) = \sum_{x_i^+ \in W_+} e^{\beta x_i^+} - \sum_{x_i^- \in W_-} e^{\beta x_i^-}
\end{equation}
\begin{equation} \label{plethystic}
    {\bE}[\mathbb{V}] = \frac{\prod_{x_i^- \in W_-} \vartheta(x_i^-)}{\prod_{x_i^+ \in W_+} \vartheta(x_i^+)}
\end{equation}
where
\begin{equation}
    \vartheta(x) = \begin{cases}
        x, \,\, 4d \, \, \text{Case} \,\, -\text{Equivariant cohomology} \\
        1-e^{-\beta x} , \,\, 5d \,\, \text{Case} \,\, - \text{Equivariant K-theory} \\
        \theta_{p_{6d}}(e^{-\beta x}) , \,\, 6d \,\, \text{Case} \,\, -\text{Equivariant elliptic cohomology} \\
    \end{cases}
\end{equation}
$^*$ - on any quantity means the torus weights in the character entering it are reversed.\\
$\mathfrak{q}$ - fugacity parameter (instanton coupling) (\ref{4dAverage}),\\
$\mathcal{Y}(x)$ - $Y$-observable (\ref{Yobs}),\\
$Y(x)$ - NS-limit ($\epsilon_1 = 0$) of its average (\ref{NSlimit}),\\
$\EuScript{X}(x)$ - $\chi$-observable (\ref{chiobs}),\\
$\chi(x)$ - NS-limit of its average (\ref{NSlimit}),\\
$\tilde{\mathcal{Q}}(x)$ - (\ref{Q}), \\
$\mathcal{Y}_{24}(x)$ - (\ref{Y24}),\\
$Y_{24}(x)$ - NS-limit of its average (\ref{Y24NS}),\\
$\EuScript{X}_{24}(x)$ - folded instanton observable (\ref{foldeddef}), \\
$\chi_{24}(x)$ - NS-limit of its average (\ref{foldedNS}).\\
$\omega = 0,...,N-1$ - index labeling orbifolded quantities, usually takes values in $\mathbb{Z}_N$,\\
For any virtual character $\chi_T(\mathbb{V})$ denote by $\chi_T(\mathbb{V})_\omega$ the component which has weight $e^{\frac{2 \pi i \omega}{N}}$ under the orbifold $\mathbb{Z}_N$ action:
\begin{equation}
    \chi_T(\mathbb{V})_\omega = \sum_{ \substack{x_i^+ \in W_+, \\
    \mathbb{Z}_N-\text{weight}(e^{\beta x_i^+}) = e^{\frac{2 \pi i \omega}{N}}} } e^{\beta x_i^+} - \sum_{ \substack{x_i^- \in W_-, \\
    \mathbb{Z}_N-\text{weight}(e^{\beta x_i^-}) = e^{\frac{2 \pi i \omega}{N}}} } e^{\beta x_i^-}
\end{equation}
$\mathfrak{q}_\omega$ - instanton couplings in the presence of the orbifold (\ref{fracinst}),\\
$\mathcal{Y}_\omega(x)$ - $Y$-observable in the presence of the orbifold (\ref{Yobsorb}),\\
$Y_\omega(x)$ - NS-limit of its average (\ref{NSlimitorb}),\\
$\EuScript{X}_\omega(x)$ - $\chi$-observable in the presence of the orbifold (\ref{qqCharOrb}),\\
$\chi_\omega(x)$ - NS-limit of its average (\ref{NSlimit}),\\
$\tilde{\mathcal{Q}}_\omega(x)$ - (\ref{Qorb}),\\
$\mathcal{Y}_{24,\omega}(x)$ - (\ref{Y24orb}),\\
$\EuScript{X}_{24, \omega}(x)$ - folded instanton observable in the presence of the orbifold (\ref{foldeddef}), \\
$\chi_{24, \omega}(x)$ - NS-limit of its average (\ref{foldedNSorb}).\\
\begin{equation}
\mathbb{Q}_\omega^{(24), \,\mu} = \prod_{(i,j) \in \mu} \mathfrak{q}_{\omega + i-j}    
\end{equation}
\begin{equation}
    \mathbb{Q}^\lambda_\omega = \prod_{j = 1}^{\lambda_1} \mathfrak{q}_{\omega+1-j}^{\lambda_j^{t}}
\end{equation}
\begin{gather}
    \mathfrak{q}_\omega = \frac{x_\omega}{x_{\omega-1}}, \qquad \omega = 1,..., N-1 \\
    \mathfrak{q}_0 = \mathfrak{q} \frac{x_0}{x_{N-1}}
\end{gather}
and extended to infinity by quasiperiodicity $x_{\omega+N} = \mathfrak{q} x_\omega$.\\
$p_\omega$'s - are canonically conjugate momenta to $\ln x_\omega$.\\
$X$ - diagonal matrix of $x_\omega$ (\ref{Xmatrix}),\\
$P$ - diagonal matrix of $p_\omega$ (\ref{Pmatrix}),\\
$S_z$ - (\ref{Sz}),\\
$\hat{\chi}(x)$ - diagonal matrix of $\chi_\omega$ (\ref{factnotation}), \\
$\hat{Y}(x)$ - diagonal matrix of $Y_{\omega+1}$ (\ref{factnotation}), \\
$\hat{\slashed{Q}}$ - diagonal matrix of $\mathfrak{q}_\omega$ (\ref{factnotation}), \\
$C$ - cyclic shift matrix - (\ref{cycl}), \\
$C_z$ - cyclic shift matrix with parameter - (\ref{cyclz}).

\section{Main notions}
Implicitly the main players in the paper are the 
 moduli space of crossed and folded instantons.   
Their current definition  - \cite{BPSCFT2} ,\cite{BPSCFT3} involves the real  geometry. However,  
the partition function of the crossed instantons in the Omega background reduces to the partition function of the gauge theory living on a first irreducible component of the cross with the $qq$-Character insertion, representing contribution from the gauge theory living on a second irreducible component. This insertion is already well-defined algebraically in terms of the techniques surrounding the notion of the moduli space of instantons. Let us now give more details on that.\\

\subsection{Moduli space of instantons}
Let $V$ be a $k$-dimensional vector space, and $W$ be an $N$-dimensional vector space.
\begin{definition}
    The space of matrices $B_1, B_2 \in \mathrm{End}(V)$, $I \in \mathrm{Hom}(W,V)$, $J \in \mathrm{Hom}(V,W)$, satisfying the moment map equation:
    \begin{equation}
        \mu: = [B_1,B_2] +IJ = 0
    \end{equation}
    and the stability condition: \\
    ``There exist no $B_1$, $B_2$ invariant subspace $S \subsetneq V$, such that $\mathrm{Im} \, I \subset S$."\\
    modded out by the $GL(V)$ group action:
    \begin{equation}
        g \cdot (B_1, B_2, I, J) = (g B_1 g^{-1}, g B_2 g^{-1}, g I, J g^{-1})
    \end{equation}
    is called the moduli space of instantons on $\mathbb{C}^2$.
\end{definition}
As usual, the data is expressed in the following diagram:
\begin{equation}
    \begin{tikzcd}
        V \arrow[loop right, "B_2"] \arrow[loop left, "B_1"] \arrow[d, "J"] \\
        W \arrow[u, shift left, "I"]
    \end{tikzcd}
\end{equation}
This space will be denoted $\mathcal{M}_{N,k}$. From this definition it comes equipped with a trivial bundle $\mathcal{W}$, with fibers being the vector space $W$, and the tautological bundle $\mathcal{V}$:
\begin{equation} \label{Tbundle}
    \begin{tikzcd}
       \big( \mu^{-1}(0)^{\text{Stable}} \times V \big) / GL(V) \arrow[d] \\
       \mathcal{M}_{N,k}
    \end{tikzcd}
\end{equation}
An equivalent definition which could be found for example, in the H. Nakajima book \cite{N1}:
\begin{definition}
    The space $\mathcal{M}_{N,k}$ is the moduli space of torsion-free sheaves $E$ on $\mathbb{P}^2$ of rank $N$ and $c_2(E) = k$ with the framing at infinity:
    \begin{equation}
        \Phi : E|_{\ell_\infty} \xrightarrow{\sim} \mathcal{O}^{\oplus N}_{\ell_\infty}
    \end{equation}
    up to isomorphism.\\
    Where $\ell_\infty = \{[0 : z_1 : z_2 ] \in \mathbb{P}^2\}\subset \mathbb{P}^2$ - is a line at infinity.    
\end{definition}
Two definitions above are equivalent, see for example \cite{N1}, where alongside with a proof as a byproduct a concrete description of the sheaf $E$ in terms of the quadruple $(B_1, B_2, I, J)$ was given. Namely, consider the sequence of sheaves on $\mathbb{P}^2$:
\begin{equation} \label{Eres}
    \begin{tikzcd}
        V \otimes \mathcal{O}_{\mathbb{P}^2}(-1) \arrow[r, "a"] &
        \begin{matrix}
            V \otimes \mathcal{O}_{\mathbb{P}^2} \\
            \oplus \\
            V \otimes \mathcal{O}_{\mathbb{P}^2} \\
            \oplus \\
            W \otimes \mathcal{O}_{\mathbb{P}^2}
        \end{matrix} 
         \arrow[r, "b"]  & V \otimes \mathcal{O}_{\mathbb{P}^2}(1)
    \end{tikzcd}
\end{equation}
where:
\begin{gather}
    a = \begin{bmatrix}
        z_0 B_1 - z_1 \\
        z_0 B_2 - z_2 \\
        z_0 J
    \end{bmatrix} \\
b = \begin{bmatrix}
    - (z_0 B_2 - z_2) & z_0 B_1 - z_1 & z_0 I
\end{bmatrix}
\end{gather}
If $(B_1, B_2, I, J)$ is stable and satisfy the moment map equation, then $a$ is injective, $b$ - is surjective, and $b a = 0$. The sheaf $E$ is then equal to the middle cohomology of the above monad:
\begin{equation}
    E \cong \mathrm{ker} \, b \, / \, \mathrm{im} \, a
\end{equation}
\textbf{Universal sheaf}. This description then allows us to give a concrete expression for the universal sheaf $\mathcal{U}$ on $\mathcal{M}_{N,k} \times \mathbb{P}^2$. Indeed, consider the following sequence of sheaves on $\mathcal{M}_{N,k} \times \mathbb{P}^2$:
\begin{equation} \label{Unisheaf}
    \begin{tikzcd}
        \mathcal{V} \boxtimes \mathcal{O}_{\mathbb{P}^2}(-1) \arrow[r, "{[a]}"] &
        \begin{matrix}
            \mathcal{V} \boxtimes \mathcal{O}_{\mathbb{P}^2} \\
            \oplus \\
            \mathcal{V} \boxtimes \mathcal{O}_{\mathbb{P}^2} \\
            \oplus \\
            \mathcal{W} \boxtimes \mathcal{O}_{\mathbb{P}^2}
        \end{matrix} 
         \arrow[r, "{[b]}"]  & \mathcal{V} \boxtimes \mathcal{O}_{\mathbb{P}^2}(1)
    \end{tikzcd}
\end{equation}
The universal sheaf $\mathcal{U}$ is then given by the middle cohomology of this sequence:
\begin{equation}
    \mathcal{U} \cong \mathrm{ker} [b] \, / \, \mathrm{im} [a]
\end{equation}
as, being restricted to $\{E\} \times \mathbb{P}^2$ it becomes $E$:
\begin{equation}
    \mathcal{U}|_{\{E\} \times \mathbb{P}^2} \cong E
\end{equation}
\textbf{Torus action}. Maximal torus of $GL(N)$: $T_{N} \subset GL(N)$ acts on $\mathcal{M}_{N,k}$ by changing the trivialization at infinity and $\mathbb{C}^* \times \mathbb{C}^*$ acts by scaling of the base $\mathbb{P}^2$. Let us denote $T = T_{N} \times \mathbb{C}^* \times \mathbb{C}^*$. Its action on quadruples looks like this:
\begin{equation}
    (U; Q_1, Q_2) \cdot (B_1, B_2, I, J) \mapsto  (Q_1 B_1, Q_2 B_2, I U, U^{-1}J)
\end{equation}
The equivariant parameters of $T$ will be denoted as:
\begin{equation} \label{equiv}
   (u_1,...,u_n;q_1,q_2):=\big(e^{\beta a_1},...,e^{\beta a_N}; e^{\beta \epsilon_1}, e^{\beta \epsilon_2} \big)
\end{equation}
\textbf{Fixed points of the torus action}. The fixed points of the torus $T$ action on $\mathcal{M}_{N,k}$ are labeled by $N$-tuples of Young diagrams $\vec{\Lambda}$. Restrictions of the corresponding tautological bundles $K$-theory classes to the fixed points will have the form:
\begin{gather} \label{FPbundles}
    \mathcal{V}|_{\boldsymbol{\Lambda}} = \sum_{(\alpha,a,b) \in \boldsymbol{\Lambda}} u_\alpha q_1^{a-1} q_2^{b-1} \\
    \mathcal{W}|_{\boldsymbol{\Lambda}} = \sum_{\alpha =1 }^N u_\alpha
\end{gather}
\textbf{Correlation functions}. We would like to define the correlation functions for a certain specific class of observables in the supersymmetric $4d / 5d / 6d$ Yang-Mills theory with matter multiplet of mass $m$ in the adjoint representation of the gauge group.\\
For the $4d$ case (Equivariant cohomology case), for every class $\omega$ in the localized equivariant cohomology ring denoted $\bigoplus_{k=0}^\infty \mathrm{H}_T^\bullet(\mathcal{M}_{N,k} )$ we define:
\begin{equation} \label{4dAverage}
    \langle \omega \rangle_{4d} = \frac{1}{Z_{4d}} \sum_{k=0}^\infty \mathfrak{q}^k \int_{\mathcal{M}_{N,k}} \omega \cdot c_m\big(T^*\mathcal{M}_{N,k}\big)
\end{equation}
\begin{equation}
Z_{4d} =  \sum_{k=0}^\infty \mathfrak{q}^k \int_{\mathcal{M}_{N,k}} c_m\big(T^*\mathcal{M}_{N,k}\big)   
\end{equation}
where integral is equivariant with respect to $T$, $\mathfrak{q}$ - is an instanton counting parameter, and for the Chern polynomial we use the following normalization:
\begin{equation}
    c_x(E) = \sum_{i=0}^r c_i(E) x^{r-i}
\end{equation}
for a vector bundle $E$ of rank $r$.\\
For the $5d$ case (Equivariant $K$-theory case), for every equivariant sheaf $\mathcal{F}$, which represents a class in the localized equivariant $K$-theory denoted $\bigoplus_{k=0}^\infty K_T(\mathcal{M}_{N,k})$ define:
\begin{equation}
    \langle \mathcal{F} \rangle_{5d} = \frac{1}{Z_{5d}} \sum_{k=0}^\infty \mathfrak{q}^k \chi_T\big(\mathcal{M}_{N,k} \,, \mathcal{F} \otimes\wedge_{-e^{\beta m}} T^*\mathcal{M}_{N,k}\big)
\end{equation}
\begin{equation}
Z_{5d} =  \sum_{k=0}^\infty \mathfrak{q}^k \chi_T\big(\mathcal{M}_{N,k}\,, \wedge_{-e^{\beta m}} T^*\mathcal{M}_{N,k}\big) 
\end{equation}
where, $\chi_T$ - is an equivariant Euler characteristic of the sheaf (derived equivariant pushforward to a point), and:
\begin{equation}
\wedge_{-e^{\beta x}} \mathcal{E}: = \sum_{n=0}^\infty (-e^{\beta x})^n  \wedge^n \mathcal{E}   
\end{equation}
for any sheaf $ \mathcal{E}$.\\
$\beta$ is a size of a $5d$ circle - the parameter controlling the limit to $4d$.\\
And finally, in the $6d$ case, we could analogously define for every equivariant elliptic cohomology class its average as a generating function of pushforwards to a point of its product with elliptic chern polynomial of the cotangent bundle $T^*\mathcal{M}_{N,k}$.\\
In places where we are not writing the subscript of an average $\langle \cdot \rangle$, this means the equality in question holds for all 3 cases.\\
The averages above could be calculated using the localization formula. Let us try to express the results for all 3 cases as uniformly as possible. Let $\mathcal{E}$ denote one of the classes described above. Then we have:
\begin{equation}
 \langle \mathcal{E} \rangle =  \sum_{\boldsymbol{\Lambda}} \mathcal{E}|_{\boldsymbol{\Lambda}} \,\mu[\boldsymbol{\Lambda}]
\end{equation}
\begin{equation}
    \mu[\boldsymbol{\Lambda}]: = \frac{\mathfrak{q}^{|\boldsymbol{\Lambda}|} }{Z} \, {\bE}\big[(1-e^{\beta m})\, T^*\mathcal{M}_{N,k} |_{\boldsymbol{\Lambda}} \big]  
\end{equation}
The operation ${\bE}$ is defined as follows. For a virtual representation $\mathbb{V}$ of torus $T$  with a virtual character written as 
\begin{equation}
    \chi_T(\mathbb{V}) = \sum_{x_i^+ \in W_+} e^{\beta x_i^+} - \sum_{x_i^- \in W_-} e^{\beta x_i^-}
\end{equation}
\begin{equation}
    {\bE}[\mathbb{V}] = \frac{\prod_{x_i^- \in W_-} \vartheta(x_i^-)}{\prod_{x_i^+ \in W_+} \vartheta(x_i^+)}
\end{equation}
where
\begin{equation}
    \vartheta(x) = \begin{cases}
        x, \,\, 4d \, \, \text{Case} \,\, -\text{Equivariant cohomology} \\
        1-e^{-\beta x} , \,\, 5d \,\, \text{Case} \,\, - \text{Equivariant K-theory} \\
        \theta_{p_{6d}}(e^{-\beta x}) , \,\, 6d \,\, \text{Case} \,\, -\text{Equivariant elliptic cohomology} \\
    \end{cases}
\end{equation}
And the $K$-theory class of the cotangent bundle appearing in the formula is equal to:
\begin{equation}
T^*\mathcal{M}_{N,k} = \mathcal{W}\mathcal{V}^* + q_1 q_2 \mathcal{V} \mathcal{W}^* - (1-q_1)(1-q_2) \mathcal{V} \mathcal{V}^*  
\end{equation}

\textbf{$Y$ - observable}. One of the most important observables, which we will use, is the characteristic polynomial of a scalar field (see formulas (5.8) and (5.11) in \cite{BPSCFT1}). Mathematically it is defined as follows:
\begin{equation} \label{Yobs}
    \mathcal{Y}(x): = \begin{cases}
        c_x(R \pi_*(\mathcal{U} \otimes  p^* \mathcal{O}_0)) \,\, \text{in the $4d$ case} \\
        \wedge_{-e^{-\beta x}} R \pi_*(\mathcal{U} \otimes  p^* \mathcal{O}_0) \,\, \text{in the $5d$ case} \\
        c_x^{Ell}(R \pi_*(\mathcal{U} \otimes  p^* \mathcal{O}_0)) \,\, \text{in the $6d$ case}
    \end{cases}
\end{equation}
where $\pi$ and $p$ are projections from the product $\mathcal{M}_{N,k} \times \mathbb{P}^2$ to $\mathcal{M}_{N,k}$ and $\mathbb{P}^2$ correspondingly.
\begin{equation}
    \begin{tikzcd}[column sep=small]
& \mathcal{M}_{N,k} \times \mathbb{P}^2 \arrow[dl, "\pi"] \arrow[dr, "p"] & \\
\mathcal{M}_{N,k}  & & \mathbb{P}^2
\end{tikzcd}
\end{equation}
and $\mathcal{O}_0$ is the skyscraper sheaf of $0 \in \mathbb{P}^2$.\\
From (\ref{Unisheaf}) one can easily get the expression for the restriction of $\mathcal{Y}(x)$ to the fixed point $\boldsymbol{\Lambda}$:
\begin{equation}
  \mathcal{Y}(x)|_{\boldsymbol{\Lambda} }=  {\bE}\big[-e^{\beta x}\big( \mathcal{W}^* - (1-q_1^{-1}) (1-q_2^{-1}) \mathcal{\mathcal{V}^*}\big)|_{\boldsymbol{\Lambda} }\big]
\end{equation}
or explicitly:
\begin{equation}
\mathcal{Y}(x)|_{\boldsymbol{\Lambda} } =  \prod_{\alpha = 1}^N \vartheta(x-a_\alpha) \prod_{(\alpha,a,b) \in \boldsymbol{\Lambda}} \frac{\vartheta(x-a_\alpha - c_{(ab)} - \epsilon_1)}{\vartheta(x-a_\alpha - c_{(ab)})} \frac{\vartheta(x-a_\alpha - c_{(ab)} - \epsilon_2)}{\vartheta(x-a_\alpha - c_{(ab)} - \epsilon)} 
\end{equation}
where
\begin{equation}
    c_\square:= c_{(ab)} = \epsilon_1(a-1) + \epsilon_2(b-1)  
\end{equation}
\begin{equation}
    \epsilon = \epsilon_1 + \epsilon_2
\end{equation}

\textbf{Remark.} In the notation, which we chose for the orbifold below the $\mathcal{Y}$-function obtained from the product of the orbifolded ones will have the parameter $\epsilon_2$ replaced by $N \epsilon_2$. We will use the same notations for them, as the former one itself will never show up in the main text.

\textbf{$\EuScript{X}$-observable}. From the definition of $\mathcal{Y}(x)$ we see that its expectation value $\langle \mathcal{Y}(x) \rangle$ has poles in the finite region of $x$. However, it is possible to find a certain combination of functions $\mathcal{Y}(x) $ with shifted arguments such that all such poles in its expectation value would cancel out. Let us hence give the following definition:
\begin{definition}
    An observable $\EuScript{X}(x)$ (also denoted as $\EuScript{X}_{34}(x)$) which is a power series in $\mathfrak{q}$ depending on a variable $x$ only through functions $\mathcal{Y}$ at different points and satisfying 2 properties:\\
1)
\begin{equation}
  \EuScript{X}(x) = \mathcal{Y}(x + \epsilon) + O(\mathfrak{q}) 
\end{equation}
2) $\langle \EuScript{X}(x) \rangle$ has no poles for $x \in \mathbb{C}$\\
is called a $qq$-Character.
\end{definition}

The concrete expression for it was discovered in \cite{BPSCFT1}, see formula (8.26) adapted for affine (adjoint matter case) in (7.6) there ((194) and (153) in the preprint version correspondingly).
\begin{equation} \label{chiobs}
    \EuScript{X}(x) = \mathcal{Y}(x + \epsilon) \sum_{\lambda} \mathbb{S}_\lambda \mathfrak{q}^{|\lambda|} \prod_{ \square \in \lambda} \frac{\mathcal{Y}(x + \sigma_\square +m + \epsilon) \mathcal{Y}(x + \sigma_\square -m )}{\mathcal{Y}(x + \sigma_\square + \epsilon) \mathcal{Y}(x + \sigma_\square)}
\end{equation}
where $\lambda$ - is a Young diagram,
\begin{equation}
    \sigma_\square:= \sigma_{(ij)} = \epsilon_3(i-1) + \epsilon_4(j-1) = m(i-j) + \epsilon (1-j) 
\end{equation}
\begin{equation}
    \mathbb{S}_\lambda = \prod_{\square \in \lambda} \mathbb{S}( m h_\square + \epsilon a_\square)
\end{equation}
$h_\square$ and $a_\square$ are hook and arm length, and
\begin{equation}
   \mathbb{S}(x) = \frac{\vartheta(x + \epsilon_1) \vartheta(x + \epsilon_2)}{\vartheta(x)\vartheta(x + \epsilon_1+ \epsilon_2)} 
\end{equation}
\begin{rem}
    $\lambda$ - is a Young diagram labeling the fixed points of the plane rotations torus action on the moduli space of non-commutative $U(1)$-instantons living on the orthogonal irreducible component of the cross (labeled 34 directions) with equivariant parameters ($m$,$-\epsilon_1-\epsilon_2-m$):= ($\epsilon_3$, $\epsilon_4$).
\end{rem} 

\textbf{$\epsilon_1 \rightarrow 0$ or $\epsilon_2 \rightarrow 0$ limit}. In this limit  described in \cite{NO} the sum over $N$-tuples of Young diagrams $\Vec{\Lambda}$ will be reduced to the evaluation at the limit shape $\vec{\Lambda}^\infty$, which is determined by the equations:
\begin{equation} \label{limitshapebulk}
\frac{\mu[\boldsymbol{\Lambda}^\infty + \square]}{\mu[\boldsymbol{\Lambda}^\infty ]} =1 
\end{equation}
For any addable box $\square$.\\
So, let us define :
\begin{gather} \label{NSlimit}
Y(x):= \lim_{\epsilon_1 \rightarrow 0} \langle \mathcal{Y}(x) \rangle =  \mathcal{Y}(x)|_{\boldsymbol{\Lambda}^\infty} \\
\chi(x):=\lim_{\epsilon_1 \rightarrow 0}\ \langle \EuScript{X}(x) \rangle =  \EuScript{X}(x)|_{\boldsymbol{\Lambda}^\infty}
\end{gather}
And hence one gets:
\begin{equation} \label{qq-CharLimitNonOrb}
    \chi(x) = Y(x+\epsilon) \sum_{\lambda}  \mathfrak{q}^{|\lambda|} \prod_{ \square \in \lambda} \frac{Y(x + \sigma_\square +m + \epsilon) Y(x + \sigma_\square -m )}{Y(x + \sigma_\square + \epsilon) Y(x + \sigma_\square)}
\end{equation}

\textbf{Folded instantons observable}.
The moduli space of folded instantons similarly to crossed instantons was defined in \cite{BPSCFT2} using real geometry. 
However, the contribution of instantons living on the 24 plane of the corresponding non-irreducible manifold could be expressed as an insertion of a certain observable into the theory living on 12 plane, which is expressed as a sum over Young diagrams $\mu$, labeling the fixed points in the $24$ directions.
Hence, it could be defined algebraically, as follows. Let us first introduce an observable $\tilde{\mathcal{Q}}(x)$ through the property:
\begin{equation} \label{Q}
    \mathcal{Y}(x) = \frac{\tilde{\mathcal{Q}}(x)}{\tilde{\mathcal{Q}}(x+ \epsilon_2)}
\end{equation}
To define it geometrically, we just need to replace the skyscraper sheaf of a point $\mathcal{O}_0 $ in the definition of $\mathcal{Y}(x)$ (\ref{Yobs}) with the skyscraper sheaf of a line $z_1 = 0$.\\
And then using it, define:
\begin{equation} \label{Y24}
    \mathcal{Y}_{24}(x) = \frac{\tilde{\mathcal{Q}}(x )}{\tilde{\mathcal{Q}}(x + m)}
\end{equation}
$\mathcal{Y}_{24}(x)$ expectation value has an infinite amount of poles in $x$, however, it is possible to find an expression made out of such functions with shifted arguments where all the poles in the finite region in its average will cancel out.
\begin{definition} \label{foldeddef}
    An observable $\EuScript{X}_{24}(x)$, which is a power series in $\mathfrak{q}$ depending on a variable $x$ only through functions $\tilde{\mathcal{Q}}$ at different points and satisfying 2 properties:\\
    1) 
\begin{equation}
    \EuScript{X}_{24}(x) = \mathcal{Y}_{24}(x+ \epsilon) + O(\mathfrak{q})
\end{equation}
   2) $\langle \EuScript{X}_{24}(x) \rangle$ has no poles for $x \in \mathbb{C}$\\
   is called a folded instanton observable.
\end{definition}
We will only need an explicit expression in the limit $\epsilon_1 \rightarrow 0$, which simplifies things a bit:  
\begin{equation} \label{foldedNS}
    \chi_{24}(x): =\lim_{\epsilon_1 \rightarrow 0} \langle \EuScript{X}_{24}(x) \rangle =  \lim_{\epsilon_1 \rightarrow 0} \Big\langle \sum_{\mu} \mathfrak{q}^{|\mu|} {\bE} \Big[ -q_1 q_2 \frac{P_3}{P_2} S_{12}^* S_{24} \Big] \Big|_{\mu} \Big\rangle
\end{equation}
Where:
\begin{gather} \label{foldednotations}
    q_i = e^{\beta \epsilon_i}, \quad i =1,2,3,4 \\
    P_i = 1- q_i , \quad i =1,2,3,4 \\
    S_{ij} = N_{ij} - P_i P_j K_{ij} , \quad i,j =1,2 \,\,\text{or} \,\, 2,4 \\
    N_{12}|_{\boldsymbol{\Lambda}}: =e^{-\beta x} \,  \mathcal{W}|_{\boldsymbol{\Lambda}} =e^{-\beta x} \sum_{\alpha =1}^N e^{\beta a_\alpha} \\
    K_{12}|_{\boldsymbol{\Lambda}} :=e^{-\beta x} \,  \mathcal{V}|_{\boldsymbol{\Lambda}}  = e^{-\beta x} \sum_{(\alpha, a,b) \in \boldsymbol{\Lambda}} e^{\beta a_\alpha} q_1^{a-1} q_2^{b-1}\\ 
    N_{24}|_\mu = 1 \\
    K_{24}|_\mu = \sum_{(i,j) \in \mu} q_2^{i-1} q_4^{j-1}
\end{gather}
For convenience, we will also denote:
\begin{equation} \label{Y24NS}
    Y_{24}(x): = \lim_{\epsilon_1 \rightarrow 0} \langle \mathcal{Y}_{24}(x) \rangle
\end{equation}

\subsection{Affine Laumon space (moduli space of instantons on the orbifold)}
We will define the Affine Laumon space $\mathcal{M}_{N,\boldsymbol{d}}$ (see for example section (6.1) in \cite{BPSCFT4} and section 3 in \cite{Negut2} or section 2 in \cite{FR} for nice exposition) as a fixed point set in $\mathcal{M}_{N,k}$ under the following $\mathbb{Z}_N$ action (on quadruples):
\begin{equation} \label{ZNaction}
    e^{\frac{2 \pi i}{N}} \cdot (B_1,B_2,I,J) \mapsto (B_1, e^{\frac{2 \pi i}{N}}B_2,I \delta, e^{\frac{2 \pi i}{N}} \delta^{-1}J)
\end{equation}
with 
\begin{equation}
    \delta = \text{diag}\big(e^{\frac{2 \pi i}{N}},...,e^{\frac{2 \pi i(N-1)}{N}}, 1 \big)
\end{equation}
This action arises from the corresponding $\mathbb{Z}_N$ action on the base and multiplication of the framing isomorphism by $\delta$.\\
For the point of $\mathcal{M}_{N,k}$ represented by the quadruple $(B_1,B_2,I,J)$ to be fixed under the action (\ref{ZNaction}) there should exist an element $\Omega \in GL(V)$ such that:
\begin{equation} \label{LaumonQuotient}
(B_1, e^{\frac{2 \pi i}{N}}B_2,I \delta, e^{\frac{2 \pi i}{N}} \delta^{-1}J) =      (\Omega B_1\Omega^{-1} ,\Omega B_2 \Omega^{-1}, \Omega I,J \Omega^{-1})
\end{equation}
By freeness of the $GL(V)$ action on the stable locus, $\Omega$ is unique. In particular, this means, that the assignment $e^{\frac{2 \pi i}{N}} \mapsto \Omega$ determines a homomorphism $\mathbb{Z}_N \rightarrow GL(V)$.  The conjugacy class of $\Omega$ does not depend on the choice of representative of $(B_1,B_2,I,J)$. Hence the fixed point set $\mathcal{M}_{N,k}^{\mathbb{Z}_N}$ breaks into connected components, corresponding to different decompositions of $V$ into weight spaces of the $\mathbb{Z}_N$ action by $\Omega$. Let us denote by $V_\omega$ the connected component with weight $e^{\frac{2 \pi i \omega}{N}}$ for $\omega = 1,..., N$, with $d_\omega = \text{dim}(V_\omega)$. So, that we have $\sum_{\omega} d_\omega = k$, and 
\begin{equation}
 \mathcal{M}_{N,k}^{\mathbb{Z}_N} = \bigsqcup_{\vec{d} \in \mathbb{N}^N}^{|\boldsymbol{d}| = k}   \mathcal{M}_{N,\boldsymbol{d}}.
\end{equation}
$\mathcal{M}_{N,\boldsymbol{d}}$ - is called affine Laumon space.\\
If we also decompose the framing space $W$ into the corresponding weight spaces as:
\begin{equation}
    W = \bigoplus_\omega \mathbb{C} w_\omega
\end{equation}
where $w_\omega$ has weight $e^{\frac{2 \pi i \omega}{N}}$, we see from (\ref{LaumonQuotient}) that the maps $B_1, B_2,I,J$ act as follows:
\begin{gather}
   V_\omega \xrightarrow{B_1} V_\omega \\
   V_\omega \xrightarrow{B_2} V_{\omega+1} \\
   \mathbb{C} w_\omega \xrightarrow{I} V_{\omega} \\
   V_{\omega} \xrightarrow{J} \mathbb{C} w_{\omega+1}
\end{gather}
This means the space $\mathcal{M}_{N,\boldsymbol{d}}$ has the following chainsaw quiver description.
\begin{equation}
    \begin{tikzcd}[column sep=small]
        ... \arrow[rr, "B_{\omega-1}"] \arrow[rd, "J_{\omega-1}"] & & V_{\omega-1} \arrow[loop, "A_{\omega-1}"] \arrow[rr, "B_{\omega}"]  \arrow[rd, "J_{\omega}"] &  & V_{\omega} \arrow[loop, "A_{\omega}"] \arrow[rr, "B_{\omega+1}"]  \arrow[rd, "J_{\omega+1}"] & & ... \\
        & \mathbb{C} w_{\omega-1} \arrow[ru, "I_{\omega-1}"] & & \mathbb{C} w_{\omega} \arrow[ru, "I_{\omega}"] & & \mathbb{C} w_{\omega+1} \arrow[ru, "I_{\omega+1}"]   
    \end{tikzcd}
\end{equation}

Fix the vector spaces $V_1,...,V_N$ of dimensions $d_1,...,d_N$ which sum up to $k$, identify $V_N$ with $V_0$. Consider the space of linear maps:
\begin{equation}
    M_{\vec{d}} = \bigoplus_{\omega =1}^N \mathrm{Hom}(V_\omega, V_\omega) \bigoplus_{\omega =1}^N \mathrm{Hom}(V_\omega, V_{\omega+1}) \bigoplus_{\omega =1}^N \mathrm{Hom}(\mathbb{C} w_\omega, V_{\omega}) \bigoplus_{\omega =1}^N \mathrm{Hom}(V_{\omega},\mathbb{C}  w_{\omega+1}) 
\end{equation}
Let us denote the elements of this space as: $(A_\omega, B_\omega, I_\omega, J_\omega)_{1 \leq \omega \leq N}$. Define the moment map:
\begin{equation}
    \nu: M_{\boldsymbol{d}} \rightarrow \bigoplus_{\omega =1}^N \mathrm{Hom}(V_\omega, V_{\omega+1})
\end{equation}
by the formula:
\begin{equation}
    \nu(A_\omega, B_\omega, I_\omega, J_\omega)_{1 \leq \omega \leq N} = ( B_\omega A_\omega - A_\omega B_{\omega-1} + I_\omega J_\omega)_{1 \leq \omega \leq N}
\end{equation}
And denote:
\begin{equation}
    GL(\vec{d}): = GL(V_1) \times...\times GL(V_N)
\end{equation}
which acts naturally on $ M_{\boldsymbol{d}}$. The Affine Laumon space then could be described by the following:
\begin{definition}
    The space
    \begin{equation}
        \mathcal{M}_{N,\boldsymbol{d}} = \nu^{-1}(0)^s / GL(\boldsymbol{d})
    \end{equation}
is called affine Laumon space. Where the stable quadruples are those for which the vector spaces $V_\omega$ are generated by the action of the maps $A$ and $B$ on the images of $I$'s.
\end{definition}
From this definition, the affine Laumon space comes equipped with trivial line bundles $\mathcal{W}_\omega$, whose fibers are $\mathbb{C} w_\omega$, and the tautological bundles $\mathcal{V}_\omega$.\\
The sequence determining the sheaf $E$ (\ref{Eres}) under the orbifold will split into the sum of:
\begin{equation} 
    \begin{tikzcd}
        V_{\omega} \otimes \mathcal{O}_{\mathbb{P}^2}(-1) \arrow[r, "a_\omega"] &
        \begin{matrix}
            V_\omega \otimes \mathcal{O}_{\mathbb{P}^2} \\
            \oplus \\
            V_{\omega+1} \otimes \mathcal{O}_{\mathbb{P}^2} \\
            \oplus \\
            \mathbb{C} w_{\omega+1} \otimes \mathcal{O}_{\mathbb{P}^2}
        \end{matrix} 
         \arrow[r, "b_\omega"]  & V_{\omega+1} \otimes \mathcal{O}_{\mathbb{P}^2}(1)
    \end{tikzcd}
\end{equation}
where:
\begin{gather}
    a_\omega = \begin{bmatrix}
        z_0 A_\omega - z_1 \\
        z_0 B_{\omega} - z_2 \\
        z_0 J_{\omega+1}
    \end{bmatrix} \\
b_\omega = \begin{bmatrix}
    - (z_0 B_{\omega+1} - z_2) & z_0 A_{\omega} - z_1 & z_0 I_{\omega+1}
\end{bmatrix}
\end{gather}
Hence, the affine Laumon space describes the flag of sheaves with quotient sheaves being equal to:
\begin{equation}
    E_\omega \cong \mathrm{ker} \, b_\omega \, / \, \mathrm{im} \, a_\omega
\end{equation}

\textbf{Universal sheaves.} As it follows from the previous discussion, the universal sheaf $\mathcal{U}$ under the orbifold will get replaced by the universal flag of sheaves with quotient sheaves $\mathcal{U}_\omega$ being given by the middle cohomology of the sequence:
\begin{equation} \label{Unisheaves}
    \begin{tikzcd}
        \mathcal{V}_\omega \boxtimes \mathcal{O}_{\mathbb{P}^2}(-1) \arrow[r, "{[a_\omega]}"] &
        \begin{matrix}
            \mathcal{V}_\omega \boxtimes \mathcal{O}_{\mathbb{P}^2} \\
            \oplus \\
            \mathcal{V}_{\omega+1} \boxtimes \mathcal{O}_{\mathbb{P}^2} \\
            \oplus \\
            \mathcal{W}_{\omega+1} \boxtimes \mathcal{O}_{\mathbb{P}^2}
        \end{matrix} 
         \arrow[r, "{[b_\omega]}"]  & \mathcal{V}_{\omega+1} \boxtimes \mathcal{O}_{\mathbb{P}^2}(1)
    \end{tikzcd}
\end{equation}
Namely:
\begin{equation}
    \mathcal{U}_\omega \cong \mathrm{ker} [b_\omega] \, / \, \mathrm{im} [a_\omega]
\end{equation}

\textbf{Torus action.} The action of the torus $T = T_N \times \mathbb{C}^* \times \mathbb{C}^*$ on $\mathcal{M}_{N,k}$ commutes with the action of $\mathbb{Z}_N$ (\ref{ZNaction}), hence its restriction to $\mathcal{M}_{N, \boldsymbol{d}}$ is well defined, and is given by the same formula:
\begin{equation}
    (U; Q_1, Q_2) \cdot (A_\omega, B_\omega, I_\omega, J_\omega)_{1 \leq \omega \leq N}\mapsto  (Q_1 A_\omega, Q_2 B_\omega, I_\omega U_\omega, U_\omega^{-1} J_\omega)_{1 \leq \omega \leq N}
\end{equation}
The equivariant parameters will be denoted by the same letters:
\begin{equation}
   (u_1,...,u_n;q_1,q_2):=\big(e^{\beta a_1},...,e^{\beta a_N}; e^{\beta \epsilon_1}, e^{\beta \epsilon_2} \big)
\end{equation}
\textbf{Fixed points of the torus action.} Fixed points of the torus $T$ action on $\mathcal{M}_{N,k}$ automatically lie inside 
$\mathcal{M}_{N,k}^{\mathbb{Z}_N}$. Hence we conclude that the fixed points in $\mathcal{M}_{N, \boldsymbol{d}}$ are also labeled by $N$-tuples of Young diagrams $\boldsymbol{\Lambda}$. But the restrictions of the $K$-theory classes of tautological bundles to the fixed points now look as:
\begin{gather}
    \mathcal{V}_\omega |_{\boldsymbol{\Lambda}} = \sum_{\substack{ (\alpha,a,b) \in \boldsymbol{\Lambda} \\ \alpha + b -1 = \omega}} u_\alpha q_1^{a-1} q_2^{b-1} \\
    \mathcal{W}_\omega|_{\boldsymbol{\Lambda}} = u_\omega
\end{gather}
\textbf{Correlation functions.} In the presence of the orbifold the instanton coupling is getting fractionalized:
\begin{gather} \label{fracinst}
    \mathfrak{q} \mapsto \mathfrak{q}_0, ..., \mathfrak{q}_{N-1} \\
    \mathfrak{q}_0\cdot ... \cdot \mathfrak{q}_{N-1} = \mathfrak{q}
\end{gather}
Now, in the 4d case for every class $\omega$ in the localized equivariant cohomology denoted $\bigoplus_{\boldsymbol{d}} \mathrm{H}_T^\bullet (\mathcal{M}_{N, \boldsymbol{d}})$ define:
\begin{equation}
    \langle \omega \rangle_{4d}^{\mathbb{Z}_N} = \frac{1}{Z_{4d}^{\mathbb{Z}_N}} \sum_{\boldsymbol{d} \in \mathbb{N}^N}\mathfrak{q}_0^{d_0} \cdot...\cdot \mathfrak{q}_{N-1}^{d_{N-1}} \int_{\mathcal{M}_{N, \boldsymbol{d}}} \omega \cdot c_m\big(T^*\mathcal{M}_{N, \boldsymbol{d}}\big)
\end{equation}
\begin{equation}
Z_{4d}^{\mathbb{Z}_N} =  \sum_{\boldsymbol{d} \in \mathbb{N}^N}\mathfrak{q}_0^{d_0} \cdot...\cdot \mathfrak{q}_{N-1}^{d_{N-1}}\int_{\mathcal{M}_{N, \boldsymbol{d}}} c_m\big(T^*\mathcal{M}_{N, \vec{d}}\big)   
\end{equation}
where integral is equivariant with respect to $T$.\\
Similarly in the 5d case for every equivariant sheaf $\mathcal{F}$ representing a class  in the localized equivariant $K$-theory denoted as $\bigoplus_{\boldsymbol{d}} K_T(\mathcal{M}_{N, \boldsymbol{d}})$ define:
\begin{equation}
    \langle \mathcal{F} \rangle_{5d}^{\mathbb{Z}_N} = \frac{1}{Z_{5d}^{\mathbb{Z}_N}} \sum_{\boldsymbol{d} \in \mathbb{N}^N}\mathfrak{q}_0^{d_0} \cdot...\cdot \mathfrak{q}_{N-1}^{d_{N-1}} \chi_T\big(\mathcal{M}_{N, \boldsymbol{d}} \,, \mathcal{F} \otimes\wedge_{-e^{\beta m}} T^*\mathcal{M}_{N, \boldsymbol{d}}\big)
\end{equation}
\begin{equation}
Z_{5d}^{\mathbb{Z}_N} =  \sum_{\boldsymbol{d} \in \mathbb{N}^N}\mathfrak{q}_0^{d_0} \cdot...\cdot \mathfrak{q}_{N-1}^{d_{N-1}} \chi_T\big(\mathcal{M}_{N, \boldsymbol{d}} \,,\wedge_{-e^{\beta m}} T^*\mathcal{M}_{N, \boldsymbol{d}}\big) 
\end{equation}
where, $\chi_T$ - is an equivariant Euler characteristic of the sheaf (derived pushforward to a point).\\
And finally, as in the non-orbifolded setup in the $6d$ case, we could analogously define for every equivariant elliptic cohomology class its average as a generating function of pushforwards to a point of its product with elliptic chern polynomial of the cotangent bundle $T^*\mathcal{M}_{N, \boldsymbol{d}}$.\\
And once again, in places where we are not writing the subscript of an average $\langle \cdot \rangle^{\mathbb{Z}_N}$, this means the equality in question holds for all 3 cases.\\
The averages above could be calculated using the localization formula. We will try to express the results for all 3 cases as uniformly as possible. Let $\mathcal{E}$ denote one of the classes described above. Then we have:
\begin{equation}
 \langle \mathcal{E} \rangle^{\mathbb{Z}_N} =  \sum_{\boldsymbol{\Lambda}} \mathcal{E}|_{\boldsymbol{\Lambda}} \,\mu^{\mathbb{Z}_N}[\boldsymbol{\Lambda}]
\end{equation}
\begin{equation}
    \mu^{\mathbb{Z}_N}[\boldsymbol{\Lambda}]: = \frac{\mathfrak{q}_0^{d_0(\boldsymbol{\Lambda})} \cdot...\cdot \mathfrak{q}_{N-1}^{d_{N-1}(\boldsymbol{\Lambda})} }{Z^{\mathbb{Z}_N}} \, {\bE}\big[(1-e^{\beta m})\, T^*\mathcal{M}_{N, \boldsymbol{d}}\,|_{\boldsymbol{\Lambda}} \big]   
\end{equation}
where:
\begin{equation}
    d_\omega(\boldsymbol{\Lambda}): = \sum_{\substack{ (\alpha,a,b) \in \boldsymbol{\Lambda} \\ \alpha + b -1 = \omega}} 1
\end{equation}
and:
\begin{equation}
    T^*\mathcal{M}_{N, \boldsymbol{d}} = \sum_{\omega}\mathcal{W}_\omega \mathcal{V}^*_\omega + q_1 q_2 \mathcal{V}_{\omega-1} \mathcal{W}^*_\omega - (1-q_1) \mathcal{V}_\omega \mathcal{V}^*_\omega + q_2 (1-q_1)\mathcal{V}_{\omega-1} \mathcal{V}^*_\omega  
\end{equation}
See formulas (119), (123), (145) in \cite{BPSCFT4}. and in different notations formulas (1.10) in \cite{Negut1} together with (3.23) in \cite{Negut2}.

\textbf{$Y$ - observable.} In the orbifolded case $\mathcal{Y}$-observable gets replaced by $N$ ones $\mathcal{Y}_0, ...,\mathcal{Y}_{N-1} $ (see formula (85) in the arxiv version of \cite{BPSCFT5}). Mathematically they are defined as follows:
\begin{equation} \label{Yobsorb}
    \mathcal{Y}_\omega(x): = \begin{cases}
        c_x(R \pi_*(\mathcal{U}_\omega \otimes  p^* \mathcal{O}_0)) \,\, \text{in the $4d$ case} \\
        \wedge_{-e^{-\beta x}} R \pi_*(\mathcal{U}_\omega \otimes  p^* \mathcal{O}_0) \,\, \text{in the $5d$ case} \\
        c_x^{Ell}(R \pi_*(\mathcal{U}_\omega \otimes  p^* \mathcal{O}_0)) \,\, \text{in the $6d$ case}
    \end{cases}
\end{equation}
where $\pi$ and $p$ are projections from the product $\mathcal{M}_{N, \boldsymbol{d}} \times \mathbb{P}^2$ to $\mathcal{M}_{N, \boldsymbol{d}}$ and $\mathbb{P}^2$ correspondingly.
\begin{equation}
    \begin{tikzcd}[column sep=small]
& \mathcal{M}_{N, \boldsymbol{d}} \times \mathbb{P}^2 \arrow[dl, "\pi"] \arrow[dr, "p"] & \\
\mathcal{M}_{N, \boldsymbol{d}}  & & \mathbb{P}^2
\end{tikzcd}
\end{equation}
and $\mathcal{O}_0$ is the skyscraper sheaf of $0 \in \mathbb{P}^2$.\\
From (\ref{Unisheaves}) one can easily get the expression for the restriction of $\mathcal{Y}_\omega(x)$ to the fixed point $\boldsymbol{\Lambda}$:
\begin{equation}
  \mathcal{Y}_\omega(x)|_{\boldsymbol{\Lambda} }=  {\bE}\big[-e^{\beta x}\big( \mathcal{W}_\omega^* - (1-q_1^{-1}) \mathcal{V}_\omega^* + q_2^{-1} (1-q_1^{-1}) \mathcal{V}_{\omega-1} \big)|_{\boldsymbol{\Lambda}}\big]
\end{equation}
or explicitly:
\begin{multline}
\mathcal{Y}_{\omega}(x)|_{\boldsymbol{\Lambda}} = \\ = \vartheta(x-a_\omega) \prod_{(\alpha,a,b) \in \boldsymbol{\Lambda}} \Bigg(\frac{\vartheta(x-a_\alpha - c_{(ab)} - \epsilon_1)}{\vartheta(x-a_\alpha - c_{(ab)})} \Bigg)^{\delta_N(\alpha+b-1-\omega)}  \Bigg(\frac{\vartheta(x-a_\alpha - c_{(ab)} - \epsilon_2)}{\vartheta(x-a_\alpha - c_{(ab)} - \epsilon)} \Bigg)^{\delta_N(\alpha+b-\omega)}   
\end{multline}
where
\begin{equation}
    \delta_N \, \, \text{- delta function on} \, \mathbb{Z}_N
\end{equation}

\textbf{Relation to the non-orbifolded case.}
By direct computation from the definition of $\mathcal{Y}_\omega$ one has
\begin{equation} \label{YtoY}
\prod_{\omega = 0}^{N-1} \mathcal{Y}_{\omega}(x + \omega \epsilon_2) = \mathcal{Y}(x)  
\end{equation}
where $\epsilon_2$ in the RHS  in the definition of $\mathcal{Y}(x)$ is replaced with $N \epsilon_2$ (see formula (122) from \cite{BPSCFT4} or formula (3.19) from \cite{NLJ} for details.)\\

\textbf{$\EuScript{X}$-observable.} Armed with the same logic as in the non-orbifolded case let us give the following:
\begin{definition}
    An observable $\EuScript{X}_\omega(x)$ ($\omega \in \mathbb{Z}_N$) (also denoted as $\EuScript{X}_{34, \omega}(x)$) which is a power series in $\mathfrak{q}_0,...,\mathfrak{q}_{N-1}$ depending on a variable $x$ only through functions $\mathcal{Y}_0, ..., \mathcal{Y}_{N-1}$ at different points and satisfying 2 properties:\\
1)
\begin{equation}
  \EuScript{X}_\omega(x) = \mathcal{Y}_{\omega+1}(x + \epsilon) + \sum_{\omega'} O(\mathfrak{q}_{\omega'}) 
\end{equation}
2) $\langle \EuScript{X}(x) \rangle^{\mathbb{Z}_N}$ has no poles for $x \in \mathbb{C}$\\
is called a $qq$-Character (for the instantons on the orbifold).
\end{definition}
The concrete expression for it could be found in \cite{BPSCFT5} (see formula (68) in the arxiv version):
\begin{multline} \label{qqCharOrb}
    \EuScript{X}_\omega(x)|_{\boldsymbol{\Lambda}} = \mathcal{Y}_{\omega+1}(x+\epsilon)|_{\boldsymbol{\Lambda}} \sum_{\lambda} \mathbb{Q}^\lambda_\omega \mathbb{B}^\lambda_\omega \prod_{ (ij) \in \lambda} \frac{\mathcal{Y}_{\omega+2-j}(x + \sigma_{(ij)} +m + \epsilon)|_{\boldsymbol{\Lambda}} \mathcal{Y}_{\omega+1-j}(x + \sigma_{(ij)} -m )|_{\boldsymbol{\Lambda}} }{\mathcal{Y}_{\omega+2-j}(x + \sigma_{(ij)} + \epsilon)|_{\boldsymbol{\Lambda}} \mathcal{Y}_{\omega+1-j}(x + \sigma_{(ij)})|_{\boldsymbol{\Lambda}}}
\end{multline}
where:
\begin{equation}
    \sigma_\square:= \sigma_{(ij)} = \epsilon_3(i-1) + \epsilon_4(j-1) = m(i-j) + \epsilon (1-j) 
\end{equation}
\begin{equation}
    \mathbb{Q}^\lambda_\omega = \prod_{j = 1}^{\lambda_1} \mathfrak{q}_{\omega+1-j}^{\lambda_j^{t}}
\end{equation}
And $\mathbb{B}^\lambda_\omega$ some function of $\epsilon$'s which becomes 1 in the limit $\epsilon_1 = 0$ - case we will be interested in.\\
\textbf{$\epsilon_1 \rightarrow 0$ limit.} In the limit $\epsilon_1 = 0$ the sum over $N$-tuples of Young diagrams $\boldsymbol{\Lambda}$ will be reduced to the evaluation at the limit shape $\boldsymbol{\Lambda}^\infty$, which is determined by the equations:
\begin{equation} \label{limitshape}
\frac{\mu^{\mathbb{Z}_N}[\boldsymbol{\Lambda}^\infty + \square_\omega]}{\mu^{\mathbb{Z}_N}[\boldsymbol{\Lambda}^\infty ]} =1 
\end{equation}
For any addable box $\square_\omega$ of weight $\omega \in \mathbb{Z}_N$, for all such $\omega$'s.\\
So, let us define:
\begin{gather} \label{NSlimitorb}
Y_\omega(x):= \lim_{\epsilon_1 \rightarrow 0} \langle \mathcal{Y}_\omega(x) \rangle^{\mathbb{Z}_N} =  Y_\omega(x)|_{\boldsymbol{\Lambda}^\infty} \\
\chi_\omega(x):=\lim_{\epsilon_1 \rightarrow 0} \langle \chi_\omega(x) \rangle^{\mathbb{Z}_N} =  \chi_\omega(x) |_{\boldsymbol{\Lambda}^\infty}
\end{gather}
And hence one gets:
\begin{equation} \label{qq-CharLim}
    \chi_\omega(x) = Y_{\omega+1}(x+\epsilon) \sum_{\lambda} \mathbb{Q}^\lambda_\omega  \prod_{ (ij) \in \lambda} \frac{Y_{\omega+2-j}(x + \sigma_{(ij)} +m + \epsilon) Y_{\omega+1-j}(x + \sigma_{(ij)} -m )}{Y_{\omega+2-j}(x + \sigma_{(ij)} + \epsilon) Y_{\omega+1-j}(x + \sigma_{(ij)})}
\end{equation}
The limit shape equation gives the equations on $Y_\omega$. If by $e^{\beta x_\square}$ we will denote the weight of the vector corresponding to the box $\square_\omega$ we add, (\ref{limitshape}) will look like (see formula (139) in \cite{BPSCFT1}):
\begin{equation} \label{limitshape1}
    \mathfrak{q}_\omega \frac{Y_{\omega+1}(x+m+ \epsilon) Y_\omega(x-m) }{Y_{\omega+1}(x + \epsilon) Y_\omega(x)} \Bigg|_{x = x_{\square_\omega}} =- 1
\end{equation}
By definition $\chi_\omega(x)$ is regular in $\mathbb{C}$. Therefore by sending $x$ to $\infty$ in both sides of the equation (\ref{qq-CharLim}) and expanding them up to a constant term in $x$ in 4d case we could find in the limit $\epsilon_1 \rightarrow 0$ (see formulas (87), (96) and (214) in \cite{BPSCFT5}) the full expression for $\chi_\omega$:
\begin{equation} \label{qqCharCl}
    \chi_\omega(x) = \frac{\vartheta(x+\epsilon - p_\omega)}{\prod_{l=1}^\infty \Big( 1 -\frac{x_\omega}{x_{\omega-l}} \Big)}
\end{equation}
where $x_0,...,x_{N-1}$ and $\mathfrak{q}$ are introduced as follows:
\begin{gather}
    \mathfrak{q}_\omega = \frac{x_\omega}{x_{\omega-1}}, \qquad \omega = 1,..., N-1 \\
    \mathfrak{q}_0 = \mathfrak{q} \frac{x_0}{x_{N-1}}
\end{gather}
and extended to infinity by quasiperiodicity.
$p_\omega$'s - are canonically conjugate momenta to $\ln x_\omega$.\\
\textbf{$\epsilon_1 , \epsilon_2 \rightarrow 0$ limit.} In this fully classical limit we will still denote by $Y_\omega$ the limit of $Y_\omega$  defined above under $\epsilon_2 \rightarrow 0$. Their relation to the scalar $Y$ now becomes:
\begin{equation}
\prod_{\omega = 0}^{N-1} Y_{\omega}(x) = Y(x)  
\end{equation}

\textbf{Folded instantons observables.}
In the presence of the orbifold folded instantons, observables also get fractionalized. Let us introduce $\tilde{\mathcal{Q}}_\omega(x)$, so that:
\begin{equation} \label{Qorb}
    \mathcal{Y}_\omega(x) = \frac{\tilde{\mathcal{Q}}_\omega(x)}{\tilde{\mathcal{Q}}_\omega(x+ \epsilon_2 N)}
\end{equation}
And using it, define:
\begin{equation} \label{Y24orb}
    \mathcal{Y}_{24, \omega}(x) = \frac{\tilde{\mathcal{Q}}_\omega(x )}{\tilde{\mathcal{Q}}_\omega(x + m)}
\end{equation}
Now similarly to the non-orbifold case, we would like to give the following:
\begin{definition} \label{foldeddeforb}
    An observable $\EuScript{X}_{24, \omega}(x)$, $\,\omega =0,...,N-1$, which is a power series in $\mathfrak{q}_0,...,\mathfrak{q}_{N-1}$ depending on a variable $x$ only through functions $\mathcal{Y}_{24, 1},...,\mathcal{Y}_{24, N}$ at different points and satisfying 2 properties:\\
    1) 
\begin{equation}
\EuScript{X}_{24, \omega}(x) =     \prod_{c = 0}^{\omega}\mathcal{Y}_{24, c}(x+(c- \omega +N)\epsilon_2) \prod_{c = \omega+1}^{N-1}\mathcal{Y}_{24, c}(x+(c- \omega)\epsilon_2)  + \sum_{\omega'} O(\mathfrak{q}_{\omega'})
\end{equation}
   2) $\langle \EuScript{X}_{24,\omega}(x) \rangle$ has no poles for $x \in \mathbb{C}$\\
   is called a folded instanton observable (in the presence of the orbifold).
\end{definition}
For our purposes, we will only need the explicit expression in the limit $\epsilon_{1} \rightarrow 0$, which is given by (in the notations (\ref{foldednotations}))
\begin{equation} \label{foldedNSorb}
\chi_{24, \omega} =  \lim_{\epsilon_1 \rightarrow 0} \Big\langle\sum_{\mu} \mathbb{Q}_\omega^{(24), \,\mu} {\bE} \Big[ -q_1 q_2 \frac{P_3}{1 - q_2^N} \sum_{\substack{a =0,...,N-1 \\ \omega', \omega'' \in \mathbb{Z}_N\\ a - \omega' + \omega'' = - \omega \, \text{mod} N}} q_2^{a} S_{12, \, \omega'+1}^* S_{24, \, \omega''} \Big|_{\mu} \Big]  \Big\rangle^{\mathbb{Z}_N}    
\end{equation}
Where:
\begin{equation} \label{fug24}
\mathbb{Q}_\omega^{(24), \,\mu} = \prod_{(i,j) \in \mu} \mathfrak{q}_{\omega + i-j}    
\end{equation}
Let us define:
\begin{equation}
    \tilde{Q}_\omega: =\lim_{\epsilon_1 \rightarrow 0} \langle \tilde{\mathcal{Q}}_\omega  \rangle^{\mathbb{Z}_N} =\lim_{\epsilon_1 \rightarrow 0} {\bE} \Big[ -\frac{1}{1 - q_2^N} S_{12, \, \omega}^* \Big|_{\boldsymbol{\Lambda}^\infty}\Big ] 
\end{equation}
So that:
\begin{equation}
    Y_\omega(x) = \frac{\tilde{Q}_\omega(x)}{\tilde{Q}_\omega(x + \epsilon_2 N)}.
\end{equation}
Let us also introduce the notation:
\begin{equation}
    Y_{24,\omega}(x): = \lim_{\epsilon_1 \rightarrow 0} \langle \mathcal{Y}_{24,\omega}  \rangle^{\mathbb{Z}_N} = \frac{\tilde{Q}_\omega(x)}{\tilde{Q}_\omega(x + m)}
\end{equation}

\section{Summary of main results}

\subsection{Factorization formula for qq-Character}

Our first result in the setup of the sections above is that a certain $\theta$-transform (infinite order shift operator) of the orbifolded $qq$-Characters matrix (in the $\epsilon_1 =0$ limit) could be expressed in terms of the infinite product. Hypothetically the $\theta$-transform operation corresponds to replacing the 34-plane which $U(1)$-instantons live on with a Taub-NUT space. The essence of the factorization formula proof lies then in the boson-fermion correspondence. The partition function of the noncommutative  $U(1)$-instantons (Hilbert scheme of points on $\mathbb{C}^2$) is a sum over Young diagrams, which label states in the boson Fock space. They are mapped to fermions with charge zero. Taub-Nut space allows for a non-trivial value of the Wilson loop at infinity, which corresponds to fermions with non-zero charges. And it is expected that the free fermion partition function takes the form of an infinite product. \\

Now let us introduce all the necessary notations for the main identity. We will need the following matrices: 
\begin{gather} \label{factnotation}
    \hat{\chi}(x) = \text{diag} \big( \chi_\omega(x) \big)_{\omega = 0 }^{\omega = N-1} \\
    \hat{Y}(x) = \text{diag} \big( Y_{\omega+1}(x) \big)_{\omega = 0 }^{\omega = N-1} \\
    \hat{\slashed{Q}} = \text{diag} \big( \mathfrak{q}_{\omega} \big)_{\omega = 0 }^{\omega = N-1}
\end{gather}
and the cyclic shift matrix:
\begin{equation} \label{cycl}
    C = \begin{pmatrix} 0 & 1 & 0 &...& 0\\ 0 & 0 & 1 &...& 0 \\ \vdots & \vdots & \ddots & \ddots & \vdots \\ 0 & 0 & 0 & ... & 1 \\ 1 & 0 & 0 & ... & 0\end{pmatrix} = 
    \sum_{\omega} e_\omega \otimes e_{\omega+1}^t
\end{equation}
\begin{theorem} \label{factorizationthm}
Let $z^{\frac{1}{N}}$ - be a formal parameter. Put $\epsilon_1 =0$, hence $\epsilon = \epsilon_2$. The following formula holds in $Mat_{N}((z^{\frac{-1}{N}}e^{\epsilon \partial_x}, z^{\frac{1}{N}} e^{-\epsilon \partial_x}))$:
\begin{multline} \label{factorizationform}
\overleftarrow{\prod_{n=1}^\infty} \Big(1- z^{\frac{1}{N}} \hat{\slashed{Q}}^n \hat{Y}(x+nm +\epsilon) \, C^{-1}e^{-\epsilon \partial_x} \hat{Y}(x+(n-1)m +\epsilon)^{-1} \Big)  \hat{Y}(x + \epsilon) \cdot \\
\cdot \overrightarrow{\prod_{n=0}^\infty} \Big(1 - z^{\frac{-1}{N}} e^{\epsilon \partial_x} (\hat{\slashed{Q}} C^{-1})^n \frac{\hat{Y}(x-(n+1)m - n \epsilon)}{\hat{Y}(x-n m - n \epsilon)} C^{\, n+1}\Big) = \\=
\sum_{n = 0}^\infty (-1)^n z^{\frac{n}{N}} \hat{\chi}(x+nm) \prod_{k = 0}^{n-1} (\hat{\slashed{Q}}^{n-k} C^{-1}) \, e^{-n \epsilon \partial_x} + \\
+ \sum_{n = 1}^\infty (-1)^n z^{\frac{-n}{N}} \hat{\chi}(x-nm) \prod_{k=1}^{n-1} C^k (\hat{\slashed{Q}}C^{-1})^k \, C^n \, e^{n \epsilon \partial_x}
\end{multline}    
\end{theorem}
Notice that by conjugation parameter $z$ could be eliminated.

\subsubsection{Change of basis}
Let us write the formula down in the basis which does not have fractional powers of $z$. For that, we need to introduce the following matrices:
\begin{equation} \label{Sz}
    S_z = \text{diag}(z^{\frac{\omega}{N}})_{\omega = 0}^{\omega = N-1}
\end{equation}
and:
\begin{equation} \label{cyclz}
    C_{z} = \begin{pmatrix} 0 & 1 & 0 &...& 0\\ 0 & 0 & 1 &...& 0 \\ \vdots & \vdots & \ddots & \ddots & \vdots \\ 0 & 0 & 0 & ... & 1 \\ z^{-1} & 0 & 0 & ... & 0\end{pmatrix} 
\end{equation}
With the help of the identity:
\begin{equation} \label{SzId}
z^{-1/N} S_{z}^{-1} C S_z = C_{z}   
\end{equation}
The main formula could be rewritten as follows:
\begin{multline}
\overleftarrow{\prod_{n=1}^\infty} \Big(1-  \hat{\slashed{Q}}^n \hat{Y}(x+nm +\epsilon) \, C_{z}^{-1}e^{-\epsilon \partial_x} \hat{Y}(x+(n-1)m +\epsilon)^{-1} \Big)  \hat{Y}(x + \epsilon) \cdot \\
\cdot \overrightarrow{\prod_{n=0}^\infty} \Big(1 -  e^{\epsilon \partial_x} (\hat{\slashed{Q}} C_{z}^{-1})^n \frac{\hat{Y}(x-(n+1)m - n \epsilon)}{\hat{Y}(x-n m - n \epsilon)} C_{z}^{\, n+1}\Big) = \\=
\sum_{n = 0}^\infty (-1)^n \hat{\chi}(x+nm) \prod_{k = 0}^{n-1} (\hat{\slashed{Q}}^{n-k} C_{z}^{-1}) \, e^{-n \epsilon \partial_x} + \\
+ \sum_{n = 1}^\infty (-1)^n \hat{\chi}(x-nm) \prod_{k=1}^{n-1} C_{z}^k (\hat{\slashed{Q}}C_{z}^{-1})^k \, C_{z}^n \, e^{n \epsilon \partial_x}
\end{multline}
Let us also notice, that the nicest way to write down the formula is probably to introduce the quantum version of the operator $C_{z}$:
\begin{equation}
    \hat{C}_z = C_{z} e^{\epsilon \partial_x}
\end{equation}
Then it will take the form:
\begin{multline} \label{factorization}
\overleftarrow{\prod_{n=1}^\infty} \Big(1-  \hat{\slashed{Q}}^n \hat{Y}(x+nm +\epsilon) \, \hat{C}_z^{-1} \hat{Y}(x+(n-1)m +\epsilon)^{-1} \Big)  \hat{Y}(x + \epsilon) \cdot \\
\cdot \overrightarrow{\prod_{n=0}^\infty} \Big(1 -  (\hat{\slashed{Q}} \hat{C}_z^{-1})^n \frac{\hat{Y}(x-(n+1)m + \epsilon)}{\hat{Y}(x-n m + \epsilon)} \hat{C}_z^{\, n+1}\Big) = \\=
\sum_{n = 0}^\infty (-1)^n \hat{\chi}(x+nm) \prod_{k = 0}^{n-1} (\hat{\slashed{Q}}^{n-k} \hat{C}_z^{-1})  + \\
+ \sum_{n = 1}^\infty (-1)^n \hat{\chi}(x-nm) \prod_{k=1}^{n-1} \hat{C}_z^k (\hat{\slashed{Q}}\hat{C}_z^{-1})^k \, \hat{C}_z^n 
\end{multline}

\subsection{Transformation property of  the factorization formula}
If we denote the LHS of the main formula (\ref{factorization}) by the symbol:
\begin{multline}
\hat{\mathfrak{D}}(x,z) =  \overleftarrow{\prod_{n=1}^\infty} \Big(1-  \hat{\slashed{Q}}^n \hat{Y}(x+nm +\epsilon) \, \hat{C}_z^{-1} \hat{Y}(x+(n-1)m +\epsilon)^{-1} \Big)  \hat{Y}(x + \epsilon) \cdot \\
\cdot \overrightarrow{\prod_{n=0}^\infty} \Big(1 -  (\hat{\slashed{Q}} \hat{C}_z^{-1})^n \frac{\hat{Y}(x-(n+1)m + \epsilon)}{\hat{Y}(x-n m + \epsilon)} \hat{C}_z^{\, n+1}\Big)
\end{multline}
Then the following theorem holds:
\begin{theorem}
\begin{equation}
  X \hat{\mathfrak{D}}(x+m,\mathfrak{q} z) = -\hat{\mathfrak{D}}(x,z) X \hat{C}_{\mathfrak{q} z}  
\end{equation}
where:
\begin{equation} \label{Xmatrix}
    X = \textnormal{diag}(x_\omega)_{\omega = 0}^{\omega = N-1}
\end{equation}
\end{theorem}
\begin{theorem} \label{6dtransform}
    In the 6d case, one also has:
    \begin{equation}
\hat{\mathfrak{D}}(x+ \frac{2 \pi i }{\beta}\tau_{6d},e^{- \beta m N} z)  = - p_{6d}^{-1} \, \big( S_{e^{-\beta m N}}\big)^{-1} e^{-\beta(x-P +\epsilon)} \, \hat{\mathfrak{D}}(x,z) \, S_{e^{-\beta m N}} 
    \end{equation}
where: 
\begin{equation} \label{Pmatrix}
    P = \textnormal{diag}(p_\omega)_{\omega = 0}^{\omega = N-1}
\end{equation}
\end{theorem}

\subsection{Scalar version of the formula}
In order to get the scalar version of the formula one needs to study the same setup as above but without an orbifold. Namely noncommutative $U(1)$-instantons living on one component of a cross and $U(N)$ on another. But this time the latter ones are described by the Gieseker scheme instead of the affine Laumon space.
For the 4d/5d/6d super-Yang-Mills  partition function with an adjoint matter on $\mathbb{R}^4$ before orbifold one would get the following formula for the $qq$-Character in the $\epsilon_1 = 0$ limit, see (\ref{qq-CharLimitNonOrb}):
\begin{equation}
    \chi(x) = Y(x+\epsilon) \sum_{\lambda}  \mathfrak{q}^{|\lambda|} \prod_{ \square \in \lambda} \frac{Y(x + \sigma_\square +m + \epsilon) Y(x + \sigma_\square -m )}{Y(x + \sigma_\square + \epsilon) Y(x + \sigma_\square)}
\end{equation}
\begin{theorem}
Let $z$ - be a formal parameter. And $\epsilon_1=0$. The following formula holds in $\mathbb{C}((z^{-1} e^{\epsilon \partial_x}, z e^{-\epsilon \partial_x}))$:
\begin{multline} \label{FactorizationScalar}
\overleftarrow{\prod_{n=1}^\infty} \Big(1- z\mathfrak{q}^n Y(x+nm +\epsilon) \, e^{-\epsilon \partial_x} Y(x+(n-1)m +\epsilon)^{-1} \Big)  Y(x + \epsilon) \cdot \\
\cdot \overrightarrow{\prod_{n=0}^\infty} \Big(1 - z^{-1} \mathfrak{q}^n \frac{Y(x-(n+1)m + \epsilon)}{Y(x-n m +  \epsilon)} e^{\epsilon \partial_x} \Big) = \\=
\sum_{n \in \mathbb{Z} } (-z)^{n} \mathfrak{q}^\frac{n^2 + n}{2} \chi(x+nm) \, e^{-n \epsilon \partial_x} 
\end{multline}
\begin{proof}
    The proof goes exactly like in the matrix case, but easier, because we do not need to keep track of indices. 
\end{proof}
\end{theorem}

\subsection{Classical limit and the spectral curve}
In the classical limit $\epsilon_1, \epsilon_2 =0$ the operator expression reduces to the identity on the section of the homomorphisms bundle between two vector bundles over a certain variety. Let us describe the variety and the bundle.\\
As the value for $\hat{\mathfrak{D}}(x,z)$ at point $(x,z)$ is related to the value at the point $(x+m, \mathfrak{q} z)$, the manifold which we need to study in the classical limit in the equivariant cohomology case is:
\begin{equation}
   ( \mathbb{C} \times \mathbb{C}^{\times} ) / \mathbb{Z}
\end{equation}
In the equivariant $K$-theory case there is also a periodicity in $x$: $x \rightarrow x + \frac{2 \pi i}{\beta}$, hence the corresponding manifold is:
\begin{equation}
   ( \mathbb{C}^{\times} \times \mathbb{C}^{\times} ) / \mathbb{Z}
\end{equation}
And in the $6d$ case the values at $(x,z)$ and $(x+\frac{2 \pi i }{\beta}\tau_{6d}, e^{-\beta m N}z)$ are also related, and therefore the corresponding manifold in consideration is the slanted product of two elliptic curves:
\begin{equation}
   ( \mathbb{C}^{\times} \times \mathbb{C}^{\times} ) / (\mathbb{Z} \times \mathbb{Z})
\end{equation}
The operator $\mathfrak{D}(x,z)$ becomes the section of the homomorphism bundle with the transition functions given by:
\begin{equation}
X \mathfrak{D}(x+m,\mathfrak{q} z) = -\mathfrak{D}(x,z) X C_{\mathfrak{q} z}      
\end{equation}
and in $6d$-case also:
    \begin{equation}
S_{e^{-\beta m N}} \, \hat{\mathfrak{D}}(x+\frac{2 \pi i }{\beta}\tau_{6d},e^{- \beta m N} z)  = - p_{6d}^{-1} e^{-\beta(x-P)} \, \hat{\mathfrak{D}}(x,z) \, S_{e^{-\beta m N}} 
    \end{equation}
The quantum equation (\ref{factorization}) then reduces to  
\begin{multline} \label{factorizationCl}
\overleftarrow{\prod_{n=1}^\infty} \Big(1-  \hat{\slashed{Q}}^n \hat{Y}(x+nm ) \, C_{z}^{-1} \hat{Y}(x+(n-1)m )^{-1} \Big)  \hat{Y}(x ) \cdot \\
\cdot \overrightarrow{\prod_{n=0}^\infty} \Big(1 -   (\hat{\slashed{Q}} C_{z}^{-1})^n \frac{\hat{Y}(x-(n+1)m)}{\hat{Y}(x-n m)} C_{z}^{\, n+1}\Big) = \\=
\sum_{n = 0}^\infty (-1)^n \hat{\chi}(x+nm) \prod_{k = 0}^{n-1} (\hat{\slashed{Q}}^{n-k} C_{z}^{-1}) \,  + \\
+ \sum_{n = 1}^\infty (-1)^n \hat{\chi}(x-nm) \prod_{k=1}^{n-1} C_{z}^k (\hat{\slashed{Q}}C_{z}^{-1})^k \, C_{z}^n \, 
\end{multline}

The scalar formula in the classical case takes the form:
\begin{multline} \label{factorizationClSc}
\prod_{n=1}^\infty \Big(1- z\mathfrak{q}^n Y(x+nm) \,  Y(x+(n-1)m )^{-1} \Big)  Y(x) \cdot \\
\cdot \prod_{n=0}^\infty \Big(1 - z^{-1} \mathfrak{q}^n \frac{Y(x-(n+1)m )}{Y(x-n m )} \Big) = \\=
\sum_{n \in \mathbb{Z} } (-z)^{n}  \mathfrak{q}^{\frac{n^2+n}{2}} \chi(x+nm) \, 
\end{multline}
By straightforward calculation the determinant of LHS of the identity (\ref{factorizationCl}) is equal to the classical limit of the LHS of (\ref{factorizationClSc}) as soon as:
\begin{equation}
\prod_{\omega = 0}^{N-1} Y_{\omega}(x) = Y(x)  
\end{equation}
This follows by direct calculation from the definition of $Y_{\omega}(x), Y(x)$ (see formula (122) from \cite{BPSCFT4} or formula (3.19) from \cite{NLJ} for details.)
Hence one obtains the formula:
\begin{equation} \label{det}
    \det \mathfrak{D}(x,z) = \sum_{n \in \mathbb{Z} } (-z)^{n}  \mathfrak{q}^{\frac{n^2+n}{2}} \chi(x+nm)
\end{equation}
The RHS of this formula equals to zero - is the spectral curve of ell CM / ell RS/ Dell systems (\cite{NG}). And the LHS of this formula then gives the expressions for conserved Hamiltonians in terms of $x_\omega$'s and $p_\omega$'s. For the $4d$ case it was proven in \cite{BPSCFT5} that these $x_\omega$'s and $p_\omega$'s are indeed the standard coordinates for the elliptic Calogero-Moser system. In the next section, we will prove it again in a different way. See also the $N=2$ general calculation in the Appendix.

\subsection{New expression for the Lax matrix of the elliptic Calogero-Moser system}

This section follows the ideas of \cite{GNlax}:
notice, that from the RHS of the formula (\ref{factorizationCl}) and linearity of the fractional qq-Characters in $x$ (in the 4d case) follows the linearity of the matrix $\hat{\mathfrak{D}}(x,z)$:
\begin{equation}
\mathfrak{D}^{4d}(x,z) = \mathfrak{D}_0(z) + x \, \mathfrak{D}_1(z)   
\end{equation}
Let
\begin{equation}
    e = \begin{pmatrix}
        1 \\ 1\\ \vdots \\ 1
    \end{pmatrix}
\end{equation}
\begin{theorem} \label{Laxthm}
The Lax matrix of the Calogero-Moser system with canonical coordinates $x_\omega$, $p_\omega$ in the standard meromorphic gauge with the residue at the simple pole $z=1$ equal to $-m \, e \otimes e^t$ and quasi-periodicity:
\begin{equation}
    X L( \mathfrak{q} z) X^{-1} - m = L(z)
\end{equation}
takes the form:
\begin{equation}
    L(z) = -U^{-1} \, \mathfrak{D}_0(z) \, \mathfrak{D}_1(z)^{-1} U
\end{equation}
where
\begin{equation} \label{Umatrix}
    U_{\omega \, \omega'} = \delta_{\omega \, \omega'} \sum_{s=0}^{\infty}(-1)^s \sum_{0< n_1 <...< n_s} \, \prod_{k=0}^{s-1} \mathfrak{q}_{\omega-k}^{n_{s-k}}
\end{equation}
Explicitly the matrix with such transformation property could be written as:
\begin{equation}
L_{ij}(z) = (p_i - m E_1(z)) \, \delta_{ij} - m (1- \delta_{ij}) \frac{\theta_{\mathfrak{q}}'(1) \theta_{\mathfrak{q}}(z x_i / x_j)}{\theta_{\mathfrak{q}}(z) \theta_{\mathfrak{q}}( x_i / x_j)}   
\end{equation}
where:
\begin{gather}
    \theta_{\mathfrak{q}}(z) = \prod_{n=0}^\infty (1 - \mathfrak{q}^{n+1}) (1- \mathfrak{q}^n z) (1- \mathfrak{q}^{n+1}z^{-1}) \\
    E_{1}(z) = \frac{z \frac{d}{dz} \theta_{\mathfrak{q}}(z)}{\theta_{\mathfrak{q}}(z)} 
\end{gather}
Also, we have found the following new formula for it:
\begin{multline}
U L(z) U^{-1} = P + m \sum_{k=1}^\infty \Pi_{k} \cdot \hat{\slashed{Q}}^k C_z^{-1} \cdot \Pi_{k-1}^{-1} - m \sum_{k=0}^\infty \Pi_0 \cdot \overrightarrow{\Pi}^{k} \cdot (\hat{\slashed{Q}} C_{z}^{-1})^k  C_{z}^{\, k+1} \cdot (\overleftarrow{\Pi}^{k+1})^{-1} \cdot \Pi_0^{-1}  
\end{multline}
where:
\begin{gather} \label{products}
    \Pi_{k} = \overleftarrow{\prod_{n =k+1}^\infty} \Big(1-  \hat{\slashed{Q}}^n  \, C_{z}^{-1} \Big) \\
  \overrightarrow{\Pi}^{k}  = \overrightarrow{\prod_{n=0}^{k-1}} \Big(1 -   (\hat{\slashed{Q}} C_{z}^{-1})^n C_{z}^{\, n+1}\Big) \\
   (\overleftarrow{\Pi}^{k})^{-1}  = \overleftarrow{\prod_{n=0}^{k-1}} \Big(1 -   (\hat{\slashed{Q}} C_{z}^{-1})^n C_{z}^{\, n+1}\Big)^{-1}
\end{gather}
\end{theorem}

Now, one can notice, that for $\mathfrak{D}_1(z)$ we also have the similar factorization identity:
\begin{multline} \label{D1}
 \mathfrak{D}_1(z) =  \overleftarrow{\prod_{n=1}^\infty} \Big(1-  \hat{\slashed{Q}}^n  \, C_{z}^{-1} \Big) 
\cdot \overrightarrow{\prod_{n=0}^\infty} \Big(1 -   (\hat{\slashed{Q}} C_{z}^{-1})^n C_{z}^{\, n+1}\Big)   = \\ =
\sum_{n = 0}^\infty (-1)^n \hat{\mathbb{B}} \prod_{k = 0}^{n-1} (\hat{\slashed{Q}}^{n-k} C_{z}^{-1}) 
+ \sum_{n = 1}^\infty (-1)^n \hat{\mathbb{B}}\prod_{k=1}^{n-1} C_{z}^k (\hat{\slashed{Q}}C_{z}^{-1})^k \, C_{z}^n 
\end{multline}
where:
\begin{equation}
    \hat{\mathbb{B}} = \text{diag}(\mathbb{B}_\omega)_{\omega = 0}^{\omega = N-1}
\end{equation}
\begin{equation}
  \mathbb{B}_\omega = \prod_{l=1}^\infty \frac{1}{\Big( 1 -\frac{x_\omega}{x_{\omega-l}} \Big)}  
\end{equation}
By taking the determinant one obtains:
\begin{equation}
    \det \mathfrak{D}_1(z) = \frac{\theta_{\mathfrak{q}}(z^{-1})}{\prod_{n=1}^{\infty} (1 - \mathfrak{q}^n)}
\end{equation}
Hence the spectral curve equation looks as follows:
\begin{equation}
0 = \theta_{\mathfrak{q}}(z^{-1}) \frac{\det(x - L(z))}{\prod_{n=1}^{\infty} (1 - \mathfrak{q}^n)} = \sum_{n \in \mathbb{Z} } (-z)^{n}  \mathfrak{q}^{\frac{n^2+n}{2}} \chi(x+nm)   
\end{equation}

\subsection{Elliptic Ruijsenaars-Schneider system Lax matrix}
In the 5d case, again from the linearity of the qq-characters now in coordinate $e^{- \beta x}$, one could write:
\begin{equation}
    \mathfrak{D}(x,z) =\mathfrak{D}_1(z) - \mathfrak{D}_\infty(z) e^{- \beta x} 
\end{equation}
for some matrices $\mathfrak{D}_1(z)$ and $\mathfrak{D}_\infty(z)$. In fact $\mathfrak{D}_1(z)$ is exactly (\ref{D1}).
\begin{theorem} \label{RSLaxthm}
    The Lax matrix for the elliptic Ruijsenaars-Schneider model with quasi-periodicity:
\begin{equation}
    X L^{RS}(\mathfrak{q}z) X^{-1} = e^{\beta m} L^{RS}(z)
\end{equation}
    could be expressed in the form:
\begin{equation}
    L^{RS}(z) = U^{-1}\mathfrak{D}_1(z) \mathfrak{D}^{-1}_\infty(z) U
\end{equation} 
where the diagonal matrix $U$ was introduced in (\ref{Umatrix}).\\
Explicitly, in terms of theta functions, it could be written as:
\begin{equation}
    L^{RS}_{ij}(z) = (1-e^{\beta m}) u^{RS}_i u_i^{-1} \frac{\theta_{\mathfrak{q}}'(1) \theta_{\mathfrak{q}}(e^{-\beta m} z x_i /x_j)}{\theta_{\mathfrak{q}}(z) \theta_{\mathfrak{q}}(e^{-\beta m} x_i /x_j)}
\end{equation}
where $u_i$'s - are the entries of the matrix $U$, and $u^{RS}_i$ are defined as the components of the vector:
\begin{gather}
    u^{RS} = u^{RS}(\slashed{Q},P) =  \overleftarrow{\prod_{n=1}^\infty} \Big(1-  \hat{\slashed{Q}}^n  \, C_{1}^{-1} \Big) \,e^{\beta P} \, \cdot e
\end{gather}
And, similarly to the 4d case, one has:
\begin{multline}
U L(z) U^{-1} = e^{\beta P} +  \sum_{k=1}^\infty \Pi_{k} \cdot \hat{\slashed{Q}}^k \Big[e^{\beta P} C_z^{-1} (1- e^{- \beta m k}) - (1-e^{- \beta m k + \beta m}) C_z^{-1} e^{\beta P} \Big]  \cdot \Pi_{k-1}^{-1} -\\ -(e^{\beta m} -1) \sum_{k=0}^\infty \Pi_0 \cdot \overrightarrow{\Pi}^{k} \cdot  e^{\beta P} e^{\beta m k}  (\hat{\slashed{Q}} C_{z}^{-1})^k  C_{z}^{\, k+1} \cdot (\overleftarrow{\Pi}^{k+1})^{-1} \cdot \Pi_0^{-1}  
\end{multline}
The products appearing in the formula were defined in (\ref{products}).
\end{theorem}
The theorem allows us to write the spectral curve equation for the elliptic Ruijsenaars system in the form:
\begin{equation}
0 = \theta_{\mathfrak{q}}(z^{-1}) \frac{\det(1 - e^{- \beta x} L^{RS}(z))}{\prod_{n=1}^{\infty} (1 - \mathfrak{q}^n)} = \sum_{n \in \mathbb{Z} } (-z)^{n}  \mathfrak{q}^{\frac{n^2+n}{2}} \chi(x+nm)   
\end{equation}

\subsection{Trigonometric limit and QC duality}
For the trigonometric limit let us introduce the following notations:
\begin{equation}
\mathfrak{D}(x,z)^{\text{trig}} = \lim_{\mathfrak{q} \rightarrow 0} \mathfrak{D}(x,z)    
\end{equation}
The main result here could be summarized in the following theorem:
\begin{theorem} \label{Trigthm}
  1)  The limit:
    \begin{equation}
       \mathfrak{D}(x)^{\text{trig}} : = \lim_{z \rightarrow \infty} \mathfrak{D}(x,z)^{\text{trig}}
    \end{equation}
    exists.\\
  2) The spectral curve equation reduces to:
     \begin{equation}
         0 =\det \mathfrak{D}(x,z)^{\text{trig}} = \chi(x) - z^{-1} \chi(x-m) = Y(x) - z^{-1} Y(x-m) 
     \end{equation}
 3) The equation (\ref{factorizationCl}) could be solved for $Y_\omega$, expressing them trough $\chi_\omega$, in the following way. Let $Q_k$ be the $k$-th principal minor of the matrix $\mathfrak{D}(x)^{\text{trig}}$:
 \begin{equation}
     Q_k (x) = \det_{1 \leq i, j \leq k} \mathfrak{D}(x)^{\text{trig}} \qquad k = 1, ..., N 
 \end{equation}
 and $Q_0 =1$\\
 Then the $Y_{\omega}$ variables are expressed in terms of them as:
 \begin{equation}
     Y_{\omega}(x) = \frac{Q_{\omega} (x)}{Q_{\omega-1}(x)} \qquad \omega = 1,..., N
 \end{equation}
\end{theorem}

\begin{corollary}[QC duality]
The set of functions $Q_\omega(x)$ satisfy the Bethe equations:
\begin{equation}
    \frac{x_\omega}{x_{\omega-1}} \frac{Q_{\omega+1}(x+m) Q_{\omega}(x-m) Q_{\omega-1}(x) }{ Q_{\omega+1}(x) Q_{\omega}(x+m) Q_{\omega}(x-m)} \Big|_{\textnormal{zeros of } Q_\omega(x)} = -1
\end{equation}
with $Q_N$ being equal to:
\begin{equation}
    Q_N(x) = \det \mathfrak{D}(x)^{\text{trig}} = \chi(x) = Y(x) = \prod_{\alpha=1}^N \vartheta(x-a_\alpha)
\end{equation}
\end{corollary}
\begin{proof}
    The first statement follows directly from the formula (\ref{limitshape1}), and the second from the theorem above.
\end{proof}
This statement is the essence of the proof (in one direction) of the Quantum-Classical duality between trig CM and XXX spin chain, trig RS and XXZ spin chain, and dual ell RS and XYZ spin chain, introduced in \cite{GK} and studied in various forms (including the spectral dual frame) in
\cite{MTV}, \cite{ALTZ}, \cite{AKLTZ}, \cite{Z}, \cite{GZZ}, \cite{BKK}, \cite{ZK}, \cite{ZKS}, \cite{FZKS} \cite{BLZZ}, \cite{MNS} \cite{NSBetheGauge}, \cite{NW}, \cite{NRS}. The coordinates of the many-body systems $x_\omega$ become twist parameters of the spin chain, and the eigenvalues of the Lax operator $a_\omega$ become inhomogeneity parameters.

\subsection{Lax eigenvector}
The Lax operator eigenvector will be constructed out of the orbifolded version of the folded instantons partition function in the limit $\epsilon_1 = 0$. Namely

\begin{theorem} \label{Eigenvectorthm}
\begin{equation}
    \sum_{n \in \mathbb{Z}} (-1)^{n} \hat{\chi}(x+nm) D^{(n)}  C^{-n}  e^{-n \ve_2 \partial_x} \, \, \vec{\chi}_{24}(x-m)  = 0
\end{equation}
where
\begin{equation}
    \hat{\chi}(x) = \textnormal{diag} \big( \chi_\omega(x) \big)_{\omega = 0 }^{\omega = N-1} 
\end{equation}
\begin{equation}
  \vec{\chi}_{24}(x)  = \begin{pmatrix} \chi_{24, 0}(x) \\ \chi_{24, 1}(x) \\
    \vdots \\
    \chi_{24, N-1}(x)
    \end{pmatrix}
\end{equation}
($\chi_{24,\omega}(x)$ were defined in (\ref{foldedNSorb})), 
and $D^{(n)}$ is a diagonal matrix:
\begin{equation}
    D^{(n)} = \begin{cases}
        \prod_{k = 0}^{n-1} (\hat{\slashed{Q}}^{n-k} C^{-1}) \, C^n , \qquad  n \geq 0 \\
        \prod_{k=1}^{n-1} C^k (\hat{\slashed{Q}}C^{-1})^k \qquad n< 0
    \end{cases}
\end{equation}
Explicitly, it looks like:
\begin{equation}
    D^{(n)} = \begin{cases}
        \textnormal{diag}(\prod_{k = 0}^{n-1} \mathfrak{q}_{\omega-k}^{n-k})_{\omega =0}^{\omega=N-1}, \qquad  n \geq 0 \\
        \textnormal{diag}(\prod_{k = 1}^{-n-1} \mathfrak{q}_{\omega+k}^{-n-k})_{\omega =0}^{\omega=N-1}, \qquad  n < 0
    \end{cases}
\end{equation}
\end{theorem}
After the Fourier transform:
\begin{equation}
    \vec{\Psi}_\delta(z) : = \mathfrak{D}_1(z) \sum_{x \in \epsilon_2 \mathbb{Z} + \delta} z^{\frac{x}{N \epsilon_2}}\, \vec{\chi}_{24}(x-m) 
\end{equation}
one obtains in the 4d case:
\begin{equation}
    \Big[N \epsilon_2 z \frac{d}{dz} - \tilde{L}(z) \Big] \vec{\Psi}_\delta(z) = 0
\end{equation}
and in the 5d case:
\begin{equation}
    \Big[e^{ \beta N \epsilon_2 z \frac{d}{dz} } - \tilde{L}^{RS}(z) \Big] \vec{\Psi}_\delta(z) = 0
\end{equation}
where $\delta$ is a parameter labeling different solutions.

\subsection{Taking the limit in the Lax eigenvector formula}
The formula above for $\vec{\chi}_{24}(x)$ depends on the parameter $\epsilon_2$, which corresponds to the quantization of the spectral parameter in the Lax matrix. 
In order to get the eigenvector for the usual classical Lax matrix we need to take the limit $\epsilon_2 \rightarrow 0$. \\
In order to take this limit we need to introduce more convenient notations for the characters appearing in the expression for folded instanton (\ref{foldedNSorb})):
\begin{gather}
    S_{12, a} = q_2^a \tilde{S}_{12, a} \\
    S_{24, a} = q_2^a \tilde{S}_{24, a}
\end{gather}
explicitly they will look like this:
\begin{equation}
\tilde{S}_{12, a} = \tilde{N}_{12,a} - P_1 \tilde{K}_{12,a} + P_1 q_2^{N \delta_{a 0}} \tilde{K}_{12,a-1}  
\end{equation}
and
\begin{equation}
\tilde{S}_{24, a} = \tilde{N}_{24,a} -  \tilde{K}_{24,a} + q_2^{N \delta_{a 0}} \tilde{K}_{24,a-1} +   q_2^{-N \delta_{a, N-1}} q_2 q_4 \tilde{K}_{24,a+1} - q_2 q_4 \tilde{K}_{24,a}
\end{equation}
where:
\begin{gather}
    \tilde{N}_{12,a} = q_2^{-a} N_{12,a}\\
    \tilde{K}_{12,a} = q_2^{-a} K_{12,a} \\
    \tilde{N}_{24,a} = q_2^{-a} N_{24,a}\\
    \tilde{K}_{24,a} = q_2^{-a} K_{24,a}
\end{gather}
for $a=0,...,N-1$. And let
\begin{gather}
    \tilde{N}_{12} = \sum_{a=0}^{N-1} \tilde{N}_{12,a} \\
    \tilde{N}_{24} = \sum_{a=0}^{N-1} \tilde{N}_{24,a}
\end{gather}
Notice that:
\begin{equation}
  \tilde{S}_{12} : =  \sum_{a=0}^{N-1} \tilde{S}_{12, a} = \tilde{N}_{12} - P_1(1-q_2^N) \tilde{K}_{12,N-1}
\end{equation}
and 
\begin{equation}
   \tilde{S}_{24}:=  \sum_{a=0}^{N-1} \tilde{S}_{24, a} = \tilde{N}_{24} - (1-q_2^N) \tilde{K}_{24,N-1} + q_2 q_4 (1-q_2^{-N}) \tilde{K}_{24,0} 
\end{equation}
contain factors $(1-q^N)$, allowing us to break the character in the argument of the plethystic exponent in (\ref{foldedNSorb})) into singular and non-singular parts as $\epsilon_2 \rightarrow 0$. Indeed, let us rewrite it as:
\begin{equation}
    \sum_{\substack{a =1,...,N \\ \omega', \omega'' \in \mathbb{Z}_N\\ a - \omega' + \omega'' = - \omega \, \text{mod} N}} q_2^{a} S_{12, \, \omega'}^* S_{24, \, \omega''} = \sum_{\substack{a =1,...,N \\ \omega', \omega'' =0,...,N-1\\ a - \omega' + \omega'' = - \omega \, \text{mod} N}} q_2^{a - \omega' + \omega''} \tilde{S}_{12, \, \omega'}^* \tilde{S}_{24, \, \omega''}
\end{equation}
which is equal to:
\begin{equation}
  \sum_{ \omega' - \omega '' >  \omega }  \tilde{S}_{12, \, \omega'}^* \tilde{S}_{24, \, \omega''} + q_2^N \sum_{ \omega -N  \leq \omega' - \omega '' \leq  \omega } \tilde{S}_{12, \, \omega'}^* \tilde{S}_{24, \, \omega''} +  q_2^{2N} \sum_{  \omega' - \omega '' < \omega -N  } \tilde{S}_{12, \, \omega'+1}^* \tilde{S}_{24, \, \omega''} 
\end{equation}
and finally gives us:
\begin{equation}
     \tilde{S}_{12}^*\tilde{S}_{24} - (1-q_2^N) \sum_{ \omega -N  \leq \omega' - \omega '' \leq  \omega }  \tilde{S}_{12, \, \omega'}^* \tilde{S}_{24, \, \omega''}  - (1-q_2^{2N}) \sum_{  \omega' - \omega '' < \omega -N  }  \tilde{S}_{12, \, \omega'}^* \tilde{S}_{24, \, \omega''} 
\end{equation}
Hence the only singular part of the character in the argument of the plethystic exponent in (\ref{foldedNSorb})) is equal to:
\begin{equation}
    -\frac{P_3}{1 - q_2^N} \tilde{S}_{12} \tilde{N}_{24} 
\end{equation}
and the non-singular part could be expressed nicely with the help of the observation, that in the limit $q_2 = 1$:
\begin{equation}
 P_3  \tilde{S}_{24, a} = P_3  \tilde{N}_{24,a}  - T_a + T_{a-1} + q_3^{-1} T_{a+1} - q_3^{-1} T_a
\end{equation}
where:
\begin{equation}
    T_a \big|_\mu = \sum_{\substack{ (i,j) \in \partial_+ \mu \\ i-j= a \,  \text{mod} N }} (1-q_3^{-j+1}) + \sum_{\substack{ (i,j) \in \partial_- \mu \\ i-j= a \, \text{mod} N }} (1-q_3^{-j+2})  
\end{equation}
Eventually one can write:
\begin{equation}
    \chi_{24,\omega}(x) = {\bE} \Big[ - \frac{P_3}{1-q_2^N} \tilde{S}_{12}^* \tilde{N}_{24}\Big] \, \chi_{24,\omega}^{reg}(x) 
\end{equation}
where the limit:
\begin{equation}
    \lim_{\epsilon_2 \rightarrow 0} \chi_{24,\omega}^{reg}(x)
\end{equation}
is well defined, and equal to:
\begin{equation}
\lim_{\epsilon_2 \rightarrow 0} \chi_{24,\omega}^{reg}(x) = \sum_{\mu} \mathbb{Q}_\omega^{(24), \,\mu} {\bE} \Big[ - \text{Ch}_{24,\omega}^{reg}(x)\Big] \Big|_\mu  
\end{equation}
\begin{multline}
\text{Ch}_{24,\omega}^{reg}(x) = -P_3 S_{12}^* K_{24, N-1} - P_3 q_3^{-1} K_{24,0} - \\- \sum_{\omega-N \leq \omega' - \omega'' \leq \omega} P_3 S_{12, \omega'}^* S_{24, \omega''} - 2 \sum_{\omega' - \omega'' < \omega - N} P_3 S_{12, \omega'}^* S_{24, \omega''}   
\end{multline}
Now we would like to reproduce the equation for the spectral curve. Let us notice, that the singular part is precisely:
\begin{equation}
  {\bE} \Big[ - \frac{P_3}{1-q_2^N} \tilde{S}_{12}^* \tilde{N}_{24}\Big] = Y_{24}(x) 
\end{equation}
Hence, as we have
\begin{equation}
    \frac{Y_{24}(x)}{Y_{24}(x + \epsilon_2 N)} = \frac{Y(x)}{Y(x+m)}
\end{equation}
in the limit $\epsilon_2 \rightarrow 0$ this gives us:
\begin{equation}
Y_{24}(x) \rightarrow e^{- \frac{S(x)}{N \epsilon_2}}    
\end{equation}
where:
\begin{equation}
    \frac{dS}{dx} = \ln Y(x) - \ln Y(x+m)
\end{equation}
Therefore:
\begin{equation}
    e^{- n \epsilon_2 \partial_x} \chi_{24,\omega}(x-m) \rightarrow z^{\frac{n}{N}} e^{- \frac{S(x)}{N \epsilon_2}} \lim_{\epsilon_2 \rightarrow 0} \chi_{24,\omega}^{reg}(x-m)
\end{equation}
where:
\begin{equation} \label{eigenvalue}
    z = \frac{Y(x-m)}{Y(x)}
\end{equation}
and
\begin{equation}
     \lim_{\epsilon_2 \rightarrow 0} \chi_{24,\omega}^{reg}(x-m)
\end{equation}
is therefore components of the zero vector of the matrix:
\begin{equation}
    S_z \mathfrak{D}(x,z) S_z^{-1}
\end{equation}
see formula (\ref{factorizationCl}), (\ref{factorizationform}) and (\ref{SzId}). And hence the eigenvector of the Lax matrix $L(z)$ has the form:
\begin{equation}
    \vec{\Phi}_{24}(x): = U^{-1} \mathfrak{D}_1(z) S_z^{-1} \cdot \lim_{\epsilon_2 \rightarrow 0} \vec{\chi}_{24}^{\,reg}(x-m)
\end{equation}
\begin{equation}
    L(z) \, \vec{\Phi}_{24}(x) = x \,  \vec{\Phi}_{24}(x)
\end{equation}
with the eigenvalue $x$ determined by the equation (\ref{eigenvalue}), which indeed corresponds to the zero of the determinant
\begin{equation}
    \det \big(x - L(z)\big)
\end{equation}
(see $n=0$ multiple of the product in the LHS of (\ref{factorizationCl}), and (\ref{factorizationClSc})).

\subsection{Towards spectral duality for ellCM, ellRS and Dell}
If $L(z)$ is a Lax operator for some integrable system depending on the spectral parameter $z$, and 
\begin{equation}
    \det(x - L(z)) = 0
\end{equation}
is its spectral curve equation, we usually call some other integrable system spectral dual to the initial one, if its Lax operator $L^{\text{Dual}}(x)$ (depending on spectral parameter $x$) satisfies the relation:
\begin{equation}
      \det(z - L^{\text{Dual}}(x))=\det(x - L(z))
\end{equation}
This type of duality was first observed for Toda chain, and then studied extensively for other systems. See, for example: \cite{Harnad},\cite{BEH}, \cite{MMRZZ} \cite{GMMM}.
For elliptic integrable systems, whose Lax operator has an elliptic dependence on parameter $z$ we would expect the spectral dual system to have Lax matrix of infinite size. Indeed,
let us introduce the following notations:
\begin{equation}
    Y_{\omega,n}:=  Y_{\omega}(x+nm) \quad n \in \mathbb{Z}
\end{equation}
$E_{nm} \, - \,\, n,m \in \mathbb{Z}$ be the standard basis in $\mathfrak{g} \mathfrak{l}(\infty)$. Define the set of operators:
\begin{equation} \label{SpinLax}
    \mathcal{L}_\omega(x) = \sum_{n \in \mathbb{Z}} \mathfrak{q}_\omega^{-n} \frac{Y_{\omega,n-1}}{Y_{\omega+1,n}} \Big(E_{nn} + E_{n n+1} \Big)
\end{equation}
Then the following theorem holds:
\begin{theorem} \label{Spectraldual}
Equation
\begin{equation}
    \det_{N \times N} \mathfrak{D}(x,z) = 0
\end{equation}
up to ill-defined infinite factor:
\begin{equation}
    \lim_{n \rightarrow \infty} z^{-n} \mathfrak{q}^{\frac{n^2-n}{2}} Y(x-nm) = \lim_{n \rightarrow \infty} z^{-n} \mathfrak{q}^{\frac{n^2-n}{2}} (-nm)^N
\end{equation}
is equivalent to:
\begin{equation}
    \det_{\infty \times \infty } \big(1 - z \, T_N(x)\big) = 0
\end{equation}
where $T_N(x)$ is given by:
\begin{equation}
    T_N(x)^{-1} = \mathcal{L}_1(x) \, \mathcal{L}_2(x)\cdot...\cdot \mathcal{L}_{N}(x)
\end{equation}
\end{theorem}
\begin{rem}
    In order to avoid nasty infinite factors the precise equality could be formulated for the ratios of the determinants:
\begin{equation}
    \frac{\det_{N \times N} \mathfrak{D}(x+m,z)}{\det_{N \times N} \mathfrak{D}(x,z)} = \frac{ \det_{\infty \times \infty } \big(1 - z \, T_N(x+m)\big)}{ \det_{\infty \times \infty } \big(1 - z \, T_N(x)\big)}
\end{equation}
And now they got canceled. \\
It is also interesting to notice that such a ratio of generating functions of Hamiltonians (in the quantum case) was first used by M. Nazarov and E. Sklyanin to study $N \rightarrow \infty$ limit of the trigonometric Ruijsenaars-Schneider model \cite{NS}. For them it was pure convenience, however, in the analogous situation for elliptic in momenta generalization of tRS system \cite{GZ1} without taking such ratios the modes of the generating function do not even provide a family of commuting Hamiltonians.
\end{rem}
From the expressions for $T_N(x)$ it is clear that the dual system is some sort of spin chain, and from the expression for $\mathcal{L}_\omega(x)$ it is clear that each site is associated with some infinite-dimensional representation of some algebra. However, as the Poisson bracket between $\mathfrak{q}_{\omega'}$ and $Y_\omega(x)$ is very non-trivial, we were not able to identify the corresponding $rLL$-structure. We expect it to be related to affine Yangian / quantum toroidal algebra / elliptic quantum toroidal algebra.

\section{ Proof of Factorization formula for qq-Character}
\textbf{Remark:} In our next paper we will give an alternative simpler proof of this statement, which will be based on the results from \cite{NG,ncJac}.

\begin{proof}[Proof of theorem \ref{factorizationthm}]
 By simple cancellation of factors the formula for the $qq$-Character (\ref{qq-CharLim}) could be rewritten as:
\begin{equation} \label{qChar}
    \chi_\omega(x) =  \sum_{\lambda} \prod_{j = 1}^{\lambda_1} \mathfrak{q}_{\omega+1-j}^{\lambda_j^{t}}  \prod_{j =1}^{\lambda_1}  \frac{Y_{\omega+2-j}(x + m(\lambda_j^t-j+1) + \epsilon(2-j) ) }{Y_{\omega+1-j}(x + m(\lambda_j^t-j) + \epsilon(1-j) )} Y_{\omega +1 - \lambda_1} (x - m \lambda_1 + \epsilon(1- \lambda_1)) 
\end{equation}
 Opening the brackets in the LHS of the formula (\ref{factorizationform}) one obtains:
 \begin{eqnarray}
     \sum_{r , s\geq 0} \sum_{\substack{ n_0 > n_1 >...>n_{r-1} \geq 1 \\ 0 \leq k_0 < k_1<...< k_{s-1}}} (-z)^{\frac{r-s}{N}} \prod_{i=0}^{r-1} \mathfrak{q}_{\omega-i}^{n_i} \prod_{i=0}^{s-1} \prod_{j=0}^{k_i - 1} \mathfrak{q}_{\omega+i-j -r} \, e_\omega \otimes e_{\omega -r +s}^t \cdot \\
     \cdot 
\prod_{i =0 }^{r-1} \frac{Y_{\omega+1-i} (x+ n_i m - (i-1) \epsilon ) }{Y_{\omega-i} (x+ (n_i -1 )m - i \epsilon )} Y_{\omega+1-r}(x-(r-1) \epsilon) \cdot \\ \cdot \prod_{i = 0}^{s-1} \frac{Y_{\omega +1 +i - k_i -r} (x - (k_i +1)m - (r+k_i -i -1) \epsilon)}{Y_{\omega +1 +i - k_i -r} (x - k_i m - (r+k_i - i -1) \epsilon)} e^{- \epsilon (r-s) \partial_x}
 \end{eqnarray}
 The further idea is to notice that the two sets of strictly increasing numbers $n_0 > n_1 >...>n_{r-1} \geq 1$ and $0 \leq k_0 < k_1<...< k_{s-1}$ precisely encode the information about a Young diagram $\lambda$ and an additional integer value, which could be interpreted as a shift of the Young diagram perpendicular to the main diagonal. The dictionary is the following. The shift is equal to $p = r-s$. The positive integers $n_j$ define the lengths of the first $r$-columns, and $k_i$ define the length of the first $s$-rows, through the formulas:
 \begin{gather}
     n_j = \lambda_{j+1}^t -j +p, \qquad j = 0,...,r-1 \\
     k_{s-i} = \lambda_i - i - p , \qquad i = 1,...,s
 \end{gather}
 This data uniquely determines the diagram. Notice that, given $\lambda$ and $p$ the numbers $r$ and $s$ are uniquely determined as such values of $i$ and $j$ that the expressions $\lambda_i - i - p $, $\lambda_{j+1}^t -j +p$ change sign at them correspondingly (last value of the index for which they are positive).
 With this substitution the above expression could be rewritten as follows:
 \begin{eqnarray}
     \sum_{p \in \mathbb{Z} } \sum_{\lambda} (-z)^{\frac{p}{N}} \prod_{j=0}^{r-1} \mathfrak{q}_{\omega-j}^{\lambda_{j+1}^t -j +p} \prod_{i=1}^{s} \prod_{j=0}^{\lambda_i - i - p - 1} \mathfrak{q}_{\omega-i-j -p} \, e_\omega \otimes e_{\omega - p}^t \cdot \\
     \cdot 
\prod_{j =0 }^{r-1} \frac{Y_{\omega+1-j} (x+ (\lambda_{j+1}^t -j +p) m - (j-1) \epsilon ) }{Y_{\omega-j} (x+ (\lambda_{j+1}^t -j +p -1 )m - j \epsilon )} Y_{\omega+1-r}(x-(r-1) \epsilon) \cdot \\
\cdot \prod_{i = 1}^{s} \frac{Y_{\omega +1 - \lambda_i} (x - (\lambda_i - i - p+1)m - (\lambda_i -1) \epsilon)}{Y_{\omega +1 - \lambda_i} (x - (\lambda_i - i - p) m - (\lambda_i -1) \epsilon)} e^{- \epsilon p \partial_x}
 \end{eqnarray} 
 Now we need to match every multiple in the product to every multiple in the expression (\ref{qChar}), shifted by $p$.\\
 Let us denote:
 \begin{multline}
 \text{LHS}(\lambda,p) = \prod_{j =0 }^{r-1} \frac{Y_{\omega+1-j} (x+ (\lambda_{j+1}^t -j +p) m - (j-1) \epsilon ) }{Y_{\omega-j} (x+ (\lambda_{j+1}^t -j +p -1 )m - j \epsilon )} Y_{\omega+1-r}(x-(r-1) \epsilon) \cdot \\
\cdot \prod_{i = 1}^{s} \frac{Y_{\omega +1 - \lambda_i} (x - (\lambda_i - i - p+1)m - (\lambda_i -1) \epsilon)}{Y_{\omega +1 - \lambda_i} (x - (\lambda_i - i - p) m - (\lambda_i -1) \epsilon)}   
 \end{multline}
 And:
 \begin{multline}
 \text{RHS}(\lambda,p) =  Y_{\omega +1 - \lambda_1} (x - m (\lambda_1 - p) + \epsilon(1- \lambda_1))  \prod_{\substack{j =1 } }^{\lambda_1} \frac{Y_{\omega+2-j}(x + m(\lambda_j^t-j+1+p) + \epsilon(2-j) ) }{Y_{\omega+1-j}(x + m(\lambda_j^t-j+p) + \epsilon(1-j) )}     
 \end{multline}

 We are going to prove that $\text{LHS}(\lambda,p) = \text{RHS}(\lambda,p)$ by induction on the number of boxes in the Young diagram.\\
 The base of the induction is the case when $\lambda = \varnothing$,  and either $r = 0$, and hence $s = -p$, or $s=0$, and $r=p$.\\
 Let us consider the case $r = 0$ first.
The LHS$(\varnothing,-s)$ of the formula above then takes the form:
\begin{equation}
Y_{\omega+1}(x+ \epsilon)
\cdot \prod_{i = 1}^{s} \frac{Y_{\omega +1} (x +(i + p-1)m + \epsilon)}{Y_{\omega +1} (x +(i + p)m + \epsilon)}    
\end{equation}
which is, after canceling all factors, equal to the RHS$(\varnothing,-s)$:
\begin{equation}
Y_{\omega +1} (x + m p + \epsilon)    
\end{equation}
Now let $s = 0$, then one has:
\begin{multline}
\text{LHS}(\varnothing,r) =  \prod_{j =0 }^{r-1} \frac{Y_{\omega+1-j} (x+ (r-j) m - (j-1) \epsilon ) }{Y_{\omega-j} (x+ (r -j -1 )m - j \epsilon )} Y_{\omega+1-r}(x-(r-1) \epsilon)  =
Y_{\omega+1-r}(x + r m+ \epsilon)
\end{multline}
which is equal to the RHS$(\varnothing,r)$.\\
For the induction step, let us assume, that we are adding one box to the $k$'th row. For the LHS the cases $k - \lambda_k -1 +p \geq 0$ and $k - \lambda_k -1 +p < 0$ should be treated separately, because they affect the product of the first $r$ factors or the last $s$ factors correspondingly, but eventually the final result is the same: 
\begin{multline}
    \frac{\text{LHS}(\lambda+1_k,p)}{\text{LHS}(\lambda,p)} = \\ =\frac{Y_{\omega - \lambda_k}(x-(\lambda_k -k -p +2)m - \lambda_k \epsilon)}{Y_{\omega+1 - \lambda_k}(x-(\lambda_k -k -p +1)m - (\lambda_k-1) \epsilon)} \frac{Y_{\omega +1- \lambda_k}(x-(\lambda_k -k -p)m - (\lambda_k-1) \epsilon)}{Y_{\omega - \lambda_k}(x-(\lambda_k -k -p +1)m - \lambda_k \epsilon)}
\end{multline}
By analogous calculation same is for the RHS. Hence the formula is proven.\\
Now it is only left to show, that $\mathfrak{q}_\omega$ factors match. Let us notice, that the RHS of the main formula could be rewritten as:
\begin{multline}
\sum_{n = 0}^\infty (-z)^{\frac{n}{N}} \hat{\chi}(x+nm) \prod_{k = 0}^{n-1} (\hat{\slashed{Q}}^{n-k} C^{-1}) \, e^{-n \epsilon \partial_x} 
+ \\ +\sum_{n = 1}^\infty (-z)^{\frac{-n}{N}} \hat{\chi}(x-nm) \prod_{k=1}^{n-1} C^k (\hat{\slashed{Q}}C^{-1})^k \, C^n \, e^{n \epsilon \partial_x}    = \\=
\sum_{n \in \mathbb{Z}} (-z)^{\frac{n}{N}} \hat{\chi}(x+nm) D_n C^{-n}  e^{-n \epsilon \partial_x} 
\end{multline}
Where $D_n$ is a diagonal matrix:
\begin{equation}
    D_n = \begin{cases}
        \prod_{k = 0}^{n-1} (\hat{\slashed{Q}}^{n-k} C^{-1}) \, C^n , \qquad  n \geq 0 \\
        \prod_{k=1}^{n-1} C^k (\hat{\slashed{Q}}C^{-1})^k \qquad n< 0
    \end{cases}
\end{equation}
Explicitly, it looks like:
\begin{equation}
    D_n = \begin{cases}
        \text{diag}(\prod_{k = 0}^{n-1} \mathfrak{q}_{\omega-k}^{n-k})_{\omega =0}^{\omega=N-1}, \qquad  n \geq 0 \\
        \text{diag}(\prod_{k = 1}^{-n-1} \mathfrak{q}_{\omega+k}^{-n-k})_{\omega =0}^{\omega=N-1}, \qquad  n < 0
    \end{cases}
\end{equation}
Hence, we need to compare the expressions:
\begin{equation}
    \mathfrak{q}(\lambda,p) = \prod_{j=0}^{r-1} \mathfrak{q}_{\omega-j}^{\lambda_{j+1}^t -j +p} \prod_{i=1}^{s} \prod_{j=0}^{\lambda_i - i - p - 1} \mathfrak{q}_{\omega-i-j -p}
\end{equation}
with
\begin{equation}
  \mathfrak{q}'(\lambda,p) = \begin{cases}
     \prod_{j = 1}^{\lambda_1} \mathfrak{q}_{\omega+1-j}^{\lambda_j^{t}} \prod_{k = 0}^{p-1} \mathfrak{q}_{\omega-k}^{p-k}, \qquad  p \geq 0 \\
     \prod_{j = 1}^{\lambda_1} \mathfrak{q}_{\omega+1-j}^{\lambda_j^{t}} \prod_{k = 1}^{-p-1} \mathfrak{q}_{\omega+k}^{-p-k}, \qquad  p < 0 
    \end{cases}
\end{equation}
The proof, again goes by induction. Let us perform the step of the induction first. If we add the box to the $k$'th row, regardless of whether $k - \lambda_k -1 +p \geq 0$ or not, the expression changes as follows:
\begin{equation}
    \frac{\mathfrak{q}(\lambda+1_k,p)}{\mathfrak{q}(\lambda,p)} = \mathfrak{q}_{\omega-\lambda_k} 
\end{equation}
The same is true for $\mathfrak{q}'(\lambda,p)$ :
\begin{equation}
    \frac{\mathfrak{q}'(\lambda+1_k,p)}{\mathfrak{q}'(\lambda,p)} = \mathfrak{q}_{\omega-\lambda_k} 
\end{equation}
So, the step is completed. Now we only need to prove the statement for the empty diagram. Here we need to consider 2 cases. The first case, when $p$ is positive, so $s=0$, and $p=r$. Then we have:
\begin{equation}
\mathfrak{q}(\varnothing,p) = \prod_{j=0}^{p-1} \mathfrak{q}_{\omega-j}^{ -j +p} =    \mathfrak{q}'(\varnothing,p)
\end{equation}
The second case,  $p$ is negative, so $r=0$, $s=-p$. Then one can write:
\begin{equation}
\mathfrak{q}(\varnothing,p)  =  \prod_{i=1}^{-p} \prod_{j=0}^{ - i - p - 1} \mathfrak{q}_{\omega-i-j -p} = \prod_{j=1}^{-p-1} \mathfrak{q}_{\omega-j -p}^{j-1} = \prod_{k = 1}^{-p-1} \mathfrak{q}_{\omega+k}^{-p-k} = \mathfrak{q}'(\varnothing,p)   
\end{equation}
where in the first step we switched the order of the product between $i$ and $j$. This completes the proof of the formula. 
\end{proof}

\section{Proof of the transformation property of the factorization formula}
\begin{proof}
    The proof is based on simple identity:
    \begin{equation}
        X \hat{C}_{\mathfrak{q} z}^{-1} X^{-1} = \slashed{Q} \hat{C}_{ z}^{-1} 
    \end{equation}
    and hence:
    \begin{equation}
     X \hat{C}_{\mathfrak{q} z} X^{-1} = \hat{C}_{ z} \slashed{Q}^{-1}   
    \end{equation}
    The idea is to consider the expression:
    \begin{equation}
      X \hat{\mathfrak{D}}(x+m,\mathfrak{q} z) X^{-1}  
    \end{equation}
    and treat every multiple in the product separately. Let us start from the left. The typical multiple will look like:
    \begin{eqnarray}
        X \Big(1-  \hat{\slashed{Q}}^n \hat{Y}(x+(n+1)m +\epsilon) \, \hat{C}_{\mathfrak{q} z}^{-1} \hat{Y}(x+nm +\epsilon)^{-1} \Big) X^{-1} =\\ = \Big(1-  \hat{\slashed{Q}}^{n+1} \hat{Y}(x+(n+1)m +\epsilon) \, \hat{C}_{z}^{-1} \hat{Y}(x+nm +\epsilon)^{-1} \Big)
    \end{eqnarray}
    Hence, we obtained all multiples in the first product starting from $n=2$. To get $n=1$ one we need to consider the term:
    \begin{eqnarray}
     X \, \hat{Y}(x + m+ \epsilon)  \Big(1 - \frac{\hat{Y}(x + \epsilon)}{\hat{Y}(x+ m + \epsilon)} \hat{C}_{\mathfrak{q} z}\Big) X^{-1} = \\
     - X \Big(1-   \hat{Y}(x+m +\epsilon) \, \hat{C}_{\mathfrak{q} z}^{-1} \hat{Y}(x +\epsilon)^{-1} \Big) X^{-1}\, \hat{Y}(x + \epsilon) \,X \hat{C}_{\mathfrak{q} z} X^{-1} = \\=
     -  \Big(1-   \hat{\slashed{Q}} \hat{Y}(x+m +\epsilon) \, \hat{C}_{z}^{-1} \hat{Y}(x +\epsilon)^{-1} \Big) \, \hat{Y}(x + \epsilon) \, X \hat{C}_{\mathfrak{q} z} X^{-1}
    \end{eqnarray}
    Now we are left to deal with the expressions of the form:
    \begin{eqnarray}
        X \hat{C}_{\mathfrak{q} z} \Big(1 -  (\hat{\slashed{Q}} \hat{C}_{\mathfrak{q} z}^{-1})^n \frac{\hat{Y}(x-nm + \epsilon)}{\hat{Y}(x-(n-1) m + \epsilon)} \hat{C}_{\mathfrak{q} z}^{\, n+1}\Big) X^{-1}
    \end{eqnarray}
    the first thing we need to notice is that:
    \begin{equation}
        (\hat{C}_{\mathfrak{q} z}^{-1} \hat{\slashed{Q}})^{n-1} \hat{C}_{\mathfrak{q} z}^{-1}\frac{\hat{Y}(x-nm + \epsilon)}{\hat{Y}(x-(n-1) m + \epsilon)} \hat{C}_{\mathfrak{q} z}^{\, n}
    \end{equation}
    is a diagonal matrix, so it commutes with $\hat{\slashed{Q}}$, so we have:
    \begin{eqnarray}
        X \hat{C}_{\mathfrak{q} z} \Big(1 -  (\hat{\slashed{Q}} \hat{C}_{\mathfrak{q} z}^{-1})^n \frac{\hat{Y}(x-nm + \epsilon)}{\hat{Y}(x-(n-1) m + \epsilon)} \hat{C}_{\mathfrak{q} z}^{\, n+1}\Big) X^{-1} = \\ =
        X \hat{C}_{\mathfrak{q} z} \Big(1 -  (\hat{C}_{\mathfrak{q} z}^{-1} \hat{\slashed{Q}})^{n-1} \hat{C}_{\mathfrak{q} z}^{-1}\frac{\hat{Y}(x-nm + \epsilon)}{\hat{Y}(x-(n-1) m + \epsilon)} \hat{C}_{\mathfrak{q} z}^{\, n} \hat{\slashed{Q}} \hat{C}_{\mathfrak{q} z }\Big)  X^{-1}
    \end{eqnarray}
    pulling $\hat{C}_{\mathfrak{q} z }$ through to the right, one gets:
    \begin{eqnarray}
     X  \Big(1 -  (\hat{\slashed{Q}} \hat{C}_{\mathfrak{q} z}^{-1})^{n-1}\frac{\hat{Y}(x-nm + \epsilon)}{\hat{Y}(x-(n-1) m + \epsilon)} \hat{C}_{\mathfrak{q} z}^{\, n} \hat{\slashed{Q}} \Big)  X^{-1} \, X\hat{C}_{\mathfrak{q} z} X^{-1}    
    \end{eqnarray}
    Carrying out the conjugation by $X$, we arrive at:
    \begin{eqnarray}
     X  \Big(1 -  (\hat{\slashed{Q}}^2 \hat{C}_{ z}^{-1})^{n-1}\frac{\hat{Y}(x-nm + \epsilon)}{\hat{Y}(x-(n-1) m + \epsilon)} (\hat{C}_{z} \hat{\slashed{Q}}^{-1})^{\, n} \hat{\slashed{Q}} \Big)  X^{-1} \, X\hat{C}_{\mathfrak{q} z} X^{-1} = \\ =    \Big(1 -  (\hat{\slashed{Q}}^2 \hat{C}_{z}^{-1})^{n-1}\frac{\hat{Y}(x-nm + \epsilon)}{\hat{Y}(x-(n-1) m + \epsilon)} (\hat{C}_{ z} \hat{\slashed{Q}}^{-1})^{\, n-1} \hat{C}_{ z} \Big)  \, X\hat{C}_{\mathfrak{q} z} X^{-1}   
    \end{eqnarray}
    Now in each factor $\hat{\slashed{Q}}^2 $ on the left, one multiple got canceled with the factor $\hat{\slashed{Q}}^{-1}$ on the right, because the matrix, standing between them is diagonal. So one gets:
    \begin{eqnarray}
       \Big(1 -  (\hat{\slashed{Q}} \hat{C}_{z}^{-1})^{n-1}\frac{\hat{Y}(x-nm + \epsilon)}{\hat{Y}(x-(n-1) m + \epsilon)} \hat{C}_{ z}^{\, n} \Big)  \, X\hat{C}_{\mathfrak{q} z} X^{-1}   
    \end{eqnarray}
    That is exactly what we need. So the proof is completed. 
\end{proof}
However, for completeness, we decided to show, that the RHS
\begin{eqnarray}
\hat{\mathfrak{D}}(x,z) = \sum_{n = 0}^\infty (-1)^n \hat{\chi}(x+nm) \prod_{k = 0}^{n-1} (\hat{\slashed{Q}}^{n-k} \hat{C}_z^{-1})  + \\
+ \sum_{n = 1}^\infty (-1)^n \hat{\chi}(x-nm) \prod_{k=1}^{n-1} \hat{C}_z^k (\hat{\slashed{Q}}\hat{C}_z^{-1})^k \, \hat{C}_z^n    
\end{eqnarray}
also have the same transformation property. 
\begin{proof}
    Indeed, for the positive shift summands one has:
    \begin{equation}
        X \prod_{k = 0}^{n-1} (\hat{\slashed{Q}}^{n-k} \hat{C}_{\mathfrak{q} z}^{-1}) X^{-1} = \prod_{k = 0}^{n} (\hat{\slashed{Q}}^{n+1-k} \hat{C}_{z}^{-1}) \, \hat{C}_{z} \hat{\slashed{Q}}^{-1}
    \end{equation}
    exactly, as we need. \\
    Now let us consider the negative shift summands:
    \begin{equation}
        X \prod_{k=1}^{n-1} \hat{C}_{\mathfrak{q} z}^k (\hat{\slashed{Q}}\hat{C}_{\mathfrak{q} z}^{-1})^k \, \hat{C}_{\mathfrak{q} z}^n X^{-1}
    \end{equation}
    Let us deal with each multiple of the form:
    \begin{equation}
      X  \hat{C}_{\mathfrak{q} z}^k (\hat{\slashed{Q}}\hat{C}_{\mathfrak{q} z}^{-1})^k  X^{-1}  
    \end{equation}
    first.
    \begin{eqnarray}
    X  \hat{C}_{\mathfrak{q} z}^k (\hat{\slashed{Q}}\hat{C}_{\mathfrak{q} z}^{-1})^k  X^{-1}  = (\hat{C}_{z} \hat{\slashed{Q}}^{-1})^k (\hat{\slashed{Q}}^2\hat{C}_{z}^{-1})^k    
    \end{eqnarray}
    Similarly to the previous proof one of the $\hat{\slashed{Q}}$ factors in $\hat{\slashed{Q}}^2$ on the far right got canceled by the first $\hat{\slashed{Q}}^{-1}$ on the left, because the number of $\hat{C}_{z}$ between them equal to the number of $\hat{C}_{z}^{-1}$. The same holds for all other $\hat{\slashed{Q}}$. So we get:
    \begin{eqnarray}
    X  \hat{C}_{\mathfrak{q} z}^k (\hat{\slashed{Q}}\hat{C}_{\mathfrak{q} z}^{-1})^k  X^{-1}  = (\hat{C}_{z} \hat{\slashed{Q}}^{-1})^k (\hat{\slashed{Q}}^2\hat{C}_{z}^{-1})^k  =    \hat{C}_{z}^k (\hat{\slashed{Q}}\hat{C}_{z}^{-1})^k
    \end{eqnarray}
    Now let us look at the last factor ($k=n-1$) and the $\hat{C}_{\mathfrak{q} z}^n$ term, standing after it:
    \begin{eqnarray}
    \hat{C}_{z}^{n-1} (\hat{\slashed{Q}}\hat{C}_{z}^{-1})^{n-1} (\hat{C}_{z} \hat{\slashed{Q}}^{-1})^n =\hat{C}_{z}^{n-1} \, \hat{C}_{z} \hat{\slashed{Q}}^{-1}    
    \end{eqnarray}
    That is exactly what we needed.
\end{proof}
\begin{proof}[Proof of the theorem \ref{6dtransform}]
   The proof follows directly from the RHS of the factorization formula and formulas (\ref{qqCharCl}) and (\ref{SzId}). 
\end{proof}

\section{Proof of the expression for the Lax matrix of the Calogero-Moser model}

\begin{proof}[Proof of the theorem \ref{Laxthm}]
Both $\mathfrak{D}_0(z)$ and $\mathfrak{D}_1(z)$ are already sections of the bundles on just the elliptic curve. $\hat{\mathfrak{D}}_1(z)$ - section of the line bundle, and $\mathfrak{D}_0(z)$ - section of the affine bundle. And their transformation properties look as follows:
\begin{gather}
X \mathfrak{D}_1(\mathfrak{q} z) = -\mathfrak{D}_1(z) X C_{\mathfrak{q} z} \\
X \mathfrak{D}_0(\mathfrak{q} z) + m X \mathfrak{D}_1(\mathfrak{q} z)  = -\mathfrak{D}_0(z) X C_{\mathfrak{q} z}
\end{gather}
Hence if we define the Lax operator by the formula:
\begin{equation}
    \tilde{L}(z) = - \mathfrak{D}_0(z) \, \mathfrak{D}_1(z)^{-1}
\end{equation}
It will transform as follows:
\begin{equation} \label{LaxTransform}
    X \tilde{L}(\mathfrak{q} z) X^{-1} - m = \tilde{L}(z) 
\end{equation}
Let us write down a more explicit expression for it.
The strategy is to read off the values of $\mathfrak{D}_0(z)$ and $\mathfrak{D}_1(z)$ from the LHS of the formula (\ref{factorizationCl}), and then, divide one by another.  At first, we need to expand the $\hat{Y}$ matrix up to constant order in $x$:
\begin{equation}
    \hat{Y}(x) = x - P + O \big(\frac{1}{x} \big)
\end{equation}
The factors appearing in the product in (\ref{factorizationCl}) then will have the expansion:
\begin{gather}
\frac{\hat{Y}(x-(n+1)m)}{\hat{Y}(x-n m)} = 1 - \frac{m}{x} + O \big(\frac{1}{x^2} \big)  \\
\hat{Y}(x+nm ) \, C_{z}^{-1} \hat{Y}(x+(n-1)m )^{-1} = C_{z}^{-1} + \frac{m}{x} C_{z}^{-1} + \frac{1}{x} [C_{z}^{-1}, P] + O \big(\frac{1}{x^2} \big)
\end{gather}
where $P$ - is the matrix of classical momenta, conjugated to $x_\omega$.
Then we can write down $\mathfrak{D}_0(z)$ and $\mathfrak{D}_1(z)$ explicitly:
\begin{equation}
 \mathfrak{D}_1(z) =  \overleftarrow{\prod_{n=1}^\infty} \Big(1-  \hat{\slashed{Q}}^n  \, C_{z}^{-1} \Big) 
\cdot \overrightarrow{\prod_{n=0}^\infty} \Big(1 -   (\hat{\slashed{Q}} C_{z}^{-1})^n C_{z}^{\, n+1}\Big)  
\end{equation}
\begin{multline}
 \mathfrak{D}_0(z) = - \overleftarrow{\prod_{n=1}^\infty} \Big(1-  \hat{\slashed{Q}}^n  \, C_{z}^{-1} \Big) 
\, P \, \overrightarrow{\prod_{n=0}^\infty} \Big(1 -   (\hat{\slashed{Q}} C_{z}^{-1})^n C_{z}^{\, n+1}\Big)  - \\
- \sum_{k=1}^\infty \overleftarrow{\prod_{n=k+1}^\infty} \Big(1-  \hat{\slashed{Q}}^n  \, C_{z}^{-1} \Big) \hat{\slashed{Q}}^k \Big(m C_{z}^{-1} + [C_{z}^{-1}, P]  \Big) \overleftarrow{\prod_{n=1}^{k-1}} \Big(1-  \hat{\slashed{Q}}^n  \, C_{z}^{-1} \Big) \cdot \overrightarrow{\prod_{n=0}^\infty} \Big(1 -   (\hat{\slashed{Q}} C_{z}^{-1})^n C_{z}^{\, n+1}\Big) + \\
+ m \sum_{k=0}^\infty \overleftarrow{\prod_{n=1}^\infty} \Big(1-  \hat{\slashed{Q}}^n  \, C_{z}^{-1} \Big) \cdot \overrightarrow{\prod_{n=0}^{k-1}} \Big(1 -   (\hat{\slashed{Q}} C_{z}^{-1})^n C_{z}^{\, n+1}\Big) \cdot (\hat{\slashed{Q}} C_{z}^{-1})^k C_{z}^{\, k+1} \cdot \overrightarrow{\prod_{n=k+1}^{\infty}} \Big(1 -   (\hat{\slashed{Q}} C_{z}^{-1})^n C_{z}^{\, n+1}\Big)
\end{multline}
Hence for the Lax matrix, we get:
\begin{multline}
    \tilde{L}(z) = \overleftarrow{\prod_{n=1}^\infty} \Big(1-  \hat{\slashed{Q}}^n  \, C_{z}^{-1} \Big) \, P \overrightarrow{\prod_{n=1}^\infty} \Big(1-  \hat{\slashed{Q}}^n  \, C_{z}^{-1} \Big)^{-1} +\\
    + \sum_{k=1}^\infty \overleftarrow{\prod_{n=k+1}^\infty} \Big(1-  \hat{\slashed{Q}}^n  \, C_{z}^{-1} \Big) \hat{\slashed{Q}}^k \Big(m C_{z}^{-1} + [C_{z}^{-1}, P]  \Big) \overrightarrow{\prod_{n=k}^{\infty}} \Big(1-  \hat{\slashed{Q}}^n  \, C_{z}^{-1} \Big)^{-1}  - \\ 
    -m \sum_{k=0}^\infty \overleftarrow{\prod_{n=1}^\infty} \Big(1-  \hat{\slashed{Q}}^n  \, C_{z}^{-1} \Big) \cdot \overrightarrow{\prod_{n=0}^{k-1}} \Big(1 -   (\hat{\slashed{Q}} C_{z}^{-1})^n C_{z}^{\, n+1}\Big) \cdot (\hat{\slashed{Q}} C_{z}^{-1})^k C_{z}^{\, k+1} \cdot \\ \cdot \overleftarrow{\prod_{n=0}^{k}} \Big(1 -   (\hat{\slashed{Q}} C_{z}^{-1})^n C_{z}^{\, n+1}\Big)^{-1} \overrightarrow{\prod_{n=1}^\infty} \Big(1-  \hat{\slashed{Q}}^n  \, C_{z}^{-1} \Big)^{-1}
\end{multline}
Using the equality:
\begin{multline}
\overleftarrow{\prod_{n=1}^\infty} \Big(1-  \hat{\slashed{Q}}^n  \, C_{z}^{-1} \Big) \, P \overrightarrow{\prod_{n=1}^\infty} \Big(1-  \hat{\slashed{Q}}^n  \, C_{z}^{-1} \Big)^{-1} = \\ = P -
\sum_{k=1}^\infty \overleftarrow{\prod_{n=k+1}^\infty} \Big(1-  \hat{\slashed{Q}}^n  \, C_{z}^{-1} \Big) \hat{\slashed{Q}}^k \, [C_{z}^{-1}, P] \, \overrightarrow{\prod_{n=k}^{\infty}} \Big(1-  \hat{\slashed{Q}}^n  \, C_{z}^{-1} \Big)^{-1}
\end{multline}
one arrives at the answer:
\begin{multline} \label{LaxFormula}
    \tilde{L}(z) = P 
    + m \sum_{k=1}^\infty \overleftarrow{\prod_{n=k+1}^\infty} \Big(1-  \hat{\slashed{Q}}^n  \, C_{z}^{-1} \Big) \, \hat{\slashed{Q}}^k  C_{z}^{-1}  \, \overrightarrow{\prod_{n=k}^{\infty}} \Big(1-  \hat{\slashed{Q}}^n  \, C_{z}^{-1} \Big)^{-1} - \\ 
    -m \sum_{k=0}^\infty \overleftarrow{\prod_{n=1}^\infty} \Big(1-  \hat{\slashed{Q}}^n  \, C_{z}^{-1} \Big) \cdot \overrightarrow{\prod_{n=0}^{k-1}} \Big(1 -   (\hat{\slashed{Q}} C_{z}^{-1})^n C_{z}^{\, n+1}\Big) \cdot (\hat{\slashed{Q}} C_{z}^{-1})^k C_{z}^{\, k+1} \cdot \\ \cdot \overleftarrow{\prod_{n=0}^{k}} \Big(1 -   (\hat{\slashed{Q}} C_{z}^{-1})^n C_{z}^{\, n+1}\Big)^{-1} \overrightarrow{\prod_{n=1}^\infty} \Big(1-  \hat{\slashed{Q}}^n  \, C_{z}^{-1} \Big)^{-1}
\end{multline}
The Lax matrix of the Calogero-Moser system is uniquely determined by its quasi-periodicity properties and the residue at the point $z=1$. The quasiperiodicity was already verified, and it matches the standard ones. Hence we are only left to compute the residue.
The only contribution to this residue comes from the factor with $k=0, \, n=0$:
\begin{equation}
C_z (1-  C_{z})^{-1}   = -\frac{1}{1-z} \sum_{\omega \in \mathbb{Z}_N} C_{z}^{\omega} 
\end{equation}
Its residue is thus equal to:
\begin{equation}
    -\underset{z=1}{\text{res}} (1-  C_{z}^{-1})^{-1} =  e \otimes e^t
\end{equation}
Therefore the Lax matrix residue is equal to:
\begin{equation} \label{Laxresidue}
\underset{z=1}{\text{res}} \tilde{L}(z)  = -m  \overleftarrow{\prod_{n=1}^\infty} \Big(1-  \hat{\slashed{Q}}^n  \, C_{1}^{-1} \Big) \, e \otimes e^t \overrightarrow{\prod_{n=1}^\infty} \Big(1-  \hat{\slashed{Q}}^n  \, C_{1}^{-1} \Big)^{-1}    
\end{equation}
Let us denote:
\begin{gather}
    u = \overleftarrow{\prod_{n=1}^\infty} \Big(1-  \hat{\slashed{Q}}^n  \, C_{1}^{-1} \Big) \, e \\
    v^t = e^t \overrightarrow{\prod_{n=1}^\infty} \Big(1-  \hat{\slashed{Q}}^n  \, C_{1}^{-1} \Big)^{-1} 
\end{gather}
From the off-diagonal part of the Lax matrix transformation property (\ref{LaxTransform}) and from the residue calculation above it follows that it could be written in the following form:
\begin{equation}
\tilde{L}_{ij}(z) =  - m \, u_i v_j \frac{\theta_{\mathfrak{q}}'(1) \theta_{\mathfrak{q}}(z x_i / x_j)}{\theta_{\mathfrak{q}}(z) \theta_{\mathfrak{q}}( x_i / x_j)}  \quad \text{for} \,\, i \neq j 
\end{equation}
In order for the residue at $z=1$ to match (\ref{Laxresidue}) the diagonal part of $\tilde{L}$ should include the term $-m u_i v_i E_1(z) \delta_{ij}$. However from the diagonal part of the transformation property (\ref{LaxTransform}) and quasi-periodicity of $E_1(z)$ ($E_1( \mathfrak{q} z) = E_1(z) -1$) it follows that $v_i = u_i^{-1}$. As from the explicit expression (\ref{LaxFormula}) we know that the constant in $z$ diagonal part of $\tilde{L}(z)$ is $P$ we arrive at the final result:
\begin{equation}
\tilde{L}_{ij}(z) = \big(p_i - m E_1(z)\big) \delta_{ij} - m (1 - \delta_{ij}) \, u_i u_j^{-1} \frac{\theta_{\mathfrak{q}}'(1) \theta_{\mathfrak{q}}(z x_i / x_j)}{\theta_{\mathfrak{q}}(z) \theta_{\mathfrak{q}}( x_i / x_j)}   
\end{equation}
\\ 
To get to the standard value of the residue, we thus need to conjugate the Lax operator by the diagonal matrix $U$, which is determined by the requirement:
\begin{equation}
\overleftarrow{\prod_{n=1}^\infty} \Big(1-  \hat{\slashed{Q}}^n  \, C_{1}^{-1} \Big) \, e = U \,e     
\end{equation}
It takes the form:
\begin{equation}
    U_{\omega \, \omega'} = u_\omega \delta_{\omega \, \omega'} =  \delta_{\omega \, \omega'} \sum_{s=0}^{\infty}(-1)^s \sum_{0< n_1 <...< n_s} \, \prod_{k=0}^{s-1} \mathfrak{q}_{\omega-k}^{n_{s-k}}
\end{equation}
So the Lax matrix of the Calogero-Moser system in the standard gauge is equal to:
\begin{equation}
    L(z) = -U^{-1} \, \mathfrak{D}_0(z) \, \mathfrak{D}_1(z)^{-1} U
\end{equation} 
and has the following expression:
\begin{equation}
L_{ij}(z) = \big(p_i - m E_1(z)\big) \delta_{ij} - m (1 - \delta_{ij}) \, \frac{\theta_{\mathfrak{q}}'(1) \theta_{\mathfrak{q}}(z x_i / x_j)}{\theta_{\mathfrak{q}}(z) \theta_{\mathfrak{q}}( x_i / x_j)}   
\end{equation}
\end{proof}

\section{Proof of the expression for the elliptic Ruijsenaars-Schneider Lax matrix}
\begin{proof}[Proof of the theorem \ref{RSLaxthm}]
From the transformation properties for $\mathfrak{D}_0(z)$ and $\mathfrak{D}_\infty(z)$:
\begin{equation}
    X \mathfrak{D}_\infty(\mathfrak{q} z) = - e^{\beta m} \mathfrak{D}_\infty (z) X C_{\mathfrak{q} z} 
\end{equation}
we immediately see, that the Lax matrix:
\begin{equation} \label{RSLaxPeriodicity}
    \tilde{L}^{RS}(z) = \mathfrak{D}_\infty(z) \mathfrak{D}_1^{-1}(z)
\end{equation}
have the needed quasi-periodicity:
\begin{equation}
    X \tilde{L}^{RS}(\mathfrak{q} z) X^{-1} = e^{\beta m} \tilde{L}^{RS}(z) 
\end{equation}
Now, let us find an explicit formula for it. We have:
\begin{equation}
    \frac{\hat{Y}(x-(n+1)m)}{\hat{Y}(x-n m)} = 1 - e^{- \beta x} e^{\beta P} e^{\beta n m} (e^{\beta m} - 1) + O(e^{-2 \beta x})
\end{equation}
\begin{multline}
    \hat{Y}(x+nm ) \, C_{z}^{-1} \hat{Y}(x+(n-1)m )^{-1} = \\ = C_z^{-1} - e^{- \beta x} e^{- \beta n m} (e^{\beta P} C_z^{-1} - e^{\beta m} C_z^{-1} e^{\beta P}) + O(e^{-2 \beta x})
\end{multline}
Hence for $\mathfrak{D}_\infty (z)$ we will have the following expression:
\begin{multline}
 \mathfrak{D}_\infty(z) = \overleftarrow{\prod_{n=1}^\infty} \Big(1-  \hat{\slashed{Q}}^n  \, C_{z}^{-1} \Big) 
\, e^{\beta P}\, \overrightarrow{\prod_{n=0}^\infty} \Big(1 -   (\hat{\slashed{Q}} C_{z}^{-1})^n C_{z}^{\, n+1}\Big)  - \\
- \sum_{k=1}^\infty \overleftarrow{\prod_{n=k+1}^\infty} \Big(1-  \hat{\slashed{Q}}^n  \, C_{z}^{-1} \Big) \hat{\slashed{Q}}^k e^{- \beta m k} \Big(e^{\beta P} C_{z}^{-1} - e^{\beta m} C_{z}^{-1} e^{\beta P} \Big) \overleftarrow{\prod_{n=1}^{k-1}} \Big(1-  \hat{\slashed{Q}}^n  \, C_{z}^{-1} \Big) \cdot \overrightarrow{\prod_{n=0}^\infty} \Big(1 -   (\hat{\slashed{Q}} C_{z}^{-1})^n C_{z}^{\, n+1}\Big) - \\-
 (e^{\beta m} - 1) \sum_{k=0}^\infty \overleftarrow{\prod_{n=1}^\infty} \Big(1-  \hat{\slashed{Q}}^n  \, C_{z}^{-1} \Big) \cdot \overrightarrow{\prod_{n=0}^{k-1}} \Big(1 -   (\hat{\slashed{Q}} C_{z}^{-1})^n C_{z}^{\, n+1}\Big) \cdot \\ \cdot e^{\beta (P+ m k)} (\hat{\slashed{Q}} C_{z}^{-1})^k C_{z}^{\, k+1} \cdot \overrightarrow{\prod_{n=k+1}^{\infty}} \Big(1 -   (\hat{\slashed{Q}} C_{z}^{-1})^n C_{z}^{\, n+1}\Big)
\end{multline}
By calculations analogous to the ones done in the 4d case  for the Lax matrix we will have the result:
\begin{multline} \label{LaxRSFormula}
    \tilde{L}^{RS}(z) = e^{\beta P} 
    + \\ + \sum_{k=1}^\infty \overleftarrow{\prod_{n=k+1}^\infty} \Big(1-  \hat{\slashed{Q}}^n  \, C_{z}^{-1} \Big) \, \hat{\slashed{Q}}^k  \Big[e^{\beta P} C_z^{-1} (1- e^{- \beta m k}) - (1-e^{- \beta m k + \beta m}) C_z^{-1} e^{\beta P} \Big] \cdot \overrightarrow{\prod_{n=k}^{\infty}} \Big(1-  \hat{\slashed{Q}}^n  \, C_{z}^{-1} \Big)^{-1} - \\ 
    - (e^{\beta m} -1) \sum_{k=0}^\infty \overleftarrow{\prod_{n=1}^\infty} \Big(1-  \hat{\slashed{Q}}^n  \, C_{z}^{-1} \Big) \cdot \overrightarrow{\prod_{n=0}^{k-1}} \Big(1 -   (\hat{\slashed{Q}} C_{z}^{-1})^n C_{z}^{\, n+1}\Big) \cdot e^{\beta P} e^{\beta m k} (\hat{\slashed{Q}} C_{z}^{-1})^k C_{z}^{\, k+1} \cdot \\ \cdot \overleftarrow{\prod_{n=0}^{k}} \Big(1 -   (\hat{\slashed{Q}} C_{z}^{-1})^n C_{z}^{\, n+1}\Big)^{-1} \overrightarrow{\prod_{n=1}^\infty} \Big(1-  \hat{\slashed{Q}}^n  \, C_{z}^{-1} \Big)^{-1}
\end{multline}
Its residue at $z =1$ equals:
\begin{equation} \label{LaxRSresidue}
\underset{z=1}{\text{res}} 
\tilde{L}^{RS}(z)  = (1-e^{\beta m})  \overleftarrow{\prod_{n=1}^\infty} \Big(1-  \hat{\slashed{Q}}^n  \, C_{1}^{-1} \Big) \, e^{\beta P} e \otimes e^t \overrightarrow{\prod_{n=1}^\infty} \Big(1-  \hat{\slashed{Q}}^n  \, C_{1}^{-1} \Big)^{-1}    
\end{equation}
Or, if we denote:
\begin{gather}
    u^{RS} = u^{RS}(\slashed{Q},P) =  \overleftarrow{\prod_{n=1}^\infty} \Big(1-  \hat{\slashed{Q}}^n  \, C_{1}^{-1} \Big) \,e^{\beta P} \, \cdot e \\
    v^t = e^t \overrightarrow{\prod_{n=1}^\infty} \Big(1-  \hat{\slashed{Q}}^n  \, C_{1}^{-1} \Big)^{-1} 
\end{gather}
then
\begin{equation}
\underset{z=1}{\text{res}} 
\tilde{L}^{RS}(z)  = (1-e^{\beta m}) \,  u^{RS} \otimes v
\end{equation}
Similar to the 4d case we could express the answer for $L^{RS}$ in terms of elliptic functions.
Indeed, from Lax quasi-periodicity property (\ref{RSLaxPeriodicity}) we conclude that
\begin{equation}
    \tilde{L}^{RS}_{ij}(z) = c_{ij} \frac{\theta_{\mathfrak{q}}'(1) \theta_{\mathfrak{q}}(e^{-\beta m} z x_i /x_j)}{\theta_{\mathfrak{q}}(z) \theta_{\mathfrak{q}}(e^{-\beta m} x_i /x_j)}
\end{equation}
where $c_{ij}$ depend on $\vec{x}$ and $\vec{p}$.\\
By comparison of the residues at $z = 1$ one gets:
\begin{equation}
    c_{ij} = (1-e^{\beta m}) u^{RS}_i v_j
\end{equation}
Hence, after conjugation by $U$ we arrive at the final formula:
\begin{equation}
    L^{RS}_{ij}(z) = (1-e^{\beta m}) u^{RS}_i u_i^{-1} \frac{\theta_{\mathfrak{q}}'(1) \theta_{\mathfrak{q}}(e^{-\beta m} z x_i /x_j)}{\theta_{\mathfrak{q}}(z) \theta_{\mathfrak{q}}(e^{-\beta m} x_i /x_j)}
\end{equation}
\end{proof}

\section{Proof of the main theorem in the Trigonometric Limit section}

 \begin{proof}[Proof of the theorem \ref{Trigthm}]
    Let us start by writing down the factors in the product (\ref{factorizationCl}) explicitly.\\
    First of all:
    \begin{equation}
        \hat{\slashed{Q}}^n C_{z}^{-1} = \begin{pmatrix}
            0 & 0 & 0 & ... & \mathfrak{q}_0^n z \\
            \mathfrak{q}_1^n & 0 & 0 & ... & 0 \\
            0 & \mathfrak{q}_2^n & 0 & ... & 0 \\
            \vdots & \vdots & \ddots & \ddots & \vdots \\
             0 & 0 & ... & \mathfrak{q}_{N-1}^n  & 0
        \end{pmatrix} 
    \end{equation}
    Hence:
\begin{multline}
    1-  \hat{\slashed{Q}}^n \hat{Y}(x+nm ) \, C_{z}^{-1} \hat{Y}(x+(n-1)m )^{-1}  =  \\ = \begin{pmatrix}
            1 & 0 & 0 & ... & -\mathfrak{q}_0^n z \frac{Y_{1,n}}{Y_{0,n-1}} \\
            -\mathfrak{q}_1^n \frac{Y_{2,n}}{Y_{1,n-1}} & 1 & 0 & ... & 0 \\
            0 & -\mathfrak{q}_2^n \frac{Y_{3,n}}{Y_{2,n-1}}  & 1 & ... & 0 \\
            \vdots & \vdots & \ddots & \ddots & \vdots \\
             0 & 0 & ... & -\mathfrak{q}_{N-1}^n \frac{Y_{0,n}}{Y_{N-1,n-1}}  & 1
        \end{pmatrix} 
\end{multline}
Therefore in the limit $\mathfrak{q} \rightarrow 0$, and hence $\mathfrak{q}_0 \rightarrow 0$, the matrix becomes uni-lower-triangular and $z$-independent.\\
Similarly one can obtain:
\begin{multline}
1 -   (\hat{\slashed{Q}} C_{z}^{-1})^n \frac{\hat{Y}(x-(n+1)m)}{\hat{Y}(x-n m)} C_{z}^{\, n+1} = \\
= \begin{pmatrix} 1 & -\delta_{N-1} & 0 &...& 0\\ 0 & 1 & -\delta_{N-2} &...& 0 \\ \vdots & \vdots & \ddots & \ddots & \vdots \\ 0 & 0 & 0 & ... & -\delta_1 \\ -\delta_0 & 0 & 0 & ... & 1 \end{pmatrix} 
\end{multline}
where:
\begin{equation}
    \delta_j = \mathfrak{q}_{N-n-j} \, ... \, \mathfrak{q}_{N-1-j} \frac{Y_{N-n-j, -n-1}}{Y_{N-n-j, -n}} \begin{cases}
        z^{-1}, \qquad j =0 \\
        1, \qquad 0 < j < N-n \\
        z, \qquad  N-n \leq j \leq N-1
    \end{cases}
\end{equation}
From this result, it follows, that as soon as $\mathfrak{q}_0 = 0$, the limit $z \rightarrow \infty$ of this matrix is well-defined and uni-upper-triangular. \\
We conclude from these two calculations that the matrix $\mathfrak{D}(x)^{\text{trig}}$ is also well defined and obtained from the diagonal matrix $\hat{Y}$ by the action of the uni-lower-triangular matrix from the left and uni-upper-triangular matrix from the right. This action preserves the principal minors. Hence part 3) of the theorem is proven. In particular, we have
\begin{equation}
    \det \mathfrak{D}(x)^{\text{trig}} = Y(x) = \chi(x).
\end{equation}
Together with the formulas (\ref{factorizationClSc}) and (\ref{det}), we obtain part 2) from here, as well. 
 \end{proof}

\section{Proof of the formula for the Lax eigenvector}

\textbf{Remark:} In our next paper we will give a much simpler proof of this theorem, based on general results of \cite{ncJac}.

\subsection{Warm up: SL(2) Example (first terms)}
Let us now find the equation for the vector:
\begin{equation}
    \begin{pmatrix} \chi_{24, 0}(x) \\ \chi_{24, 1}(x)
    \end{pmatrix}
\end{equation}
Put $\nu = \varnothing$ first, that is looking at the term $\mathfrak{q}_0^0 \mathfrak{q}_1^0$. One has:
\begin{gather}
    \chi_{24,0}^{(\varnothing)}(x) = \frac{\tilde{Q}_0(x+2 \epsilon_2) \tilde{Q}_1(x+ \epsilon_2)}{\tilde{Q}_0(x+2 \epsilon_2 +m ) \tilde{Q}_1(x+ \epsilon_2 +m)} \\
    \chi_{24,1}^{(\varnothing)}(x) = \frac{\tilde{Q}_0(x+ \epsilon_2) \tilde{Q}_1(x+ 2\epsilon_2)}{\tilde{Q}_0(x+ \epsilon_2 +m ) \tilde{Q}_1(x+ 2\epsilon_2 +m)}
\end{gather}
They satisfy the equation:
\begin{equation}
    \begin{pmatrix}
        Y_1(x + \epsilon_2 +m) & -Y_1(x + \epsilon_2) e^{\epsilon_2 \partial_x} \\
        -Y_0(x + \ve_2) e^{\epsilon_2 \partial_x} &  Y_0(x + \epsilon_2 +m) \\
    \end{pmatrix} 
    \begin{pmatrix} \chi_{24, 0}^{(\varnothing)}(x) \\ \chi_{24, 1}^{(\varnothing)}(x)
    \end{pmatrix} = 0
\end{equation}
Now let us find the first order in $\mathfrak{q}_1$ correction to it. For $qq$-Characters we have:
\begin{gather}
    \chi_0(x) = Y_1(x+\epsilon_2) \\
    \chi_1(x) = Y_0(x+\epsilon_2) + \mathfrak{q}_1 \frac{Y_0(x+m+\epsilon_2) Y_1(x-m)}{Y_1(x)} + O(\mathfrak{q}_1^2)
\end{gather}
And for the folded instanton partition functions:
\begin{gather}
\chi_{24,0}(x) = \chi_{24,0}^{(\varnothing)}(x) + O(\mathfrak{q}_1^2) \\
\chi_{24,1}(x) = \chi_{24,1}^{(\varnothing)}(x) +\mathfrak{q}_1 \chi_{24,1}^{(\square)}(x) + O(\mathfrak{q}_1^2)  
\end{gather}
where:
\begin{equation}
    \chi_{24,1}^{(\square)}(x) = {\bE} \Big[\frac{-q_1 q_2 P_3}{\tilde{P}_2}\big( - q_2 q_4 S_{12,0}^* + S_{12,1}^*(q_2+q_4) \big) - \frac{q_1 q_2^2 P_3}{\tilde{P}_2} \big( - q_2 q_4 S_{12,1}^* + S_{12,0}^*(q_2+q_4) \big)  \Big]
\end{equation}
Or in the multiplicative form:
\begin{equation}
\chi_{24,1}^{(\square)}(x) =\frac{\tilde{Q}_1(x+2 \epsilon_2) \tilde{Q}_1(x-m) \tilde{Q}_1(x+2 \epsilon_2) \tilde{Q}_0(x+3 \epsilon_2)}{\tilde{Q}_1(x) \tilde{Q}_1(x+2 \epsilon_2+m) \tilde{Q}_1(x+2 \epsilon_2-m) \tilde{Q}_1(x+3 \epsilon_2+m)}    
\end{equation}
Thanks to easy verifiable equations:
\begin{gather}
    Y_1(x+ \epsilon_2) \chi_{24,1}^{(\square)}(x+ \epsilon_2) = Y_1(x+ \epsilon_2 - m) \chi_{24,0}^{(\varnothing)}(x+ 2\epsilon_2) \\
    \frac{Y_0(x+\epsilon_2 +m) Y_1(x-m)}{Y_1(x)} \chi_{24,0}^{(\varnothing)}(x + \ve_2) = Y_0(x+ \epsilon_2 + m) \chi_{24,1}^{(\square)}(x) \\
    \frac{Y_0(x+\epsilon_2 +2m) Y_1(x)}{Y_1(x+m)} \chi_{24,1}^{(\varnothing)}(x) = Y_0(x + 2m + \epsilon_2) \chi_{24,0}^{(\varnothing)}(x- \epsilon_2)
\end{gather}
one obtains:
\begin{equation}
    \begin{pmatrix}
        \chi_0(x +m) + \mathfrak{q}_1 \chi_0(x-m) e^{2 \epsilon_2 \partial_x} & -\chi_0(x) e^{\epsilon_2 \partial_x} \\
        -\chi_1(x) e^{\epsilon_2 \partial_x} - \mathfrak{q}_1 \chi_1(x+2m) e^{-\epsilon_2 \partial_x}  &  \chi_1(x + m) \\
    \end{pmatrix} 
    \begin{pmatrix} \chi_{24, 0}^{(\varnothing)}(x) \\ \chi_{24, 1}^{(\varnothing)}(x) + \mathfrak{q}_1 \chi_{24, 1}^{(\square)}(x)
    \end{pmatrix} = 0
\end{equation}
Shifting $x$ by $-m$ gives us:
\begin{equation}
    \begin{pmatrix}
        \chi_0(x) + \mathfrak{q}_1 \chi_0(x-2m) e^{2 \epsilon_2 \partial_x} & -\chi_0(x-m) e^{\epsilon_2 \partial_x} \\
        -\chi_1(x-m) e^{\epsilon_2 \partial_x} - \mathfrak{q}_1 \chi_1(x+m) e^{-\epsilon_2 \partial_x}  &  \chi_1(x) \\
    \end{pmatrix} 
    \begin{pmatrix}
    \psi_0(x) \\
    \psi_1(x)
\end{pmatrix} = 0
\end{equation}
where
\begin{equation}
\begin{pmatrix}
    \psi_0(x) \\
    \psi_1(x)
\end{pmatrix}  
= \begin{pmatrix} \chi_{24, 0}^{(\varnothing)}(x-m) \\ \chi_{24, 1}^{(\varnothing)}(x-m) + \mathfrak{q}_1 \chi_{24, 1}^{(\square)}(x-m)
    \end{pmatrix}
\end{equation}
which agrees with our expectations.

\subsection{Warm up: General case first terms}
The potential eigenvector should have the form:
\begin{equation}
  \vec{\chi}_{24}(x)  = \begin{pmatrix} \chi_{24, 0}(x) \\ \chi_{24, 1}(x) \\
    \vdots \\
    \chi_{24, N-1}(x)
    \end{pmatrix}
\end{equation}
Similarly to the $SL(2)$-case let us start with the case of a trivial diagram:
\begin{equation}
  \vec{\chi}_{24}^{(\varnothing)}(x)  = \begin{pmatrix} \chi_{24, 0}^{(\varnothing)}(x) \\ \chi_{24, 1}^{(\varnothing)}(x) \\
    \vdots \\
    \chi_{24, N-1}^{(\varnothing)}(x)
    \end{pmatrix}
\end{equation}
Where:
\begin{equation}
\chi_{24, \omega}^{(\varnothing)} (x) =  {\bE} \Big[ -\frac{P_3}{1 - q_2^N} \sum_{\substack{a,c = \,0,...,N-1\\ a - c = - \omega \, \text{mod} N}} q_2^{a+1} S_{12, \, c+1}^* \Big]     
\end{equation}
Or, more explicitly: 
\begin{equation}
\chi_{24, \omega}^{(\varnothing)} (x) =  {\bE} \Big[ -\frac{P_3}{1 - q_2^N} \sum_{c = 0}^\omega q_2^{c- \omega + N} S_{12, \, c}^*  -\frac{P_3}{1 - q_2^N} \sum_{c = \omega+1}^{N-1} q_2^{c- \omega} S_{12, \, c}^*  \Big]     
\end{equation}
In the multiplicative form:
\begin{equation}
\chi_{24, \omega}^{(\varnothing)} (x) = \prod_{c = 0}^\omega \frac{\tilde{Q}_c\big(x + (c- \omega + N) \epsilon_2\big)}{\tilde{Q}_c\big(x + m + (c- \omega + N) \epsilon_2\big)}  \prod_{c = \omega+1}^{N-1} \frac{\tilde{Q}_c\big(x + (c- \omega ) \epsilon_2\big)}{\tilde{Q}_c\big(x + m + (c- \omega) \epsilon_2\big)}  
\end{equation}
From this definition, it follows that:
\begin{equation}
\chi_{24, \omega}^{(\varnothing)} (x + \epsilon_2) = \frac{Y_\omega(x+m + \epsilon_2)}{Y_\omega(x + \epsilon_2)} \chi_{24, \omega-1}^{(\varnothing)} (x) 
\end{equation}
In the matrix form, this equation could be rewritten as:
\begin{equation}
    \hat{Y}(x+ \epsilon_2) \vec{\chi}_{24}^{(\varnothing)}(x-m) - \hat{Y}(x+ \epsilon_2 - m) \, C \, e^{\epsilon_2 \partial_x}\vec{\chi}_{24}^{(\varnothing)}(x-m) = 0
\end{equation}
where:
\begin{equation}
    \hat{Y}(x) = \text{diag} \big( Y_{\omega+1}(x) \big)_{\omega = 0 }^{\omega = N-1} 
\end{equation}
and
\begin{equation}
    C = \begin{pmatrix} 0 & 1 & 0 &...& 0\\ 0 & 0 & 1 &...& 0 \\ \vdots & \vdots & \ddots & \ddots & \vdots \\ 0 & 0 & 0 & ... & 1 \\ 1 & 0 & 0 & ... & 0\end{pmatrix} = 
    \sum_{\omega} e_\omega \otimes e_{\omega+1}^t
\end{equation}
Which is equal to $\vec{\mathfrak{q}} = 0$ term of the general formula from the theorem \ref{Eigenvectorthm}.

\subsection{Proof of the general case}

\begin{proof}[Proof of theorem \ref{Eigenvectorthm}]
Similarly to the non-orbifold case, we could single out the part of the character of our main interest:
\begin{equation}
    P_3 \sum_{\omega'} \sum_{a=1}^N q_2^{a+n} q_3^{-1} S^*_{12, \omega+a+n+\omega'} S_{24, \omega'} + q_2 \tilde{P}_2 \sum_{\omega'} q_3^{-n} S^*_{12, \omega + \omega' + 1} S_{34, \omega'}
\end{equation}
Let us consider the 24-part first:
\begin{equation}
   \textnormal{Ch}^{24}_{\omega+n}(x+\epsilon_2 n) = P_3 \sum_{\omega'} \sum_{a=1}^N q_2^{a+n} q_3^{-1} S^*_{12, \omega+a+n+\omega'} S_{24, \omega'}
\end{equation}
We will use the following expression for $S_{24}$:
\begin{equation}
    S_{24, \omega} |_\mu = \sum_{ \substack{(i,j) \in \partial_+ \mu \\ i-j = \omega \text{mod} N} }  q_2^{i-1} q_4^{j-1} - \sum_{ \substack{(i,j) \in \partial_- \mu \\ i-j = \omega \text{mod} N} }  q_2^{i} q_4^{j}
\end{equation}
Let us define:
\begin{equation}
    \tilde{\mu}_{i} = \mu_{i+1}, \qquad i = 1,..,\mu_1^t \\
\end{equation}
Then, one of the 2 options is possible:
\begin{gather}
    1) \quad \mu_2 < \mu_1 \\
    2) \quad \mu_2 = \mu_1
\end{gather}
In scenario 1), the diagram $\mu$ has one addable and one removal box in the first row, and they differ by the shift in $\epsilon_4$, meaning, that the removable box shifted by $\epsilon_4 + \epsilon_2$ now differ from the addable one by $\epsilon_2$. Hence, as we have:
\begin{gather}
    \partial_+ \mu = (\epsilon_2 + \partial_+ \tilde{\mu}) \cup \{(1, \mu_1 +1)\} \\
    \partial_- \mu = (\epsilon_2 + \partial_- \tilde{\mu}) \cup \{(1, \mu_1)\}
\end{gather}
the character:
\begin{multline}
   \textnormal{Ch}^{24}_{\omega+n}(x+\epsilon_2 n) |_{\mu} = P_3 \sum_{a=1}^N \sum_{(i,j) \in \partial_+ \mu} q_2^{a+n+i-1} q_4^{j-1} q_3^{-1}  S^*_{12, \omega+a+n+i-j} - \\ - P_3 \sum_{a=1}^N \sum_{(i,j) \in \partial_- \mu} q_2^{a+n+i} q_4^{j} q_3^{-1}  S^*_{12, \omega+a+n+i-j}
\end{multline}
could be rewritten as
\begin{multline}
P_3 \sum_{a=1}^N q_2^{a+n} q_4^{\mu_1}q_3^{-1}  S_{12, \omega+a+n - \mu_1^t}^* - P_3 \sum_{a=1}^N q_2^{a+n+1} q_4^{\mu_1}q_3^{-1}  S_{12, \omega+a+n - \mu_1^t +1}^* + \\
+ P_3 \sum_{a=1}^N \sum_{(i,j) \in \epsilon_2 +\partial_+ \tilde{\mu} } q_2^{a+n+i-1} q_4^{j-1} q_3^{-1}  S^*_{12, \omega+a+n+i-j} - \\ - P_3 \sum_{a=1}^N \sum_{(i,j) \in \epsilon_2 +\partial_- \tilde{\mu}} q_2^{a+n+i} q_4^{j} q_3^{-1}  S^*_{12, \omega+a+n+i-j}    
\end{multline}
this is the same as:
\begin{equation}
    P_3 \Tilde{P_2} q_2^{n+1}q_4^{\mu_1^t}q_3^{-1} S^*_{12, \omega+n+1- \mu_1} + \textnormal{Ch}^{24}_{\omega+n+1}(x+\epsilon_2 (n+1)) |_{\tilde{\mu}}
\end{equation}
In scenario 2), the diagrams $\mu$ and $\tilde{\mu}$ have the same number of addable and removable boxes. But boxes located outside of the first column got shifted by one to the left. Hence:
\begin{gather}
\partial_+ \mu \backslash \{(1, \mu_1 +1)\} = \epsilon_2 + \partial_+ \tilde{\mu} \backslash \{(1, \mu_1 +1)\}    \\
\partial_- \mu  = \epsilon_2 + \partial_- \tilde{\mu} 
\end{gather}
Therefore for the character we have:
\begin{multline}
P_3 \sum_{a=1}^N q_2^{a+n} q_4^{\mu_1}q_3^{-1}  S_{12, \omega+a+n - \mu_1}^* + \\ + P_3 \sum_{a=1}^N \sum_{(i,j) \in \epsilon_2 +\partial_+ \tilde{\mu} \backslash \{(1, \mu_1 +1)\} } q_2^{a+n+i-1} q_4^{j-1} q_3^{-1}  S^*_{12, \omega+a+n+i-j} - \\ - P_3 \sum_{a=1}^N \sum_{(i,j) \in \epsilon_2 +\partial_- \tilde{\mu}} q_2^{a+n+i} q_4^{j} q_3^{-1}  S^*_{12, \omega+a+n+i-j}     
\end{multline}
And this is equal to the same expression:
\begin{equation}
    P_3 \Tilde{P_2} q_2^{n+1}q_4^{\mu_1}q_3^{-1} S^*_{12, \omega+n+1- \mu_1} + \textnormal{Ch}^{24}_{\omega+n+1}(x+\epsilon_2 (n+1))  |_{\tilde{\mu}}
\end{equation}
Now let us have a look at the $34$ part of the character:
\begin{equation}
\textnormal{Ch}_\omega^{34}(x- \epsilon_3 n) =  q_2 \tilde{P}_2 \sum_{\omega'} q_3^{-n} S^*_{12, \omega + \omega' + 1} S_{34, \omega'}  
\end{equation}
We could write it down as:
\begin{equation}
\textnormal{Ch}_\omega^{34}(x- \epsilon_3 n) |_{\tilde{\nu}} = q_2 \tilde{P}_2 \sum_{(i,j) \in \partial_+ \tilde{\nu}} \ S^*_{12, \omega + 2- j} q_3^{i-1-n} q_4^{j-1}   - q_2 \tilde{P}_2 \sum_{(i,j) \in \partial_- \tilde{\nu}}  S^*_{12, \omega + 1- j} q_3^{i-n} q_4^{j} 
\end{equation}
Now let us try to separate the contribution from the first row. Introduce the diagram $\nu$:
\begin{equation}
    \nu_i = \tilde{\nu}_{i+1}
\end{equation}
Here we need to consider two cases as well.  The first one $\tilde{\nu}_1 > \tilde{\nu}_2$. Then we have one addable and one removable box in the first row, and the character could be written as:
\begin{multline}
\textnormal{Ch}_\omega^{34}(x- \epsilon_3 n) |_{\tilde{\nu}} =  q_2 \tilde{P}_2 S^*_{12, \omega + 1- \tilde{\nu}_1} q_3^{-n} q_4^{\tilde{\nu}_1} - q_2 \tilde{P}_2 S^*_{12, \omega + 1- \tilde{\nu}_1} q_3^{1-n} q_4^{\tilde{\nu}_1} + \\+ q_2 \tilde{P}_2 \sum_{(i,j) \in \partial_+ \nu} \ S^*_{12, \omega + 2- j} q_3^{i-n} q_4^{j-1}   - q_2 \tilde{P}_2 \sum_{(i,j) \in \partial_-\nu}  S^*_{12, \omega + 1- j} q_3^{i+1-n} q_4^{j}     
\end{multline}

which is:
\begin{equation}
\textnormal{Ch}_\omega^{34}(x- \epsilon_3 n) |_{\tilde{\nu}} =  q_2 \tilde{P}_2 P_3 S^*_{12, \omega + 1- \tilde{\nu}_1} q_3^{-n} q_4^{\tilde{\nu}_1} + \textnormal{Ch}_\omega^{34}(x- \epsilon_3 (n-1)) |_{\nu}   
\end{equation}
Or the other way around:
\begin{equation}
\textnormal{Ch}_\omega^{34}(x- \epsilon_3 n) |_{\nu}   = - q_2 \tilde{P}_2 P_3 S^*_{12, \omega + 1- \tilde{\nu}_1} q_3^{-n-1} q_4^{\tilde{\nu}_1} + \textnormal{Ch}_\omega^{34}(x- \epsilon_3 (n+1))   |_{\tilde{\nu}}
\end{equation}
Now consider the case $\tilde{\nu}_1 = \tilde{\nu}_2$. Then we have only one addable box in the first row and no removable boxes. However, the diagram $\nu$ will have one more addable box than $\tilde{\nu}$ with the first row deleted. So:
\begin{multline}
\textnormal{Ch}_\omega^{34}(x- \epsilon_3 n) |_{\tilde{\nu}} =  q_2 \tilde{P}_2 S^*_{12, \omega + 1- \tilde{\nu}_1} q_3^{-n} q_4^{\tilde{\nu}_1} + \\+ q_2 \tilde{P}_2 \sum_{(i,j) \in \partial_+ \nu \backslash \{(1, \nu_1 +1)\} } \ S^*_{12, \omega + 2- j} q_3^{i-n} q_4^{j-1}   - q_2 \tilde{P}_2 \sum_{(i,j) \in \partial_-\nu}  S^*_{12, \omega + 1- j} q_3^{i+1-n} q_4^{j}     
\end{multline}
which, however, gives the same final result:
\begin{equation}
\textnormal{Ch}_\omega^{34}(x- \epsilon_3 n) |_{\tilde{\nu}} =  q_2 \tilde{P}_2 P_3 S^*_{12, \omega + 1- \tilde{\nu}_1} q_3^{-n} q_4^{\tilde{\nu}_1} + \textnormal{Ch}_\omega^{34}(x- \epsilon_3 (n-1)) |_{\nu}   
\end{equation}
Now, if we pick $\tilde{\nu}_1 = \mu_1 - n$, thanks to equalities:
\begin{equation}
q_2^{n+1} q_4^{\mu_1} q_3^{-1} = q_2^{n+2} q_4^{\mu+1}   = q_2^{n+2} q_4^{\tilde{\nu}_1+n+1} = q_2 q_3^{-n-1} q_4^{\tilde{\nu}_1}
\end{equation}
we have:
\begin{multline}
\textnormal{Ch}^{24}_{\omega+n}(x+\epsilon_2 n) |_{\mu} + \textnormal{Ch}_\omega^{34}(x- \epsilon_3 n) |_{\nu}  =  \textnormal{Ch}^{24}_{\omega+n+1}(x+\epsilon_2 (n+1)) |_{\tilde{\mu}} + \textnormal{Ch}_\omega^{34}(x- \epsilon_3 (n+1))   |_{\tilde{\nu}}
\end{multline}
This formula will make sense as soon as $\tilde{\nu}$ is actually a Young diagram, namely when $\mu_1 - n \geq \nu_1$. This question will be addressed below.\\
Now, we need to check that the prefactor matches:
\begin{equation}
    D_{\omega}^{(-n-1)} \mathbb{Q}^{(34), \tilde{\mu}}_{\omega+n+1} \mathbb{Q}_\omega^{\tilde{\nu}} = D_{\omega+n}^{(-n)}\mathbb{Q}^{(34), \mu}_{\omega+n} \mathbb{Q}_\omega^{\nu}
\end{equation}
Let us look at $\mathbb{Q}_\omega^{\tilde{\nu}} $ first:
\begin{equation}
\mathbb{Q}_\omega^{\tilde{\nu}} = \prod_{i = 1}^{l(\tilde{\nu})} \mathfrak{q}_\omega... \mathfrak{q}_{\omega - \tilde{\nu}_i + 1} =  \mathfrak{q}_\omega... \mathfrak{q}_{\omega - \tilde{\nu}_1 + 1} \,  \mathbb{Q}_\omega^{\nu} = \mathfrak{q}_\omega... \mathfrak{q}_{\omega - \mu_1 + n + 1} \,  \mathbb{Q}_\omega^{\nu}
\end{equation}
Now let us have a look at $\mathbb{Q}^{(34), \mu}_{\omega+n+1}$:
\begin{multline}
\mathbb{Q}^{(34), \mu}_{\omega+n}  = \prod_{i =1}^{l(\mu)} \prod_{j=1}^{\mu_i}  \mathfrak{q}_{\omega +i -j + n} = \prod_{j =1}^{\mu_1} \mathfrak{q}_{\omega + n +1 -j} \prod_{i =2}^{l(\mu) } \prod_{j=1}^{\mu_{i}} \mathfrak{q}_{\omega+1-j+n} = \\ =  \prod_{j =1}^{\mu_1} \mathfrak{q}_{\omega + n +1 -j} \prod_{i =1}^{l(\mu)-1 } \prod_{j=1}^{\mu_{i+1}} \mathfrak{q}_{\omega+1-j+n+1} = \prod_{j =1}^{\mu_1} \mathfrak{q}_{\omega + n +1 -j} \mathbb{Q}^{(34), \tilde{\mu}}_{\omega+n+1}
\end{multline}
Hence, we have:
\begin{equation}
     \mathbb{Q}^{(34), \tilde{\mu}}_{\omega+n+1} \mathbb{Q}_\omega^{\tilde{\nu}} = \frac{\mathfrak{q}_\omega... \mathfrak{q}_{\omega - \mu_1 + n + 1} }{\prod_{j =1}^{\mu_1} \mathfrak{q}_{\omega + n +1 -j}} \mathbb{Q}^{(34), \mu}_{\omega+n} \mathbb{Q}_\omega^{\nu}
\end{equation}
We have to consider 2 cases again, and we see that:
\begin{equation}
  \frac{\mathbb{Q}^{(34), \tilde{\mu}}_{\omega+n+1} \mathbb{Q}_\omega^{\tilde{\nu}}}{\mathbb{Q}^{(34), \mu}_{\omega+n} \mathbb{Q}_\omega^{\nu}} = 
  \begin{cases}
      \mathfrak{q}_{\omega+1}^{-1} ... \mathfrak{q}_{\omega+n}^{-1} \qquad n >0\\
      \mathfrak{q}_{\omega} ... \mathfrak{q}_{\omega+n+1} \qquad n  \leq 0
  \end{cases}
\end{equation}
Let us make sure, that this result agrees with the transformation laws of $D_{\omega}^{(-n)}$.
If $n>0$, one has:
\begin{equation}
    D_{\omega}^{(-n-1)} = \prod_{k=1}^{n} \mathfrak{q}_{\omega+k}^{n+1-k} = D_{\omega}^{(-n)} \prod_{k=1}^{n} \mathfrak{q}_{\omega+k} 
\end{equation}
And if $n \leq 0$, as one has:
\begin{equation}
D_{\omega}^{(-n)} = \prod_{k=0}^{-n-1} \mathfrak{q}_{\omega-k}^{-n-k}   
\end{equation}
we could write:
\begin{equation}
D_{\omega}^{(-n-1)} = \prod_{k=0}^{-n-2} \mathfrak{q}_{\omega-k}^{-n-k-1}  = D_{\omega}^{(-n)} \prod_{k=0}^{-n-1} \frac{1}{\mathfrak{q}_{\omega-k}}  
\end{equation}
Therefore, we got exactly what we needed
\begin{equation}
  \frac{D_{\omega}^{(-n)}}{D_{\omega}^{(-n-1)}} = 
  \begin{cases}
      \mathfrak{q}_{\omega+1}^{-1} ... \mathfrak{q}_{\omega+n}^{-1} \qquad n >0\\
      \mathfrak{q}_{\omega} ... \mathfrak{q}_{\omega+n+1} \qquad n  \leq 0
  \end{cases}
\end{equation}
To complete the full proof, we only need to notice that the transformation:
\begin{equation}
    (\mu, \nu, n) \rightarrow (\tilde{\mu}, \tilde{\nu}, \tilde{n} = n+1)
\end{equation}
 maps the triples $(\mu, \nu, n)$ with $\mu_1 - n \geq \nu_1$ to the triples $(\tilde{\mu}, \tilde{\nu}, \tilde{n})$ satisfying $\tilde{\mu}_1 - \tilde{n} < \tilde{\nu}_1$, and is invertible on its image. 
\end{proof}

\section{Proof of the spectral duality}
\begin{proof}[Proof of the theorem \ref{Spectraldual}]
    Let us look at one multiple in the middle of the LHS of the classical factorization formula (\ref{factorizationCl}):
    \begin{equation}
        \hat{Y}(x) \Big(1 - \frac{\hat{Y}(x-m)}{\hat{Y}(x)} C_z\Big)
    \end{equation}
It is easy to see that it could be rewritten as:
\begin{equation}
    - \Big(1 - \hat{Y}(x) C_z^{-1} \hat{Y}(x-m)^{-1}\Big) \hat{Y}(x-m) C_z
\end{equation}
Now let us drag the new factor appearing on the right through the next multiple in a product (\ref{factorizationCl}):
\begin{multline}
    -\hat{Y}(x-m) C_z \Big(1 - (\hat{\slashed{Q}} C_z^{-1}) \frac{\hat{Y}(x-2m)}{\hat{Y}(x-m)} C_z^2 \Big) =\\=
    \Big(1 - \hat{\slashed{Q}}^{-1}\hat{Y}(x-m) C_z^{-1} \hat{Y}(x-2m)^{-1}\Big) \hat{Y}(x-2m) C_z \hat{\slashed{Q}} C_z
\end{multline}
Repeating this procedure $n$ times we will get for the LHS of (\ref{factorizationCl}):
\begin{multline}
    \overleftarrow{\prod_{k=-n+1}^\infty} \Big(1-  \hat{\slashed{Q}}^k \hat{Y}(x+km ) \, C_{z}^{-1} \hat{Y}(x+(k-1)m )^{-1} \Big)  (-1)^n\hat{Y}(x-nm) \prod_{i=1}^n C_z \hat{\slashed{Q}}^{n-i} \cdot \\
\cdot \overrightarrow{\prod_{k=n}^\infty} \Big(1 -   (\hat{\slashed{Q}} C_{z}^{-1})^k \frac{\hat{Y}(x-(k+1)m)}{\hat{Y}(x-k m)} C_{z}^{\, k+1}\Big)
\end{multline}
Indeed, the proof is carried out by induction. First of all, after the factor $(\hat{\slashed{Q}} C_{z}^{-1})^n$ in the first next multiple of the product on the right will get multiplied by $\prod_{i=1}^n C_z \hat{\slashed{Q}}^{n-i}$, the ratio $\frac{\hat{Y}(x-(n+1)m)}{\hat{Y}(x-n m)}$ could be taken out to the left and denominator will get canceled with $\hat{Y}(x-nm)$. So, we are only left to prove that:
\begin{lemma}
 \begin{equation}
 \prod_{i=1}^n C_z \hat{\slashed{Q}}^{n-i}   (\hat{\slashed{Q}} C_{z}^{-1})^n \, C_z^{n+1} = \prod_{i=1}^{n+1} C_z \hat{\slashed{Q}}^{n+1-i}
\end{equation}   
\end{lemma}
\begin{proof}
Let us write the LHS explicitly:
\begin{equation}
    C_z\hat{\slashed{Q}}^{n-1} C_z\hat{\slashed{Q}}^{n-2}...C_z \hat{\slashed{Q}} C_z \underbrace{\hat{\slashed{Q}}C_z^{-1}\cdot...\cdot \hat{\slashed{Q}}C_z^{-1}}_n C_z^{n+1}
\end{equation}
Let us look at the middle: the matrix $C_z \hat{\slashed{Q}} C_z^{-1}$ is diagonal, so could be exchanged with $\hat{\slashed{Q}}$ to the left of it:
\begin{equation}
    C_z\hat{\slashed{Q}}^{n-1} C_z\hat{\slashed{Q}}^{n-2}...C_z \hat{\slashed{Q}}^2 C_z^2 \hat{\slashed{Q}} C_z^{-1} \hat{\slashed{Q}}^2 C_z^{-1}\underbrace{\hat{\slashed{Q}}C_z^{-1}\cdot...\cdot \hat{\slashed{Q}}C_z^{-1}}_{n-2} C_z^{n+1}
\end{equation}
Now the matrix $C_z^2 \hat{\slashed{Q}} C_z^{-1} \hat{\slashed{Q}}^2 C_z^{-1}$ is diagonal, so could be exchanged with $\hat{\slashed{Q}}^2$ to the left of it too. Continue this process further one gets:
\begin{equation}
    C_z^n \hat{\slashed{Q}} C_z^{-1} \hat{\slashed{Q}}^{2} C_z^{-1}\cdot...\cdot \hat{\slashed{Q}}^{n-1} C_z^{-1} \hat{\slashed{Q}}^n C_z^n
\end{equation}
Now the matrix: $C_z^{-1} \hat{\slashed{Q}}^{2} C_z^{-1}\cdot...\cdot \hat{\slashed{Q}}^{n-1} C_z^{-1} \hat{\slashed{Q}}^n C_z^{n-1}$ is diagonal, so it could be interchanged with $\hat{\slashed{Q}}$ standing to its left to get:
\begin{equation}
    C_z^{n-1} \hat{\slashed{Q}}^2 C_z^{-1} \hat{\slashed{Q}}^{2} C_z^{-1}\cdot...\cdot \hat{\slashed{Q}}^{n-1} C_z^{-1} \hat{\slashed{Q}}^n C_z^{n-1}\hat{\slashed{Q}} C_z 
\end{equation}
The matrix $C_z^{-1} \hat{\slashed{Q}}^{2} C_z^{-1}\cdot...\cdot \hat{\slashed{Q}}^{n-1} C_z^{-1} \hat{\slashed{Q}}^n C_z^{n-2}$ is diagonal and could be exchanged with $\hat{\slashed{Q}}^{2}$ to its left as well. Continuing this process further, we are getting the desired equality.
\end{proof}
Notice that for the determinant we have:
\begin{equation}
    \det_{N \times N} \Big[(-1)^n\hat{Y}(x-nm) \prod_{i=1}^n C_z \hat{\slashed{Q}}^{n-i} \Big] = z^{-n} \mathfrak{q}^{\frac{n^2-n}{2}} Y(x-nm)
\end{equation}
Hence, repeating the procedure described above but now for the scalar version of the factorization formula (\ref{factorizationClSc}) one finds that the following formula holds:
\begin{multline}
     \det_{N \times N} \overleftarrow{\prod_{n \in \mathbb{Z}}} \Big(1-  \hat{\slashed{Q}}^n \hat{Y}(x+nm ) \, C_{z}^{-1} \hat{Y}(x+(n-1)m )^{-1} \Big) = \\= \prod_{n \in \mathbb{Z}} \Big(1-  z\mathfrak{q}^n Y(x+nm ) \, Y(x+(n-1)m )^{-1} \Big)
\end{multline}
and both sides differ from the $\det \mathfrak{D}(x,z)$ by the same ill-defined factor:
\begin{equation} \label{illfactor}
    \lim_{n \rightarrow \infty} z^{-n} \mathfrak{q}^{\frac{n^2-n}{2}} Y(x-nm) = \lim_{n \rightarrow \infty} z^{-n} \mathfrak{q}^{\frac{n^2-n}{2}} (-nm)^N
\end{equation}
Now let us denote:
\begin{equation}
    \tilde{\mathfrak{D}}(x,z) := \overleftarrow{\prod_{n \in \mathbb{Z}}} \Big(1-  \hat{\slashed{Q}}^n \hat{Y}(x+nm ) \, C_{z}^{-1} \hat{Y}(x+(n-1)m )^{-1} \Big)
\end{equation}
We would like to study the null-vectors of it:
\begin{equation}
    \tilde{\mathfrak{D}}(x,z) \, \psi = 0
\end{equation}
To each null-vector we can associate a sequence of vectors $\psi_n, \, n \in \mathbb{Z}$, such that:
\begin{gather}
\Big[1-  \hat{\slashed{Q}}^n \hat{Y}(x+nm ) \, C_{z}^{-1} \hat{Y}(x+(n-1)m )^{-1} \Big] \psi_n = \psi_{n+1}\\
    \lim_{n \rightarrow - \infty} \psi_n = \psi \\
    \lim_{n \rightarrow + \infty} \psi_n = 0 
\end{gather}

Let us denote their components by $\psi_{n}^{(\omega)}$. The first of the equations above in coordinates will look like:
\begin{equation}
    \psi_n^{(\omega+1)} - z^{\delta_{0,\omega}}\mathfrak{q}_\omega^n \frac{Y_{\omega+1,n}}{Y_{\omega,n-1}} \psi_n^{(\omega)} = \psi_{n+1}^{(\omega+1)}
\end{equation}
Solving each equation with respect to $\psi_n^{(\omega)}$ one obtains:
\begin{gather}
    \psi^{(\omega)} = \mathcal{L}_\omega \, \psi^{(\omega+1)} \, , \quad \omega =1,...,N-1 \\
    z \psi^{(N)} = \mathcal{L}_N \psi^{(1)}
\end{gather}
where we have introduced new notations for the infinite-dimensional vectors:
\begin{equation}
   \psi^{(\omega)} = \sum_{n \in \mathbb{Z}}  \psi^{(\omega)}_n e_n
\end{equation}
and the notations for $\mathcal{L}_\omega$ were introduced in the formulation of the theorem (\ref{SpinLax}). From these equations it clearly follows that:
\begin{equation}
    z \psi^{(1)} = \mathcal{L}_1 \mathcal{L}_2 \cdot...\cdot \mathcal{L}_N \psi^{(1)}
\end{equation}
Hence, the infinite-dimensional matrix:
\begin{equation}
    z -\mathcal{L}_1 \mathcal{L}_2 \cdot...\cdot \mathcal{L}_N 
\end{equation}
should be degenerate. Now we only need to find the correct normalization for its determinant to match our desired answer. To do that, notice, that the matrices $\mathcal{L}_\omega$ are triangular, therefore the determinant 
\begin{equation}
    \det_{\infty \times \infty}\big(1 - z \mathcal{L}_N^{-1} \mathcal{L}_{N-1}^{-1}\cdot...\cdot \mathcal{L}_1^{-1}\big)
\end{equation}
is easy to calculate:
\begin{equation}
    \det_{\infty \times \infty}\big(1 - z \mathcal{L}_N^{-1} \mathcal{L}_{N-1}^{-1}\cdot...\cdot \mathcal{L}_1^{-1}\big) = \prod_{n \in \mathbb{Z}} \Big(1-  z\mathfrak{q}^n Y(x+nm ) \, Y(x+(n-1)m )^{-1} \Big)
\end{equation}
and it matches the expression for the determinant of $\tilde{\mathfrak{D}}(x,z)$:
\begin{equation}
   \det_{N\times N} \tilde{\mathfrak{D}}(x,z) = \det_{\infty \times \infty}\big(1 - z\, T_N(x)\big)
\end{equation}
where we have introduced:
\begin{equation}
    T_{N}(x) = \mathcal{L}_{N-1}^{-1}(x)\cdot...\cdot \mathcal{L}_1^{-1}(x)
\end{equation}
Which means:
\begin{equation}
\det_{N\times N} \mathfrak{D}(x,z) = (\text{ill-defined factor}) \times\, \det_{\infty \times \infty}\big(1 - z\, T_N(x)\big)    
\end{equation}
The expression for this ill-defined factor (\ref{illfactor}) implies that the equality on the following ratios will not contain it:
\begin{equation}
    \frac{\det_{N \times N} \mathfrak{D}(x+m,z)}{\det_{N \times N} \mathfrak{D}(x,z)} = \frac{ \det_{\infty \times \infty } \big(1 - z \, T_N(x+m)\big)}{ \det_{\infty \times \infty } \big(1 - z \, T_N(x)\big)}
\end{equation}
\end{proof}

\section{Future directions}
The most fascinating generalization of the above results would be extending them to the case of both $(\epsilon_1, \epsilon_2)$ being non-zero, potentially uncovering their geometric and representation-theoretic meaning. Specifically, we hope for some relation of the infinite product formula to the ones in \cite{BBZ}, \cite{Zenk}, and possibly  \cite{KonSoib}.\\
We anticipate that it will pave the way toward the solution of the longstanding problem of the quantum separation of variables for the elliptic Calogero-Moser system. Namely, following E. Sklyanin prescription \cite{Skl}, the expression for the wavefunction $\Psi^{ellCM}$, found in \cite{BPSCFT5} after certain integral transform will take the factorized form:
\begin{equation}
    \Psi^{ellCM}(x_1,...,x_N) = \int d^N \Vec{x}^{SV} K(\Vec{x}|\Vec{x}^{SV}) \prod_{i=1}^N Q_{SV}(x^{SV}_i)
\end{equation}
Where $x^{SV}_i$ - are separated variables of the system, and $Q_{SV}$ is a null vector of the quantum spectral curve (\ref{FactorizationScalar}), found in \cite{NG}. 

\section{Appendix: N = 2 Example}
For $N =2$ the matrix $\mathfrak{D}(x,z)$ could be written down explicitly:
\begin{equation}
 \mathfrak{D}(x,z) = \begin{pmatrix}
     \sum_{k \in \mathbb{Z}} \mathfrak{q}^{k(k+1)} z^k \Big(\frac{x_0}{x_1}\Big)^k \chi_0(x+ 2 k m) & \sum_{k \in \mathbb{Z}} \mathfrak{q}^{k^2} z^k \Big(\frac{x_0}{x_1}\Big)^k \chi_0(x+ (2 k-1) m) \\
     \sum_{k \in \mathbb{Z}} \mathfrak{q}^{k(k-1)} z^{k-1} \Big(\frac{x_1}{x_0}\Big)^k \chi_1(x+ (2 k-1) m) & \sum_{k \in \mathbb{Z} }\mathfrak{q}^{k^2} z^k \Big(\frac{x_1}{x_0}\Big)^k \chi_1(x+ 2 k m)
 \end{pmatrix}   
\end{equation}
This matrix has a great resemblance to the one used in the IRF-Vertex correspondence discovered in \cite{H}, and further studied in \cite{LOZ} for the purpose of the Symplectic Hecke Correspondence (see formula 4.25). However, for the number of particles $N$ greater than 2, it starts to look somewhat different.\\ 
For visibility, let us demonstrate how everything works in the trigonometric limit. The formula above then takes the form:
\begin{equation}
 \mathfrak{D}(x,z)^{\text{trig}} = \begin{pmatrix}
     \chi_0(x) + z^{-1} \, \frac{x_1}{x_0} \chi_0(x- 2  m) & \chi_0(x - m) \\
     \frac{x_1}{x_0} \chi_1(x+m) + z^{-1}  \chi_1(x - m) &  \chi_1(x)
 \end{pmatrix}   
\end{equation}
And the main identity (\ref{factorizationCl})  reduces to:
\begin{multline}
\overleftarrow{\prod_{n=1}^\infty} \begin{pmatrix}
    1 & 0 \\ \Big( \frac{x_1}{x_0}\Big)^n \frac{ Y_0(x+nm)}{Y_1(x+(n-1)m)} & 1
\end{pmatrix} 
\begin{pmatrix}
 Y_1(x)  &  Y_1(x-m) \\ z^{-1} Y_0(x-m) & Y_0(x) 
\end{pmatrix}
\begin{pmatrix}
    1 & 0 \\ z^{-1}  \frac{x_1}{x_0} \frac{Y_1(x-2m)}{Y_1(x-m)} & 1
\end{pmatrix} = \\ =
\begin{pmatrix}
     \chi_0(x) + z^{-1} \, \frac{x_1}{x_0} \chi_0(x- 2  m) & \chi_0(x - m) \\
     \frac{x_1}{x_0} \chi_1(x+m) + z^{-1}  \chi_1(x - m) &  \chi_1(x)
 \end{pmatrix} 
\end{multline}
which could be even easily verified directly, as we know that:
\begin{gather}
    \chi_0(x) = Y_1(x) 
    \\
    \chi_1(x) = \sum_{n=0}^\infty \Big( \frac{x_1}{x_0}\Big)^n \frac{ Y_0(x+nm) Y_1(x-m)}{Y_1(x+(n-1)m)} 
\end{gather}
Now let's explore what kind of Hamiltonians one obtains from this matrix:
\begin{equation}
    \det \mathfrak{D}(x,z)^{\text{trig}} = D(x) - z^{-1} D(x-m)
\end{equation}
where:
\begin{equation}
    D(x) = \chi_0(x) \chi_1(x) - \frac{x_1}{x_0} \chi_0(x-m) \chi_1(x+m)
\end{equation}
In the 6d (Equivariant elliptic cohomologies) case when the $qq$-Characters have the form:
\begin{gather}
    \chi_0(x) = \theta_{p_{6d}} \big( e^{\beta p_0 } e^{-\beta x } \big) \\
    \chi_1(x) = \frac{\theta_{p_{6d}} \big( e^{\beta p_1 } e^{-\beta x } \big)}{1 - \frac{x_1}{x_0} }
\end{gather}
one has:
\begin{equation}
    D^{6d}(x) =  \sum_{n_0, n_1 \in \mathbb{Z} } e^{2\pi i \tau_{6d}\sum_{k=0}^1 \frac{n_k^2-n_k}{2}}  \big(-e^{-\beta x} \big)^{(n_1 + n_2)} \frac{1 - e^{\beta m(n_1-n_0)}\frac{x_1}{x_0}}{1 - \frac{x_1}{x_0}}  e^{\beta n_0 p_0 + \beta n_1 p_1}
\end{equation}
That is exactly the classical, trigonometric limit of the Dell generating function found in the paper \cite{KS}.\\
In the $5d$ (Equivariant K-theory)- limit ($\tau_{6d} \rightarrow \infty$) it reduces to the generating function of the tRS Hamiltonians:
\begin{equation}
    D^{5d}(x) =  1 - e^{-\beta x } \Big[ \frac{1 - e^{\beta m}\frac{x_1}{x_0}}{1 - \frac{x_1}{x_0}} e^{\beta p_1}  + \frac{1 - e^{-\beta m}\frac{x_1}{x_0}}{1 - \frac{x_1}{x_0}} e^{\beta p_0} \Big] +  e^{-2 \beta x } e^{\beta p_0 + \beta p_1}
\end{equation}

\section{Appendix: Matrix elliptic Jacobi identity}
Elliptic theta-function:
\begin{equation}
    \theta_{\mathfrak{q}}(z) = \sum_{n \in \mathbb{Z}} (-z)^n \mathfrak{q}^{\frac{n^2-n}{2}} 
\end{equation}
could be seen as a 1 by 1 matrix of the homomorphism between a bundle of rank 1 degree 0 and a bundle of rank 1 degree 1 on the elliptic curve with parameter $\mathfrak{q}$. This could be encoded in its transformation property:
\begin{equation}
    \theta_{\mathfrak{q}}(\mathfrak{q}z) = - z^{-1}  \theta_{\mathfrak{q}}(z)
\end{equation}
It satisfies the following Jacobi identity:
\begin{equation}
\frac{1}{\prod_{n=1}^\infty (1- \mathfrak{q}^n)}\sum_{n \in \mathbb{Z}} (-z)^n \mathfrak{q}^{\frac{n^2-n}{2}} = \prod_{n=0}^\infty (1- \mathfrak{q}^n z) (1-\mathfrak{q}^{n+1}  z^{-1})    
\end{equation}
We found a remarkable matrix analog of this identity (proven in the main text). Let $\mathfrak{D}_1(z)$ be the matrix of the homomorphism between bundle of rank N degree 1 and bundle of rank N degree 0, both with the moduli $x_0,...,x_{N-1}$ on the elliptic curve with parameter $\mathfrak{q}$. These words could be expressed in terms of its transformation property:
\begin{equation}
    X \mathfrak{D}_1(\mathfrak{q} z) = -\mathfrak{D}_1(z) X C_{\mathfrak{q} z} 
\end{equation}
where:
\begin{equation}
    X = \textnormal{diag}(x_\omega)_{\omega = 0}^{\omega = N-1}
\end{equation}

\begin{equation}
    C_{z} = \begin{pmatrix} 0 & 1 & 0 &...& 0\\ 0 & 0 & 1 &...& 0 \\ \vdots & \vdots & \ddots & \ddots & \vdots \\ 0 & 0 & 0 & ... & 1 \\ z^{-1} & 0 & 0 & ... & 0\end{pmatrix} 
\end{equation}

Then the following matrix Jacobi identity holds:
\begin{eqnarray}
 \mathfrak{D}_1(z) =  \overleftarrow{\prod_{n=1}^\infty} \Big(1-  \hat{\slashed{Q}}^n  \, C_{z}^{-1} \Big) 
\cdot \overrightarrow{\prod_{n=0}^\infty} \Big(1 -   (\hat{\slashed{Q}} C_{z}^{-1})^n C_{z}^{\, n+1}\Big)   = \\ =
\sum_{n = 0}^\infty (-1)^n \hat{\mathbb{B}} \prod_{k = 0}^{n-1} (\hat{\slashed{Q}}^{n-k} C_{z}^{-1}) 
+ \sum_{n = 1}^\infty (-1)^n \hat{\mathbb{B}}\prod_{k=1}^{n-1} C_{z}^k (\hat{\slashed{Q}}C_{z}^{-1})^k \, C_{z}^n 
\end{eqnarray}
where:
\begin{gather}
    \mathfrak{q}_\omega = \frac{x_\omega}{x_{\omega-1}}, \qquad \omega = 1,..., N-1 \\
    \mathfrak{q}_0 = \mathfrak{q} \frac{x_0}{x_{N-1}}
\end{gather}
and extended to infinity by quasiperiodicity $x_{\omega+N} = \mathfrak{q} x_\omega$,
\begin{equation}
\hat{\slashed{Q}} = \text{diag} \big( \mathfrak{q}_{\omega} \big)_{\omega = 0 }^{\omega = N-1}    
\end{equation}

\begin{equation}
    \hat{\mathbb{B}} = \text{diag}(\mathbb{B}_\omega)_{\omega = 0}^{\omega = N-1}
\end{equation}
\begin{equation}
  \mathbb{B}_\omega = \prod_{l=1}^\infty \frac{1}{\Big( 1 -\frac{x_\omega}{x_{\omega-l}} \Big)}  
\end{equation}
The determinant of the matrix Jacobi identity gives back the usual one.

\end{document}